\documentclass{article}
\usepackage[english]{babel}
\usepackage[letterpaper,top=2cm,bottom=2cm,left=3cm,right=3cm,marginparwidth=1.75cm]{geometry}
\usepackage{amsmath,amsthm}
\usepackage{mathabx}
\usepackage{graphicx}
\usepackage[colorlinks=true, allcolors=blue]{hyperref}
\usepackage{algorithm2e}
\RestyleAlgo{boxed}
\LinesNumbered
\usepackage[mathscr]{euscript}
\usepackage{xspace}
\usepackage{makecell}
\usepackage{soul}

\usepackage{thmtools}
\usepackage{thm-restate}

\SetKwProg{Fn}{Function}{:}{}
\SetKwFunction{FDFSFix}{Fix}
\SetKwFunction{FInsert}{Insert}
\SetKwFunction{FRemove}{Remove}
\SetKwFunction{FDynamic}{Dynamic-LLL}
\SetKwInOut{KwEnsure}{Ensure}
\SetKwInOut{KwRequire}{Require}
\SetKwInOut{KwInput}{Input}
\SetKwInOut{KwOutput}{Output}

\usepackage[capitalize, nameinlink]{cleveref}
\usepackage[section]{placeins}
\usepackage{soul}
\crefname{theorem}{Theorem}{Theorems}
\Crefname{lemma}{Lemma}{Lemmas}
\Crefname{invariant}{Invariant}{Invariants}
\Crefname{claim}{Claim}{Claims}
\Crefname{observation}{Observation}{Observations}
\Crefname{algorithm}{Algorithm}{Algorithms}
\Crefname{figure}{Figure}{Figures}
\Crefname{condition}{Condition}{Conditions}

\newtheorem{theorem}{Theorem}[section]
\newtheorem{lemma}[theorem]{Lemma}

\newtheorem{definition}[theorem]{Definition}

\newtheorem{observation}[theorem]{Observation}
\newtheorem{claim}[theorem]{Claim}
\newtheorem{condition}{Condition}

\newtheorem{remark}{Remark}
\def\AIS{\text{AIS}}

\def\calF{\mathcal{F}}

\def\Omg{\Omega}
\def\vbl{\mathsf{vbl}}
\def\vio{\mathsf{violate}}
\def\eps{\varepsilon}
\def\blank{\mathsf{Blank}}
\def\tO{\tilde{O}}

\newcommand{\algRecolor}{\textsc{Recolor}\xspace}
\newcommand{\algCompare}{\textsc{Compare}\xspace}
\newcommand{\algResample}{\textsc{Resample-and-Find}\xspace}
\newcommand{\algCheck}{\textsc{Check-Flaw}\xspace}

\newcommand{\TCheck}{\mathcal{T}_\text{check}\xspace}
\newcommand{\TResample}{\mathcal{T}_\text{resample}\xspace}
\newcommand{\TCompare}{\mathcal{T}_\text{compare}\xspace}

\newcommand{\LLL}{Lovász Local Lemma\xspace}
\newcommand{\algOneShot}{\textsc{One-Shot-Coloring}\xspace}
\newcommand{\alp}{\alpha}

\newcommand{\E}{\mathbf{E}}

\newcommand{\prob}[1]{\Pr \left( #1 \right)}

\newcommand{\EE}[2][]{%
  \if\relax\detokenize{#1}\relax
    \E\left[#2\right]%
  \else
    \E_{#1}\left[#2\right]%
  \fi
}
\newcommand{\rb}[1]{\left(#1\right)}

\DeclareMathOperator{\poly}{poly}
\DeclareMathOperator{\polylog}{polylog}
\DeclareMathOperator{\Color}{Color}
\newcommand{\eqdef}{\stackrel{\text{\tiny\rm def}}{=}}
\def\InSet{\mathrm{In}}
\def\Inc{\text{In}} % Why do we have two commands for incoming??

\newcommand{\cC}{\mathcal{C}}
\newcommand{\cF}{\mathcal{F}}
\newcommand{\cL}{\mathcal{L}}

\newcommand{\Gfinal}{G_{\textsc{final}}}
\newcommand{\cCfinal}{\cC_{\textsc{final}}}
\newcommand{\Deltafinal}{\Delta_{\textsc{final}}}

\title{Dynamic Construction of the Lov\'asz Local Lemma}
\author{Bernhard Haeupler\thanks{
        INSAIT, Sofia University ``St.~Kliment Ohridski'' and ETH Zürich,
        \texttt{bernhard.haeupler@insait.ai}.
        Partially funded by the Ministry of Education and Science of Bulgaria's support for INSAIT as part of the Bulgarian National Roadmap for Research Infrastructure and through the European Research Council (ERC) under the European Union's Horizon 2020 research and innovation program (ERC grant agreement 949272).} \and 
Slobodan Mitrović\thanks{University of California, Davis. Emails: \texttt{\{smitrovic, sramach, wsheu\}@ucdavis.edu}. S. Mitrovi\'c was supported by Google Research Scholar. S. Mitrovi\'c, S. Ramachandran, and W.-H. Sheu were supported by NSF Faculty Early Career Development Program No.~2340048. Part of this work was done at INSAIT. Partially funded by the Ministry of Education and Science of Bulgaria's support for INSAIT as part of the Bulgarian National Roadmap for Research Infrastructure.} \and
Srikkanth Ramachandran\footnotemark[2] \and
Wen-Horng Sheu\footnotemark[2] \and
Robert Tarjan\thanks{Princeton University, \texttt{ret@princeton.edu}. Partially supported by a gift from Microsoft. Part of this work was done at INSAIT. Partially funded by the Ministry of Education and Science of Bulgaria's support for INSAIT as part of the Bulgarian National Roadmap for Research Infrastructure.}
}

\date{}

\begin{document}
\maketitle

\begin{abstract}
This paper proves that a wide class of
local search algorithms extend \emph{as is} to the fully dynamic setting with an adaptive (and even clairvoyant)
adversary, achieving an amortized $\tilde{O}(1)$ number of local-search steps per update.

\medskip

A breakthrough by Moser (2009) introduced the witness-tree and entropy compression techniques for analyzing local resampling processes for the Lovász Local Lemma. 
These methods have since been generalized and expanded to analyze a wide variety of local search algorithms that can efficiently find solutions to many important local constraint satisfaction problems. 
These algorithms either extend a partial valid assignment and backtrack by unassigning variables when constraints become violated, or they iteratively fix violated constraints by resampling their variables. 
These local resampling or backtracking procedures are incredibly flexible, practical, and simple to specify and implement. Yet, they can be shown to be extremely efficient on static instances, typically performing only (sub)-linear number of fixing steps.  
The main technical challenge lies in proving conditions that guarantee such rapid convergence.

\smallskip

This paper extends these convergence results to fully dynamic settings, where an adaptive adversary may add or remove constraints. 
We prove that applying the same simple local search procedures to fix old or newly introduced violations leads to a total number of resampling steps near-linear in the number of adversarial updates. 

\smallskip

Our result is very general and yields several immediate corollaries. 
For example, letting $\Delta$ denote the maximum degree, for a constant $\eps$ and $\Delta = \text{poly}(\log n)$, we can maintain a $(1+\eps) \Delta$-edge coloring in $\text{poly}(\log n)$ amortized update time against an adaptive adversary.
The prior work for this regime has exponential running time in $\sqrt{\log n}$ [Christiansen, SODA '26].
%we can maintain a $\Delta + \tilde{O}(\sqrt{\Delta})$-edge coloring in $(\Delta \log n)^{O(\log \Delta)}$ work per edge insertion 
Also, our result implies an $O(\frac{\Delta}{\ln \Delta})$ vertex coloring in a triangle-free dynamic graphs in $\tilde{O}(\Delta^3)$ update time.
Moreover, our approach yields a fully dynamic algorithm for a good local scheduling in the style of Leighton-Maggs-Rao under adversarial path additions and deletions.
%\stodo{Be a bit more explicit about what improvements we make.}
\end{abstract}

\newpage

\tableofcontents

\newpage
\section{Introduction}

The \LLL (LLL) is a powerful probabilistic, non-constructive tool that guarantees the existence of a wide variety of ``flawless'' combinatorial structures. 
These structures are described as simple objects satisfying several imposed constraints.
For example, the $k$-coloring of a graph $G$ can be described as a simple function $f : V(G) \rightarrow [k]$ that satisfies $m$ constraints, namely $f(u) \neq f(v)$ for every edge $e \in E(G)$. 
The \LLL is useful in scenarios where it is easy to satisfy exactly one of the constraints, but hard to determine whether all of them can be simultaneously satisfied. 
Indeed, in our example, it is extremely simple to find a coloring that does not violate only a single edge, but notoriously hard to determine if a valid coloring exists. 
This tool was invented in 1974 by Erdős and László~\cite{LLL74} to obtain bounds on the chromatic number of hypergraphs.
Since then, this tool has found a wide-range of applications, including in obtaining bounds on Ramsey numbers~\cite{spencer1977asymptotic}, graph colorings~\cite{LLL74, alon1991parallel}, satisfiability of $k$-CNF formulas~\cite{beck1991algorithmic}, analysis of communication networks~\cite{alon1989complexity}, path selection on expander graphs~\cite{broder1997static}, packet routing~\cite{Leighton1994,feige2002improved}, Latin transversals~\cite{erdos1991lopsided}, hypergraph partitioning~\cite{leighton2001new}, and even integer programming~\cite{srinivasan2006extension}.
Over time, many refinements and generalizations have been developed, such as the asymmetric and lopsided versions~\cite{spencer1977asymptotic,erdos1991lopsided}, which further broadened the scope of the lemma.
%\stodo{\cite{srinivasan2006extension} could be a quite interesting application to open up.}

The LLL is a foundational \emph{existential} result, guaranteeing that a collection of undesirable events can be simultaneously avoided under suitable probabilistic dependencies.
A large body of follow-up work has focused on the \emph{algorithmic} LLL, that is, developing efficient procedures to construct a flawless object whose existence is guaranteed by the lemma.
The first breakthrough was due to Beck~\cite{beck1991algorithmic}, who provided an algorithmic interpretation of the lemma for sparse $k$-CNF formulas by partially resampling violated clauses. 
Since then, numerous refinements and extensions were developed~\cite{alon1991parallel,molloy1998further,czumaj2000coloring,czumaj2000new,srinivasan2008improved}.
This progression culminated in the seminal Moser–Tardos (MT) framework~\cite{moser2009constructive,moser2010constructive}, whose remarkably elegant resampling algorithm gave the first fully constructive proof of the Lovász Local Lemma.
In particular, the MT framework consists of a simple procedure that fixes flaws one by one through resampling -- \emph{possibly introducing new flaws} in the process -- until all flaws are eliminated. 
% Since then, the entropy compression and witness tree techniques of Moser et al. have led to very general-purpose recipes proving that a wide variety of simple resampling and back-tracking processes terminate in (sub-)linear number of resamplings. 
%\stodo{Check whether someone else before proposed this resampling scheme.}
Given the wide applicability and importance of the LLL, we ask:
\begin{center}
    \emph{How efficiently can the algorithmic LLL be made dynamic?}
\end{center}
Obtaining fast dynamic algorithms is typically much harder than designing fast static ones.  
Even when small update times are achievable under an \emph{oblivious adversary} -- one that fixes the sequence of updates before the computation begins -- achieving comparable performance under an \emph{adaptive adversary} -- one that chooses each update based on the algorithm’s previous outputs -- is far more challenging.  
Moreover, there are now provable separations between the achievable complexities in these two models.  
For a collection of such examples, we refer the reader to \cref{sec:our-results} and to~\cite{beimel2022dynamic,ben2022framework,kaplan2021separating,bernstein2025separations} and references therein.
% \stodo{Make bibtex \href{https://arxiv.org/abs/2510.20341}{https://arxiv.org/abs/2510.20341}}\rtodo{done.}

Furthermore, in numerous applications, one algorithm invokes others as subroutines -- for example, a dynamic algorithm may rely on a (dynamic) edge-coloring procedure.  
Even if these subroutine invocations are not adversarial \emph{per se}, they are often generated \emph{adaptively}, as they depend on the current state and intermediate computation of the calling algorithm.  
In such cases, it is highly beneficial for the underlying subroutine to remain efficient under an adaptive adversary.  
Consequently, developing general tools such as the algorithmic LLL that are efficient even under adaptive updates can unlock a wide range of new dynamic applications.

% To give several illustrative examples, in the static setting, a proper $\Delta + 1$ vertex coloring can be computed in near-linear time by a simple greedy procedure.  
% In contrast, dynamic coloring remains much harder.  
% Bhattacharya et al.~\cite{bhattacharya2018new} gave a randomized algorithm maintaining a valid $\Delta + 1$ coloring in $O(\log \Delta)$ expected amortized update time, but their analysis holds only against an \emph{oblivious} adversary.
% The same work also provides a deterministic algorithm for maintaining a $(1 + o(1))\Delta$ coloring in $\poly \log(\Delta)$ amortized update time.
% More recently, Behnezhad et al.~\cite{behnezhad2025adaptive} obtained the first $\Delta + 1$ vertex coloring against an adaptive adversary with $\tO(n^{8/9})$ update time.  
% Subsequent work further improved this to $\tO(n^{2/3})$ under the same adaptive model~\cite{flin2025adaptive}.  

Perhaps surprisingly, we show that dynamic variants of the LLL can be maintained with only poly-logarithmic amortized changes per update\footnote{The actual update time is application-dependent due to the overhead of finding and fixing flaws.}, and even under a \emph{clairvoyant adversary}, which is a more powerful variant of adaptive adversaries that have access to the random bits used by the algorithm.
Specifically, under $q$ updates -- insertions and deletions -- the total number of resampling steps is $\tO_{\psi,\eps}(q + m + \log |\Omega|)$ even in the most general LLL setting currently known~\cite{achlioptas2019beyond}, where $m$ denotes the number of bits sufficient to represent any update sequence of length $q$, and $\Omega$ is the finite set of all objects.
% Note that $\log |\Omega|$ is the number of bits required to encode an object, and is often the output size of a search problem.
% Observe that the $O(m)$ term is needed for any algorithm that reads the entire the update sequence. Our bounds essentially are equivalent to the bounds in the static setting.
% \footnote{Because of the slack-ness in the constraints of the \LLL, it might be possible to avoid reading the entire sequence in certain applications.}
Here, we hide with the $\tilde{O}_{\psi,\eps}(\cdot)$ notation parameters of the LLL instance. In most typical applications of the LLL, these parameters are not significant. 
\footnote{See \cref{cond:general} for description of these parameters. For most LLL instances, the parameter $\psi$ upper bounds (up to constant factors) the probability that a randomly selected object has a given flaw. The parameter $\eps$ specifies that the instance satisfies the LLL condition with a slack of $\eps$. Similar parameterizations have been used in \cite{moser2010constructive,HaeuplerH17,achlioptas2019beyond}.}

\subsection{Our results}
\label{sec:our-results}

\paragraph{Dynamic \LLL.} 
The main contribution of our work is to formulate a \emph{dynamic} version of the \LLL and establish that the LLL tools can be used to obtain efficient dynamic algorithms for the \textit{strongest adversaries} in an almost black-box way.
Our dynamic formulation allows each flaw to be added or removed. When a flaw is added to removed, the set of objects that are considered flawless changes. An algorithm for the dynamic LLL is required to output a flawless object after \textit{every} update.
Our framework applies to any LLL instance that fits into the most general LLL setting currently known \cite{achlioptas2019beyond}.
This includes problem instances where the existence of a flawless object is guaranteed by the classic \LLL or its variants such as \textit{entropy compression} \cite{molloy2019list}.

\begin{theorem}[Informal version of \cref{thm:dyn-full}] \label{thm:informal}
    Consider an LLL instance $\cL$. Let $q$ be the number of updates to $\cL$ and $m$ be the number of bits needed to represent $q$ updates.
    Then, a flawless object for $\cL$ can be maintained dynamically under an adaptive adversary in total of $\tO_{\psi, \eps}(q + m + \log |\Omega|)$ flaw resamples with probability $1 - 2^{-q}$.
\end{theorem}

%% \btodo{I wanted to have a comparison with AIS's running time, but their refined time bound in Remark 2.4 can really avoid q resamplings in some extreme example (acyclic coloring)}
% In the static setting \cite{achlioptas2019beyond}, the number of resampling steps required to compute \emph{one} flawless object for $q$ flaws is upper bounded by $\tO_{\psi, \eps}(q + T)$, where $T = O(\log|\Omega|)$ is an additive factor depending on the instance.
% Our result matches this time bound up to a $\tO_{\psi, \eps}(1)$ factor.
% In many LLL settings, $m = \tO(q)$, and the number of resampling steps in \cref{thm:informal} matches, up to poly-logarithmic factors, the number of resamplings needed in the static setting for computing \textit{one} flawless object for $q$ flaws.
% In many LLL settings \footnote{\cite{achlioptas2019beyond} had a more refined time bound in Remark 2.4. However, for most LLL settings considered}, , applying the result for static LLL \cite{achlioptas2019beyond}[Theorem 2.4] 
\cref{thm:informal} should be compared with \cite{achlioptas2019beyond}'s result in the static setting.
\cite[Theorem 2.4]{achlioptas2019beyond}\footnote{This theorem is the main result of \cite{achlioptas2019beyond}. In the same paper, \cite{achlioptas2019beyond} also gave a more refined time bound in Remark 2.4, which provides improved running time for certain LLL instances. For most LLL settings considered, this refined bound still has an $\Omg(q)$ or $\Omg(\log |\Omega|)$ factor.} showed that finding \emph{one} flawless object with respect to a static set of $q$ flaws requires $\tO_{\psi, \eps}(q + T)$ resampling steps, where $T \geq \log_2 |\Omega|$ is a factor depending on the LLL instance.
When $m = \tO(q)$, \cref{thm:informal} achieves the same upper bound for maintaining a flawless object \emph{dynamically} over $q$ updates.

In many LLL settings, $\Omega$ is a product space on $n$ variables, each flaw is defined on at most $k$ variables, each flaw shares variables with at most $d$ other flaws, $m$ is either $\Theta(q)$ or $\Theta(qk)$, and each resampling step can be implemented in $\poly(dk)$ time.
In such settings, the above theorem implies an $\tO(\poly(dk))$ amortized update time.
This complexity depends only on the number of variables involved in a flaw and all its neighboring flaws, rather than on the entire input size.

% \cref{thm:informal} implies new dynamic algorithms for several combinatorial problems.
This paper also presents details of several dynamic algorithms that are obtained by our framework based on different variants of the LLL.
The update time for each specific application is summarized as follows.

% Hence, \cref{thm:application}(A1) gives an $\tO(\poly(k, D))$ amortized update time with an exponential tail.

% \begin{theorem}[Informal]
%     Consider an
% \end{theorem}

\paragraph{$k$-CNF (\cref{sec:CNF}).}
Given a CNF in which each clause has exactly $k$ distinct variables, the $k$-CNF problem asks whether the formula is satisfiable.
This problem is a standard application of the \LLL.
Define the \emph{dependency degree} $d$ as the maximum number of clauses that a clause shares variables with.
The \LLL implies that if $d < 2^k / e - 1$, then the formula is satisfiable;
in addition, \cite{moser2010constructive}'s framework can find such an assignment in polynomial time.
By applying our framework, we extend this result to the dynamic setting, where a clairvoyant adversary can add or remove clauses, and $d < (1 - \eps) \cdot (2^k / e - 1)$ is satisfied at any moment for some constant $\eps$.
The amortized update time of our algorithm is polynomial in $dk$.
% For this problem, $\Omega$ is a product space of $n$ binary variables, each flaw (clause) is defined on $k$ variables, and each flaw shares variables with at most $d$ other flaws.

\begin{theorem}[Dynamic $k$-CNF]
    There is an algorithm that maintains a satisfying assignment of a $k$-CNF formula with $d < (1 - \eps) \cdot (2^k / e - 1)$ under insertions and deletions of clauses, using $\tO(dk^2)$ amortized update time with probability at least $1 - 2^{-q}$ against a clairvoyant adversary.
\end{theorem}

\paragraph{Triangle-free coloring (\cref{sec:triangle-free}).}
% Let $\Delta$ be the maximum degree.
% We achieve an $\tO(\Delta^3)$ amortized update time.
Let $G = (V, E)$ be a graph with maximum degree $\Delta$.
The \emph{list-chromatic number} of $G$ is the smallest positive integer $\ell$ satisfying the following:
For any assignment of color lists of size $\ell$ to each vertex, it is possible to obtain a proper coloring by giving each vertex a color from its list.
It is not hard to show that the list-chromatic number of any graph is $\Delta + 1$.
On triangle-free graphs, a series of works \cite{borodin1977upper,catlin1978bound,lawrence1978coveing,johansson1996asymptotic,pettie2015distributed,molloy2019list} improved this bound to $O(\frac{\Delta}{\ln \Delta})$, which is asymptotically optimal.
We provide a brief review in \cref{sec:related-work}.
The state-of-the-art result \cite{molloy2019list} is proven by an entropy compression argument, which is a variant of algorithmic LLL.
% Their proof is constructive, implying an algorithm that finds an $O(\frac{\Delta}{\ln \Delta})$-list-coloring in polynomial time.
Our framework is applicable to this variant, implying a fully dynamic algorithm that maintains an $O(\frac{\Delta}{\ln \Delta})$-list-coloring under edge updates.
To the best of our knowledge, no previous dynamic algorithm for this problem is known.
% We extend this result to the fully dynamic setting, obtaining an algorithm that maintains an $O(\frac{\Delta}{\ln \Delta})$-list-coloring under edge updates.
% Our algorithm applies to all ranges of $\Delta$.

\begin{theorem}[Triangle free coloring]
    An $O(\frac{\Delta}{\ln \Delta})$-list-coloring can be maintained dynamically on a triangle-free graph while supporting edge insertion and deletion updates in $\tO(\Delta^3)$ amortized update time with probability at least $1 - 2^{-q}$ against a clairvoyant adversary.
\end{theorem}

\paragraph{LMR scheduling (\cref{sec:LMR-scheduling}).}
The Leighton–Maggs–Rao (LMR) \cite{Leighton1994} scheduling result addresses the problem of routing packets through a network.
The network is modeled as a directed graph $G=(V,E)$, where vertices represent processors and edges represent communication links.
The input is a collection of paths, each corresponding to a packet that must be routed from a source vertex to a destination vertex along a specified path in the network.
Two critical parameters of this scenario are:
\begin{itemize}
    \item  Congestion ($C$): The maximum number of paths that traverse any single edge in the network.
    \item  Dilation ($D$): The maximum length of any of these paths.
\end{itemize}

The problem is to determine an efficient schedule that delivers all packets along their prescribed routes while minimizing the total completion time, assuming that each edge can transmit at most one packet per time step.
Clearly, both $C$ and $D$ are lower bounds on the completion time. Leighton, Maggs, and Rao \cite{Leighton1994} proved that any set of paths with congestion $C$ and dilation $D$ can be scheduled in $O(C+D)$ time steps.
The proof of this result is rather intricate.
Later, Rothvoss~\cite{rothvoss2013simpler} presented a greatly simpflied proof.
% Their proof can be made constructive by \cite{moser2009constructive}'s framework.

The result of \cite{rothvoss2013simpler} is obtained by composing several levels of LLL instances.
Each individual level fits into the algorithmic LLL of \cite{moser2009constructive}, but the set of flaws is defined based on a flawless object of previous levels.
Our framework is applicable to such composition of LLL instances, implying a dynamic algorithm for the scheduling problem.
The amortized update time of our algorithm is $\tO(1)$ when $C$ and $D$ are constants.
To the best of our knowledge, no previous dynamic algorithm for this problem is known.
% Their result is obtained by composing several \emph{levels} of LLL instances.
% In their scheme, the universe $\Omega$ is a product space on $k \cdot \poly(C, D)$ variables, where $k$ is the number of paths.
% The variables are partitioned into $L$ sets $V_1 \cup V_2 \cup \dots \cup V_L$, one for each level.
% Each level $\ell \in [1, L]$ is associated with a flaw set.
% Their construction iterates from level $1$ to $L$, and for each level $\ell$, applies the LLL to find a flawless assignment for $V_\ell$.

% The arguments on how to get to the $O(C+D)$ result are rather involved.
% Later, a drastic simplification by Rothvoss~\cite{rothvoss2013simpler} gives a beautiful proof that fits into the Moser-Tardos framework.
% We extend \cite{rothvoss2013simpler}'s result to the dynamic setting.
% The amortized update time of our algorithm is $\tO(1)$ when $C$ and $D$ are constants.

\begin{theorem}[Leighton-Maggs-Rao scheduling]
    An asymptotically optimal schedule can be maintained over a dynamically changing set of paths with dilation at most $D$ and congestion at most $C$ within $(C + D + \log(|E| + q))^{O(\log \log (C + D))}$ amortized update time with probability $1 - 2^{-q}$ against a clairvoyant adversary.
\end{theorem}

\paragraph{$(1+\eps) \Delta$-edge coloring (\cref{sec:edge-coloring}).}
Given a graph with maximum degree $\Delta$, the $(1+\eps)\Delta$-edge-coloring problem asks for an assignment of colors ${1, 2, \ldots, (1+\eps)\Delta}$ to the edges so that no two incident edges share the same color.
The celebrated theorem of Vizing~\cite{vizing1964estimate} guarantees that every graph can be edge-colored with at most $\Delta + 1$ colors.
However, in many algorithmic settings, obtaining a \emph{near-optimal} $(1+\eps)\Delta$-edge coloring is considerably simpler than achieving the exact optimum of $\Delta + 1$ colors.
A long line of work focused on designing efficient algorithms for this problem, in both static and dynamic setting, e.g., see ~\cite{duan2019dynamic,solomon2020improved,christiansen2023power,bhattacharya2024nibbling,bernshteyn2024linear,elkin2024deterministic,bhattacharya2024arboricity,christiansen2024sparsity,assadi2025faster,dhawan2025fast,sadeh2026beyond,christiansen2026deterministic} and references therein.
For a constant $\eps$ and $\Delta = \poly(\log n)$, we obtain a dynamic algorithm, against an adaptive adversary, with the amortized update time of only $\poly(\log n)$.
The prior work for this regime has exponential running time in $\sqrt{\log n}$ \cite{christiansen2026deterministic}.
See \cref{sec:related-work} for a review of prior results.

\begin{theorem}
\label{thm:edge-coloring}
    Let $G$ be a graph on $n$ vertices with maximum degree $\Delta$.
    Given a parameter $\eps \in \omega\rb{\frac{\log^{2.5} \Delta}{\sqrt{\Delta}}}$, there exists a fully dynamic randomized algorithm that under $q$ updates maintains $(1+\eps) \Delta$-edge coloring on $G$ in $(\Delta \cdot (\log q + \log n))^{O(\log 1/\eps)}$ amortized update time, even against an adaptive adversary.
\end{theorem}

%\paragraph{Comparison with prior work.}
% To the best of our knowledge, all prior work on the algorithmic LLL is restricted to the static setting.
% We initiate the study of the dynamic LLL and showed that, even against an adaptive adversary, one can obtain resampling algorithms whose overall time complexity matches the static setting up to polylogarithmic factors.
% Our result is very general and immediately yields new dynamic algorithms for several problems solvable by the algorithmic LLL.

% For the problems of $k$-CNF under the LLL constraints, triangle-free coloring, and LMR scheduling, we are not aware of any previous dynamic algorithms.
% Hence, our framework implies the first dynamic algorithms with polylogarithmic amortized update time when the key problem parameters (e.g., $k$ and $d$ in the CNF problem) are constants.

\def\cR{\mathcal{R}}
\def\cA{\mathcal{A}}

\subsection{Related work} \label{sec:related-work}
Our dynamic LLL framework builds on the work of \cite{achlioptas2019beyond,moser2010constructive}. Moser and Tardos \cite{moser2010constructive}'s breakthrough proof-strategy for analyzing resampling algorithms inspired a flurry of results in the field, e.g., see \cite{HaeuplerH17} and references therein. Among those, the work of Achlioptas, Iliopoulos, and Sinclair \cite{achlioptas2019beyond} provides a very general framework by which almost all conceivable resampling algorithms for different applications of the LLL can be analyzed. 

\paragraph{Triangle-free coloring.}
The \textit{chromatic number} of $G$ is the minimum number of colors required to properly color $G$.
Clearly, the chromatic number is always smaller than the list chromatic number of the same graph.
It is not hard to prove that the chromatic number of any graph is at most $\Delta + 1$.
For odd cycles and cliques, this upper bound is the best possible.
Brook's famous theorem \cite{brooks1941colouring} further showed that every other connected graph admits a $\Delta$-coloring.
Vizing \cite{vizing1968some} asked whether the theorem can be improved on triangle-free graphs. 
A series of works \cite{borodin1977upper,catlin1978bound,lawrence1978coveing} answer the question affirmatively, resulting in an upper bound of $\frac{2}{3}(\Delta + 2)$ on triangle-free graphs.
This result is later improved by Johansson \cite{johansson1996asymptotic}, who proved that the list-chromatic number of any triangle-free graph is at most $\frac{9\Delta}{\ln \Delta}$ for all sufficiently large $\Delta$.
The leading constant was improved to 4 by Pettie and Su \cite{pettie2015distributed}, and later to $1+o(1)$ by Molloy \cite{molloy2019list}.
This bound is optimal up to a factor of $2+o(1)$ \cite{molloy2019list}.

\paragraph{$(1+\eps) \Delta$-edge coloring (\cref{sec:edge-coloring}).}
% For instance, in the centralized setting, the classical greedy approach outputs $2 \Delta - 1$ coloring; this is the case as an edge is incident to at most $2 \Delta - 2$ other edges.
The problem of dynamically maintaining near-optimal edge colorings has been studied extensively in recent years.
See \cref{tab:coloring} for an overview.
Barenboim and Maimon~\cite{barenboim2017fully} gave a deterministic fully dynamic algorithm achieving an $O(\Delta)$-edge coloring with $\tO(\sqrt{\Delta})$ worst-case update time.
Duan, He, and Zhang~\cite{duan2019dynamic} obtained a randomized algorithm that maintains a $(1+\eps)\Delta$-edge coloring in $O(\log^7 n / \eps^2)$ amortized update time. The algorithm assumes that $\Delta = \Omega(\log^2 n / \eps^2)$ and $\eps$ is a fixed constant.
Christiansen~\cite{christiansen2023power} later achieved a $(1+\eps)\Delta$-edge coloring in $O(\log^9 n \log^6 \Delta / \eps^6)$ time per update without restrictions on $\Delta$.
More recently, Bhattacharya, Costa, Panski, and Solomon~\cite{bhattacharya2024nibbling} gave a randomized algorithm with $O(\poly(1/\eps))$ amortized update time for $(1+\eps)\Delta$-edge coloring, provided that $\Delta \ge (\log n / \eps^4)^{\Theta((1/\eps)\log(1/\eps))}$.
The algorithms of \cite{duan2019dynamic,christiansen2023power,bhattacharya2024nibbling} only work against an oblivious adversary.
% The algorithms of \cite{duan2019dynamic,christiansen2023power,bhattacharya2024nibbling} achieve a $\poly(\log n, 1/\eps)$ or even $\poly(1/\eps)$ update time, but they only work against an oblivious adversary.
Finally, Christiansen~\cite{christiansen2026deterministic} obtained a fully dynamic \emph{deterministic} algorithm maintaining a $(1+\eps)\Delta$-edge coloring with amortized update time $\exp\rb{\tilde O\rb{\sqrt{\log n} \cdot \polylog 1/\eps}}$, which is $n^{o(1)}$ when $\eps^{-1} = 2^{O(\log^{0.49} n)}$. Our algorithm is faster when $\Delta = e^{o(\sqrt{\log n})}$ and $\eps = \omega\rb{\frac{\log^{2.5} \Delta}{\sqrt{\Delta}}}$.
 % The little-o here cannot be taken out because I don't know the constants in both algorithms.
%\cite{barenboim2017fully}'s algorithm has a $\tO(\sqrt{\Delta})$ update time, but maintains only an $O(\Delta)$-edge coloring.

% We obtain a new dynamic algorithm by extending the static (distributed) algorithm of \cite{chang2018complexity} to the dynamic setting.
% For constant $\eps$ and $\Delta = \poly(\log n)$, the amortized update time of our approach is only $\poly(\log n)$.
\begin{table}[ht]
\centering
\begin{tabular}{|l|c|c|c|c|c|}
\hline
\textbf{Ref.} & \textbf{Update time} & \makecell{\textbf{Lower bound}\\\textbf{on $\eps$}} & \makecell{\textbf{Lower bound}\\\textbf{on $\Delta$}} & \textbf{Det.} & \textbf{Adversary} \\
\hline
% \cite{barenboim2017fully}\ $\dagger \dagger$ & $\tO\rb{\sqrt{\Delta}}$ & -- & -- & Y & Adaptive \\
% \hline
\cite{duan2019dynamic} & $O\rb{\log^7(n)/\eps^2}$ & constant $\eps$ & $ \Omega(\log^2 n / \eps^2)$ & N & Oblivious \\
\hline
\cite{christiansen2023power} & $O\rb{\log^9 n \log^6 \Delta / \eps^6}$ & -- & -- & N & Oblivious \\
\hline
\cite{bhattacharya2024nibbling} & $O(\poly(1/\eps))$ & any $\eps > 0$ & $\rb{\tfrac{\log n}{\eps^4}}^{\Theta\rb{\frac{1}{\eps}\log 1/\eps}}$ & N & Oblivious \\
\hline
\cite{christiansen2026deterministic} & $\exp\rb{\tilde{O}\rb{\sqrt{\log n} \cdot \polylog 1/\eps}}$  & -- & -- & Y & Adaptive \\
\hline
\textbf{here} & $(\Delta \cdot \log n)^{O(\log 1/\eps)}$ & $\omega\rb{\frac{\log^{2.5} \Delta}{\sqrt{\Delta}}}$ &  $\Delta_\eps$ & N & Adaptive \\
\hline
\end{tabular}
\caption{ \label{tab:coloring} Dynamic $(1+\eps)\Delta$-edge-coloring algorithms, in which we use $\Delta_\eps$ to denote a constant depending only on $\eps$.
%All the results except \cite{barenboim2017fully} guarantee amortized update time.
All the results guarantee amortized update time.
%$^{\dagger \dagger}$This result gives an $O(\Delta)$-edge coloring, not a $(1+\eps)\Delta$-edge coloring.
}
\end{table}

\paragraph{Remark on concurrent work.}
Parallel and independent to our work, Assadi and Yazdanyar \cite{assadi2026fully} presented a fully dynamic algorithm that maintains an $O(\frac{\Delta}{\ln \Delta})$-coloring on triangle-free graphs within $\Delta^{o(1)}\log n$ amortized update time. 
Their dynamic algorithm is based on entropy compression arguments and "recording the adversaries action" when specifically applied to the triangle-free coloring problem, while our framework is designed for very general algorithmic LLL and entropy compression settings. Their quantitatively better update time of \cite{assadi2026fully} is achieved by using a technique of \cite{bhattacharya2018dynamic} which recursively splits degrees and colors. Another independent concurrent work by Haeupler and Paramonov~\cite{haeupler2026maintainingrandomassignmentsadversarial} shows how to transform any linear (or sub-linear) static coloring algorithm in a black-box way into a fully-dynamic vertex coloring algorithm with $O(\Delta \cdot n^{\eps})$ update time against an adaptive adversary with at the cost of increasing the number of colors used by a constant $O_\eps(1)$. The techniques of \cite{haeupler2026maintainingrandomassignmentsadversarial} are completely different and disjoint from this work. 

\subsection{Organization}

\cref{sec:prelims} reviews the search problems covered by the LLL and algorithms that solve these problems efficiently in the static setting.
\cref{sec:dynamic-intro} formally defines the dynamic version of LLL and states the resampling algorithms that maintain a solution under flaw updates.
\cref{sec:moser-tardos-convergence,sec:gen-framework-convergence} provide proofs of our bounds on the number of resamplings of the dynamic algorithms.
In \cref{sec:applications}, we establish a framework by which to analyze the running times of the dynamic algorithms that we develop. Specifically, we (i) abstract out broad implementation-specific routines and express the running times in terms of these routines, and (ii) analyze the book-keeping process by which we maintain and choose flaws in our applications. 
\cref{sec:CNF,sec:triangle-free,sec:LMR-scheduling,sec:edge-coloring} outline our applications of the dynamic LLL to various problems.
We conclude this paper in \cref{sec:future-work} with a few directions for future research.

% \section{A general framework and scheme for analyzing Local Resampling algorithms}
\section{Preliminaries}
\label{sec:prelims}

In this section, we give an overview of the \LLL, and the different local search procedures in literature that constructively prove it, i.e., explicitly construct the objects that the LLL guarantees to exist.

In \cref{sec:local-search}, we define search problems that are covered by the LLL. In \cref{sec:search-algorithms} we describe algorithms for these search problems in the static setting. 
When the search problems satisfy certain additional criteria (see \cref{sec:condition}), these search algorithms are known to converge fast into a valid solution due to a seminal work of \cite{moser2010constructive}.
% We give an outline of the method in \cref{sec:dynamic-intro}.
The original proof of \cite{moser2010constructive} was restricted to the setting of product distributions and variable resamplings.
\cite{achlioptas2019beyond} extended the analysis to more general resampling algorithms that cover almost all resampling algorithms for LLL-related problems known to date.
We refer to their approach as the AIS framework.
% We follow the same notation as \cite{moser2010constructive,achlioptas2019beyond}.

% First we define the general setup of search problems and resampling algorithms that fall within the framework of the \LLL. 

\subsection{Local search problems}
\label{sec:local-search}
% \paragraph{States and flaws.}
%this is copied in from the achlioptas paper:
% We largely follow \cite{achlioptas2019beyond} with the following notation (but mix in some MT terminology as well):
% We follow the same notation as \cite{achlioptas2019beyond}, but first the initial setup

A search problem is specified by a finite set of states $\Omega$ and by a collection of \emph{flaws} $\mathcal{F} = \{F_1, F_2, \ldots, F_m \}$, where each flaw $F_i$ is a subset of $\Omega$.
The goal of the search problem is to find a flawless state, i.e., a state in $\Omega\ \setminus\ \Omega^*$ where $\Omega^* = \bigcup_{i \in [m]} F_i$ is the set of all flawed states.
For a state $\sigma$, we denote by $U(\sigma) = \{ j \in [m] : F_j \ni \sigma \}$ the set of (indices of) flaws present in~$\sigma$.
Examples of such search problems include the following.
\begin{itemize}
    \item CNF satisfiability (see \cref{sec:CNF} for formal definitions): Let $\Omega = \{0,1\}^n$ be the set of truth assignments to $n$ Boolean variables, and for each clause of some CNF formula, let the flaw $F_i$ be the assignments that violate the $i^{th}$ clause. Any flawless state then corresponds to a satisfying assignment. 
    \item Graph coloring (see \cref{sec:triangle-free} for formal definitions): Let $\Omega = \{[C] \cup \{\blank\}\}^n$ be all partial vertex colorings with $C$ colors in some $n$-vertex graph $G$. Flaws can, for example, be defined for each neighborhood of a vertex or each path or cycle in $G$. Such flaws can specify a large variety of forbidden colorings, such as no cycle being bi-chromatic and/or no adjacent vertices having the same color, etc. One can also have a flaw for each vertex forbidding it from being colored blank. A flawless state then corresponds to a complete and proper vertex coloring without bi-chromatic cycles. 
%    \item Latin squares: Let $\Omega = \{S_n\}^n$ be $n$ permutations of $[n]$. A flaw ...
\end{itemize}

\subsection{Local search algorithms}
\label{sec:search-algorithms}
% \paragraph{Resamplings and flaw selection rules.}

Given a search problem specified by $\Omega$ and $\calF$, we analyze (local) \emph{search algorithms} for finding a flawless state in $\Omega$. The search algorithms start with an arbitrary (or random) state $\sigma \in \Omega$. The algorithm then iteratively tries to identify a \textit{flawless} state, by constructing intermediate states that are free from a subset of flaws.
The search algorithms that we analyze are specified by the following:
\begin{itemize}
    \item[(i)](Initial State). A starting state $\sigma_0 \in \Omega$ or distribution $\theta$ over starting states $\sigma_0 \sim \theta$. Typically, $\sigma_0$ is an \textit{easy to find} object in $\Omega$. 
    \item[(ii)](Flaw resampling procedure). For every flaw $F_i$, there is a randomized procedure $\rho_i$ that is invoked to ``fix'' flaw $F_i$. Upon invoking $\rho_i$, the algorithm transitions from the current state $\sigma$ into another state $\tau$ computed by $\rho_i$. We say that the flaw $F_i$ is \textit{resampled}.

    % \btodo{TO DISCUSS: do we use/derive non-deterministic selection rule somewhere?}
    % The procedure $\rho_i$ is allowed to depend in an arbitrary way on the current state $\sigma \in F_i$ and is specified through a state transition distribution $\rho_i^{\sigma}$ where $\rho_i^{\sigma}(\tau)$ is the probability that the algorithm \emph{transitions} to state $\tau$ when fixing flaw $F_i$ in state $\sigma$.
    \item[(iii)] (Flaw selection procedure). Given an iteration $j$ of the algorithm, and the state $\sigma \in \Omega$ during this iteration, the flaw selection procedure outputs a flaw $F_i \ni \sigma$. 
    % The flaw selection procedure can otherwise be arbitrary.
    
    % Although we allow the flaw selection to be quite generic (even possibly random), for the applications in this paper, we will only require priority-based flaw selection procedures. The priorities are specified by a permutation $\pi$ over the flaws, and the algorithm selects the flaw that has the highest rank according to $\pi$.
\end{itemize}
% \rtodo{Only describe the generic flaw selection above, and then have a separate paragraph describing the specific flaw selection rules. Also see if the algorithms can be described here.}

% \paragraph{Moser-Tardos resamplings.} In the Moser-Tardos setting, the set of states is a Cartesian product over several variables. The resampling procedures to fix a flaw $F_i$ simply uniformly resample all the variables that determine whether flaw $F_i$ exists. For this particular setting, the flaw selection procedure can be \textit{arbitrary}. For the general framework, there are additional restrictions on the flaw selection procedure that we will outline later.

The algorithm then operates as follows: Initially, the starting state $\sigma_0$ is sampled from (i), the specified initial distribution $\theta$. Subsequently, while the current state $\sigma$ is not flawless, the flaw selection procedure in (iii) is invoked to select a flaw $F_i \ni \sigma$.
Finally, (ii): the resampling procedure is invoked with flaw $F_i$ and state $\sigma$ that yields a new state $\sigma'$ according to the resampling distribution $\rho^\sigma_i$.
The algorithm terminates when the a flawless state is reached.

Assuming a deterministic flaw selection rule based only on the current state $\sigma$ or the currently existing flaws $\calF(\sigma_{t-1})$, a resampling algorithm fully specifies a Markov chain $M$ that starts in $\sigma_0$ and then for every time $t$ chooses one of the existing flaws $F_i \in \calF(\sigma_{t-1})$ based on the flaw selection rule and then transitions to the next state $\sigma_t$ based on the distribution $\rho_i^{\sigma_{t-1}}$. 

We note that the flaw selection rule can be randomized, for e.g., one may sample a random permutation of flaws $\pi_j$ for each iteration $j$ before executing the algorithm and then given state $\sigma$, the flaw selection rule can choose the flaw $f \ni \sigma$ with highest priority according to $\pi_j$. All of our applications only use a deterministic flaw selection rule.
% until at time $t^*$ the state $\sigma_{t^*}$ is flawless and the algorithm terminates.

% \begin{remark}
% % In many cases, one also specifies some (deterministic) flaw selection rule that specifies which flaw to resample next.
% % Often, this flaw selection rule depends only on the current set of flaws $F(\sigma)$ and the sequence of flaws fixed so far.
% In most cases, the flaw selection rule depends only on the current set of flaws $F(\sigma)$ and the sequence of flaws fixed so far.
% However, in some cases, an arbitrary and even non-deterministic adversarial choice of which flaw to resample can also be analyzed.
% We explicitly mention where this is the case.
%  %Typical flaw selection rules either pick the flaw with the smallest index among the current flaws $F(\sigma)$, or the smallest index flaw in the "neighborhood" of the last flaw fixed - essentially performing a depth first search flaw fixing strategy. 
% \end{remark}

\paragraph{Flaw introductions and flaw neighborhoods.}%, Primary Flaws}

We say that a resampling of flaw $F_i$ that transitions from state $\sigma$ to state $\tau$ \emph{introduces} flaw $F_j $ if either (i) $F_j$ was created due to the resampling, i.e., $F_j \in U(\tau) - U(\sigma)$ or (ii) $j = i$ and $F_i$ remains a flaw even after resampling.
Intuitively, we consider that resampling a flaw $F_i$ fixes it and if it remains a flaw after resampling, then we say that $F_i$ \textit{introduced itself}.
That is, if $\cF$ is the current set of flaws then upon resampling $F_i$, the set of flaws becomes $(\cF - F_i) \cup C$ where $C$ is the set of flaws introduced.

%In the later case the resampling of flaw $F_i$ essentially reintroduced itself again. 
For each flaw $F_i$, we define the neighborhood $\Gamma_i \subseteq \calF$ to be the set of (indices of) flaws that can (possibly) be introduced by some resampling of $F_i$.
Note that the neighborhood relationships may be a digraph: For $F_i,F_j \in \cF$, $j \in \Gamma_i$ may not imply $i \in \Gamma_j$.
% {\color{red} We assume that $i \in \Gamma_i$.}
% \btodo{This assumption is for the derivation of Condition 2. Maybe change this notation (like $\Gamma^+$) because we are not really making any assumptions on structure of dependency graphs.}
%A primary flaw is a flaw that cannot be collaterally fixed - where does this come in???

\subsection{A common setting: product state space and MT-resampling} 
\label{sec:mt-setup}

A common and particularly simple setting is as follows.

\begin{definition}[The variable setting] \label{def:variable}
In the variable setting, the state space is an $n$-dimensional product space, i.e., $\Omega \subseteq \bigtimes_{i=1}^n [c_i]$, where $[c_i]$ denotes $\{1, 2, \dots, c_i\}$.
Each state corresponds to an assignment to $n$ variables $v_1, v_2,\ldots, v_n$, where variable $v_i$ takes on $c_i$ different values.
% and defines a distribution over these values (often the uniform distribution).
A flaw corresponds to a subset of assignments.
\end{definition}

We denote by $\vbl(F_i)$ the minimal set of all variables that fully determine whether flaw $F_i$ is present in a given state $\sigma \in \Omega$, i.e., $\vbl(F) = \{j\ \mid \exists (\sigma \in F, \sigma' \not\in F) \text{ such that } \sigma, \sigma' \text{ differ only in } v_j\}$.
In the variable setting, the neighborhood $\Gamma_i$ for flaw $F_i$ only contains flaws that share variables with $F_i$, i.e., $\Gamma_i = \{j \mid \vbl(F_i) \cap \vbl(F_j) \neq \emptyset\}$.

Let $F$ be a flaw and W.L.O.G. assume that $\vbl(F) = \{v_1, v_2, \dots, v_k\}$.
Denote by $\vio(F)$ the number of different assignments for $v_1, v_2, \dots, v_k$ that make $F$ present.
Define $\Pr[F]$ as $\frac{\vio(F)}{\Pi_{i=1}^k c_{i}}$.
That is, $\Pr[F]$ is the probability that a randomly selected state has $F$.

The most common local search algorithms for the variable setting satisfy the following:
\begin{itemize}
    \item The distribution $\theta$ of starting states is uniform.
    \item The resampling procedure for a flaw $F_i$ is to transition to a state that only changes the variables in $\vbl(F_i)$ and assigns these variables new independent random values (according to their distributions). We call this an \emph{MT-resampling}.
    \item The flaw selection rule can be arbitrary.
\end{itemize}

%For example, if a resampling for $F_i$ in a product space typically only changes variables in $\vbl(F_i)$ and therefore $\Gamma_i$ only contains flaws $F_j$ for which $\vbl(F_i) \cup \vbl(F_j) \neq \emptyset$.

%For example, for a given CNF formula on $n$ variables with clauses $c_1,\ldots,c_m$, we take $\Omega=\{0,1\}^n$ to be the set of all possible variable assignments, and $F_i$ the set of assignments that fail to satisfy clause~$c_i$.

\subsection{AIS Framework}
\label{sec:general-setting-setup}

% In this section we overview the general setting of \cite{achlioptas2019beyond}.

% One common approach to analyzing a local search algorithm is to argue that it terminates after only a few resamplings with very high probability.
In this section we overview the AIS framework \cite{achlioptas2019beyond}. Their framework allows for a black-box analysis of almost all conceivable resampling algorithms for local search algorithms. In particular, the algorithm can select flaws and execute resampling procedures that go beyond variable resampling.

% For this purpose, we want to show that the flaws are unlikely to occur, and the resampling procedures addressing them do not introduce too many other flaws on average.
To incorporate general flaw resampling schemes, they introduce the notion of flaw \textit{charges} formally defined below. Intuitively, the charge for a flaw $F_i$ is an upper bound on the probability that the current (flawed) state was due to the resampling of $F_i$.
% $\gamma_i$, which intuitively are an upper bound on the probability of how likely it is that resampling $F_i$ is necessary.

For the case where $\Omega$ is a product state space of variables with $\mu$ being the product distribution that samples every variable independently and all resampling procedures are  MT-resamplings, the charge $\gamma_i$ can be set equal to $p_i = \mu(F_i)$, i.e., the probability of a uniformly random variable assignment to contain flaw $F_i$. 

For general resampling steps, it is useful to define $\InSet_i(\tau)$ to be the set of states from which it is possible that a resampling of flaw $F_i$ leads to state $\tau$. Then $\gamma_i$ is simply the union bound over the probabilities of transitioning into state $\tau$, in the worst-case, i.e., 
$\gamma_i = \max_{\tau \in \Omega} \left\{\sum_{ \sigma \in \InSet_i(\tau)}  \rho_i(\sigma,\tau )\right\}$. 

% Intuitively, if there are not many likely ways to get to any given state via a $F_i$-resampling, then it seems also less likely that many such resamplings should occur during a probabilistic local search. 

% To give even more precise bounds and allow for sharper convergence criteria in several important applications, one can further generalize this as follows:
To allow for sharper convergence criteria in several  applications, we extend the charge definition as follows.
For an arbitrary state $\tau \in \Omega$, flaw $F_i$, and set of flaws $S$, we denote by $\InSet_i^S(\tau)$ the set of states $\sigma$ from which it is possible that a resampling of flaw $F_i$ leads to state $\tau$ while also introducing new flaws that include the flaws in $S$, i.e., 

\begin{align}
\InSet_i^S(\tau)  :=  & \{ \sigma \in F_i : \text{the set of flaws introduced by the transition $\sigma \to \tau$ includes $S$ and} \nonumber \\
& \text{$F_i$ is chosen to be resampled at $\sigma$ } \}.\label{def:Incoming}
\end{align}

\begin{definition}[Generalized charges]
\label{defn:charges}
For any fixed probability distribution $\mu>0$ on~$\Omega$ of our choosing, we define the \emph{generalized charge} as: 
\begin{align}\label{def:charge}
\gamma_i^S := \max_{\tau \in \Omega} \left\{\frac{1}{\mu(\tau)}\sum_{ \sigma \in \InSet_i^S(\tau)} \mu(\sigma) \rho_i(\sigma,\tau )\right\} .
\end{align} 
\end{definition}

% \rtodo{ Explain how these conditions reduce to the standard setting.  Explain monotonicity.}

One can think of $\gamma_i$ as an upper bound on the probability or how unlikely it is that fixing $F_i$ is necessary and think of $\gamma_i^S$ as the probability that fixing $F_i$ is necessary \emph{and} this resampling is also responsible for causing flaws from $S$ to have to be resampled later. 

%In many cases one chooses $\mu$ to be the uniform distribution on $\Omega$ in which case this simplifies to $\gamma_i^S := \max_{\tau \in \Omega} \left\{\sum_{ \sigma \in \InSet_i^S(\tau)} \rho_i(\sigma,\tau )\right\}$. which is simply 

Given the above, it is clear that $\gamma_i^S$ is monotone in $S$, i.e., $\gamma_i := \gamma_i^\emptyset \geq \gamma_i^S$ and $\gamma_i^S \geq \gamma_i^{S'}$ if $S \subseteq S'$. Furthermore, $\gamma_i^S = 0$ if $S$ contains any flaw not in $\Gamma_i$ since the probability of a resampling introducing a non-neighboring flaw is zero. 

We prove the above claim that for product spaces and MT-resampling steps $\gamma_i$ is indeed at most the probability $p_i = \mu(F_i)$ of an independently random variable assignment to contain flaw $F_i$ in \cref{sec:missing-proof}. 

\begin{lemma} \label{lem:prod}
Consider the variable setting (\cref{def:variable}).
Let $F_j \subseteq \Omg$ be a flaw whose resampling procedure is the MT-resampling.
% Let $\mu$ be the uniform distribution over $\Omg$.
Then, $\gamma_j = \Pr[F_j]$, where the charge is computed with respect to the uniform distribution.
\end{lemma}

\subsection{Convergence Criteria} 
\label{sec:condition}

Here we give several criteria that are sufficient for a search algorithm to terminate fast. We start with the most general one and then show how to successively simplify the conditions (at the cost of making the conditions successively less powerful). \cref{cond:general} requires the general framework setup discussed in \cref{sec:general-setting-setup}. \cref{cond:asymmetric,cond:local-union} describe the most popular LLL setting and are covered by Moser and Tardos' framework. \cref{cond:local-union} is a useful simplification of \cref{cond:general} that applies to many settings. 

% \rtodo{Explain $1 - \eps$ for convergence criteria for fast algorithms.}

\begin{condition} \label{cond:general}
If there exist positive real numbers $\{\psi_i\}_{i=1}^m$ and $\eps \in (0, 1)$ such that, for all $i \in [m] $,
\begin{align}\label{eq:softprelim}
   \frac{ 1 }{ \psi_i }   \sum_{  \substack{ S \subseteq \Gamma_i  } }  \gamma_i^S\prod_{ j \in S} \psi_j    \enspace \leq 1 - \eps,
\end{align}
then the search algorithm terminates rapidly\footnote{See \cref{thm:informal,thm:dyn-full} for the exact number of resampling steps before the algorithm converges.}.
The parameter $\eps$ is called the \emph{slack parameter}.
The running time depends on $1/\eps$. For convenience, we define $\Psi(S) = \prod\limits_{j \in S} \psi_j$ and the product inside the summand can be replaced by $\Psi(S)$. For definitions of the charges $\gamma^S_i$, refer \cref{sec:general-setting-setup}.
\end{condition} 

If we use the monotonicity of charges and simply upper bound $\gamma_i^S$ by $\gamma_i$, then \cref{cond:general} simplifies to $\frac{ \gamma_i }{ \psi_i } \sum_{  \substack{ S \subseteq \Gamma_i  } }  \prod_{ j \in S} \psi_j < 1$. If one sets $x_i = \psi_i / (1+\psi_i) \in [0,1)$, then
\begin{align*}
    \frac{ \gamma_i }{ \psi_i } \sum_{  \substack{ S \subseteq \Gamma_i  } }  \prod_{ j \in S} \psi_j
    = \frac{\gamma_i}{\psi_i}\prod_{j \in \Gamma_i} (1 + \psi_j)
    = \gamma_i \cdot \frac{1}{x_i} \prod_{j \in \Gamma_i \setminus \{i\}} \frac{1}{1-x_j},
\end{align*}
and one obtains the following condition often called the general or asymmetric LLL condition:

\begin{condition} \label{cond:asymmetric}
If there exist numbers $x_1, x_2, \ldots, x_m \in [0,1)$ and $\eps \in (0, 1)$ such that, for all $i \in [m] $,
\begin{align}
   \gamma_i \leq (1-\eps) \cdot x_i \prod_{j \in \Gamma_i \setminus \{i\} } (1- x_j),
\end{align}
then the search algorithm terminates rapidly. 
\end{condition} 

% \noindent To be precise, if an algorithm satisfies \cref{cond:asymmetric}, then it satisfies \cref{cond:general} with $\psi_i = \frac{x_i}{1-x_i}$.

Setting $x_i = \frac{e \cdot \gamma_i}{1 + e \cdot \gamma_i} \in (0,1],$ we have $x_i \prod_{j \in \Gamma_i \setminus \{i\} } (1- x_j) = e\gamma_i \prod_{j \in \Gamma_i} \frac{1}{1+e\gamma_j}$, and the above condition simplifies to $\prod_{j \in \Gamma_i} (1+e\gamma_j) \leq e(1-\eps)$. Using the fact that $1+x \leq e^x$ and $\ln(1-\eps) \geq -\eps(1+\eps) \geq -e\eps$ for $\eps \in (0, 0.68]$, this further simplifies to the following simple local union bound condition:

\begin{condition} \label{cond:local-union}
If there exists $\eps \in (0, 0.68]$ such that, for all $i \in [m]$,
$$\sum_{j \in \Gamma_i} \gamma_j < 1/e - \eps,$$ 
then the search algorithm terminates rapidly. 
\end{condition} 
%\begin{proof}
%\gamma_i \leq e \gamma_i \prod_{j \in \Gamma_i} (1- e \gamma_j/ (1 + e \gamma_j))
%because 1/(1 + x) = 1 - x/(1+x) > 1 - x 
%thus
%1/e \leq \prod_{j \in \Gamma_i} (1- e \gamma_j/ (1 + e \gamma_j))
%because 1 - x/(1 + x) = 1+x/(1+x) - x/(1+x) = 1/(1+x)
%and because 1+x < e^x   and thus 1/(1+x) > 1/e^x = e^{-x}
%we get (1- e \gamma_j/ (1 + e \gamma_j)) > e^{e \gamma_j}
%e^{-1} \leq e^{sum_j - e \gamma_j}
%-1 \leq e^{- e \cdot sum_j  \gamma_j}
%thus
%e sum_j \gamma_j < 1
%\end{proof}

which further simplifies to the following classical condition:

\begin{condition} \label{cond:pde}
Let $p = \max_i \gamma_i$, $d = \max_i |\Gamma_i|$, and $\eps \in (0, 1)$. If
$$p(d+1)e < 1 - \eps,$$
then the search algorithm terminates rapidly. 
\end{condition} 

\noindent In summary, we can convert \cref{cond:local-union} to \cref{cond:asymmetric} by setting $x_i = \frac{e\gamma_i}{1+e\gamma_i}$, and to \cref{cond:general} by setting $\psi_i = \frac{x_i}{1-x_i} = e\gamma_i$.

The following observation will be useful in proving the convergence of our algorithms.

\begin{observation} \label{obs:subset-convergence}
Let $\cF$ be a set of flaws on a set $\Omega$, in which each flaw corresponds to a resampling procedure.
If $\cF$ satisfies \cref{cond:general}, then any subset of $\cF$ also satisfies \cref{cond:general} with the same set of parameters.
\end{observation}

% \rtodo{Explain the techniques and overview of analyzing resampling algorithms}
% \rtodo{Make an informal introduction of the results and proofs}
% \rtodo{Describe }

\section{The Dynamic \LLL} 
\label{sec:dynamic-intro}

% FILL IN: This is easy, one starts with a starting state and a set of $m$ initial flaws that satisfy some convergence criterion, then the search algorithms runs until a flawless state is found, the adversary can then change the instance by adding or removing flaws. For example adding or removing an edge $e = (u,v)$ to a graph changes neighborhoods for $u$ and $v$ and changes the flaws for these neighborhoods. This might lead to some new flaws in the current state and if so the search algorithms continues its resampling until again a flawless state (according to the current set of flaws) is reached at which point the adversary gets to change the instance again. It is assumed that after any changes again some convergence criterion holds for the new set of flaws. 

% \subsection{adversaries and stability/recourse}

% oblivious, adaptive, clairvoyant (note that clairvoyant implies non-uniformly efficient deterministic adaptive algorithms which in turn implies deterministic fixing strategies with small recourse)

% define recourse/stability. 

% \rtodo{Add a subsection on adversaries and recourse/stability}
\def\bF{\mathbf{F}}

In this section, we outline a framework by which nearly all of the combinatorial objects that can be efficiently generated by the constructive versions of the \LLL \cite{achlioptas2019beyond,moser2010constructive} can be made \textit{dynamic}. An adversary can provide a sequence of flaws, and a natural dynamic variant of the sampling algorithms studied previously can generate flawless objects that satisfy all the constraints imposed at \textit{every} intermediate time step.

\subsection{Dynamic setup.}
\label{sec:dynamic-setup}

Initially, we start with a \LLL instance without any flaws and a universe $\Omega$ of combinatorial objects from which we are required to output one flawless object.  Next, an adversary provides a sequence of $q$ queries.
Each query is one of the following,

% \rtodo{Make sure to mention that the algorithm cannot change adaptively. Emphasize that the parameters can change arbitrarily but uniquely.}

\begin{itemize}
    \item `Add flaw $i$' - $i$ is an integer that represents a unique identifier for a flaw to be added. 
    % in such a way that the dependencies and \LLL parameters may be deduced from this unique identifier and the identifiers of the set of queries made so far. For most typical applications, $i \in \poly(q)$, though we leave this as a parameter for the algorithm designer.
    
    % \textbf{Example:} As an example, consider dynamic problems on graphs. In the following sections, we consider this problem on triangle-free graphs with $O(\frac{\Delta}{\ln \Delta})$ colors, where $\Delta$ is the maximum degree. In this example, one requires a flaw set to include $1$-hop neighborhoods. Note that the potential number of $1$-hop neighborhoods is large, up to $\binom{n}{\Delta}$. However, we may associate with each edge a unique $O(\log q)$ bit identifier that denotes the presence of this edge. Using these identifiers, the $1$-hop neighborhood may be expressed using only $O(\log q)$ bits -- which is much fewer than the $\Delta \log n$ needed to specify neighborhoods.
    
    \item `Delete flaw $i$' - A flaw is removed from the constraints imposed on our desired object. It is guaranteed that $i$ is present prior to this operation. Note that any flawless object obtained for the previous query is still considered flawless after this query. 
    % This query is helpful if the dependency graph is to be modified.
    
    % \item `Expand the state space $\Omega$ with $S$' - We set $\Omega_{\text{new}} \gets \Omega_\text{old} \times S$. The set of flaws are retained in the natural way. In particular we consider $(\sigma, s) \in F_{\text{new}}$ for some $\sigma \in \Omega_{\text{old}}$ and $s \in S$ if and only if, $\sigma \in F_{\text{old}}$. Note that the flaws (and their IDs) remain unchanged after this change. This operation may be used if one desired to add variables or expand the state space. The latter may be need if the \LLL is changed due to changes in the underlying dependency graphs. For e.g., if we are to support graph coloring, then as the degree of the graph changes, the number of colors that can be assigned to vertices changes. 
\end{itemize}

Let $\bF_i$ denote the set of flaws imposed for the $i^{th}$ query for $i = 1, 2, \dots q$. We define $\bF_0 = \emptyset$. Note that $|\bF_i \setminus \bF_{i-1}| \leq 1$. We require a dynamic \LLL algorithm to construct $q$ flawless states, where the $i^{th}$ state is free from all flaws in $\bF_i$. We assume that the set of states $\Omega$ remains the same throughout and that the underlying measure $\mu$ over $\Omega$ that satisfies the conditions of \LLL remains the same for each of the $q$ instances. We note that the parameters of the LLL conditions, including the transition probabilities, the fixing procedure, and other parameters of the algorithm, can change, but are completely determined by the sequence of updates.

We assume that the set of flaws in $\bF_i$ for every $i$ satisfy \LLL criteria as needed. In \cref{sec:condition} we provide an overview of different criteria.

\paragraph{Adversaries.} We consider \textit{oblivious}, \textit{adaptive} and \textit{clairvoyant} adversaries. Our most general results and many of our applications hold for the most general clairvoyant adversary. We explicitly mention when this is not the case.

An \textit{oblivious} adversary is impervious to the algorithm and the source of randomness that the algorithm uses. In particular, such an adversary fixes an update sequence of flaws and outputs this sequence to the algorithm. The adversary selects a flaw that is independent of any random bits used by the algorithm.

An \textit{adaptive} adversary can observe the output of the algorithm and the random bits that it has used so far. An adaptive adversary, upon providing an update sequence of flaws, observes the algorithm's output for these updates and subsequently provides the next update.

A \textit{clairvoyant} adversary is as powerful as an adaptive adversary, except that it receives the entire random bit string that the algorithm uses before even the first update. Based on this information, the adversary can then decide an update sequence to provide to the algorithm. Observe that, given the random string, the clairvoyant adversary can simulate the algorithm's output and thus compute and provide a worst-case input. We do not impose any other restriction on the adversary.

\paragraph{Recourse and stability.} The recourse of a dynamic LLL algorithm is the number of changes made to the output of a solution between successive updates. In the setting where the state is described by a Cartesian product of independent variables, the recourse is the number of variables whose values change between successive updates.

\subsection{Dynamic resampling algorithms}
\SetKwProg{Fn}{Function}{:}{}
\SetKwFunction{FDFSFix}{Fix}
\SetKwFunction{FInsert}{Insert}
\SetKwFunction{FRemove}{Remove}
\SetKwFunction{FExpand}{Expand}
\SetKwFunction{FSelect}{Select}
\SetKwFunction{FAddress}{Address}

We study natural dynamic variants of the resampling algorithms for the constructive \LLL analyzed in the literature. We begin by outlining a generic dynamic algorithm \cref{alg:ais} that covers any application that uses the AIS framework. Next, we discuss an algorithm for LLL instances satisfying \cref{cond:asymmetric,cond:pde}. Largely, this is a special case of \cref{alg:ais} except for some subtleties that we explain. We also briefly discuss implementation details of efficiently maintaining flaws so that they can be efficiently identified.

In \cref{sec:overview-analyzing-resampling-algorithms} we provide an outline of how we analyze these dynamic resampling algorithms and how our approaches differ from the literature. Formal proofs of the number of resamplings are provided in \cref{sec:moser-tardos-convergence,sec:gen-framework-convergence}. We discuss the implementation details of maintaining flaws in \cref{sec:applications}. 

\subsubsection{A general dynamic resampling procedure}

Our general resampling algorithm is outlined in \cref{alg:ais} and is based on the framework of \cite{achlioptas2019beyond}, which encompasses nearly all conceivable LLL instances. A crucial difference of our resampling algorithms is that flawless objects are always constructed in between flaw updates. One may take any \LLL instance that is describable via the framework and plug it in \cref{alg:ais} to obtain the corresponding dynamic variant. Next, we describe the functions $\FInsert, \FRemove, \FSelect$ and $\FAddress$. This property will be crucial in our analysis of the resampling algorithm.

\begin{algorithm}[htbp]
    \Fn{\FInsert{$B$}}{
        \textbf{Ensure:} No flaw other than $B$ exists and \cref{cond:general} is satisfied after update.

        \textbf{Global Parameters:} 
        
        \hspace{0.4cm} (i) $t$ denotes the total number of calls to $\FAddress$ made so far (initially $0$).

        \hspace{0.4cm} (ii) $j$ denotes the number of calls to $\FInsert$ made so far (initially $0$).

        $F \gets \FSelect{\textbf{F}, t}$
        
        \While{$F \neq \perp$}{
            $\sigma \gets \mathtt{Address}(F, \sigma, \bF, j)$
            
            $F \gets \FSelect{\textbf{F}, t}$
            
            $t \gets t + 1$
        }
        
        $j \gets j + 1$
    }
    \Fn{\FRemove{$B$}} {
        Remove $B$ from the current set of flaws
    }
    \Fn{\FSelect{$\textbf{F}, t$}} {
        Outputs a flaw $F \in \textbf{F}$. 
        We emphasize that the selection of $F$ depends only on $\textbf{F}$, $t$, and, in particular, is independent of the randomness used by $\FAddress$.

        Return $\perp$ if the state is flawless.
    }
    \Fn{\FAddress{$F, \sigma, \bF$}}{
    \textbf{Ensure:} State transition probabilities and their charges satisfy \cref{cond:general}.
    
    \textbf{Input:} A flaw $F$, the current state $\sigma$ and the current set of active flaws $\textbf{F}$.
    
        A randomized procedure that computes a state $\tau \in \Omega$ so as to ``fix'' $F$. The procedure can be arbitrary and defines transition probabilities $\rho^\sigma_F$. The only restriction we impose is that for every $\bF$, the transition probabilities satisfy \cref{cond:general}.
        % The function reads random bits from a sufficiently large string of bits. Each flaw $F$ has its own string bits.
    %     % Each successive call of this function reads more bits from that input tape.

        \textbf{Output:} A state $\tau \in \Omega$.
    }
    \caption{A generic template for dynamic LLL resampling procedures fitting the framework of \cite{achlioptas2019beyond}.}
    \label{alg:ais}
\end{algorithm}

The $\FInsert$ procedure takes as input (i) the current state $\sigma \in \Omega$ of the algorithm (or assignment of variables in the variable setting) and (ii) the flaw to be added. It returns a state that is free from the current set of flaws. 

The function $\FInsert$ invokes two functions, $\FSelect$ and $\FAddress$. The $\FSelect$ function selects a flaw $F$ to fix and the $\FAddress$ procedure executes the resampling algorithm to fix $F$. 

\paragraph{Description of flaw selection.} It is important to note that we require that $\FSelect$ \textbf{does not} depend on the random bits used by $\FAddress$. In particular, given the current active set of flaws $\bF$ and the iteration of the algorithm, the output of $\FSelect$ must be determined. It is possible for $\FSelect$ to change across different iterations. For e.g., it is even possible to randomly select a flaw. The crucial property of the $\FSelect$ function is that it is \textit{independent} of the random bits used by $\FAddress$. 

In all of our applications, the $\FSelect$ procedure decides on a permutation $\pi$ over the set of all flaws and selects the present flaw with the highest rank.
% Recall that we will use the witness and compute the probability that it is consistent with a random table $\cR$. In this analysis

\paragraph{Description of flaw resampling procedures.} The $\FAddress$ procedure takes as input the current state $\sigma$, flaw $F$, and current active set of flaws $\bF$ and then outputs a new state $\sigma'$ after invoking a randomized resampling procedure designed for the specific application. We note that the procedure can depend in an arbitrary way on the \textit{current state} of the algorithm, as long as the charges satisfy \cref{cond:general}.

\paragraph{LLL parameters.} We allow the LLL parameters described by \cref{cond:general} to vary within iterations. In particular, the parameters for the set of active flaws $\bF_i$ after $i$ updates may be different from those of $\bF_{i-1}$. Note that a change in the flaw addressing procedure changes the charges associated with that flaw and hence also the parameters of \cref{cond:general}.

% On top of the $\FInsert$ and $\FRemove$ functions, we require (i) a \textit{flaw selection procedure} $\FSelect$ and (ii) a \textit{flaw resampling procedure} $\FAddress$. The procedure $\FSelect$ depends on the set of updates made so far and the current set of flaws. In particular, $\FSelect$ does not depend on the current state. 

\subsubsection{Dynamic Moser-Tardos resamplings}

A simpler setting is the variable setting where the set of states $\Omega$ is a product distribution over several variables. 
Each flaw $F$ has a minimal set of variables $\vbl(F)$ whose values completely describe whether a given state has flaw $F$.

In this setting, the flaw selection procedure $\FSelect$ can be \textit{arbitrary}. 
Once a flaw $F$ is selected, the $\FAddress$ procedure resamples all the variables in $\vbl(F)$ uniformly at random. If the resampling procedures satisfy \cref{cond:asymmetric}, then the resampling algorithm converges quickly.

\subsection{Efficient flaw selection procedures.}

In \cref{alg:ais} we outlined $\FSelect$ and $\FAddress$ as generic procedures that select and address flaws respectively. However, in many algorithmic settings, it may not be clear how to efficiently identify and select the highest priority flaw. 

A common tactic we employ to identify and maintain the highest priority flaw is outlined in \cref{alg:select}. We maintain a global priority queue $Q$ that stores flaws that were created by the $\FAddress$ procedure. The flaws are ordered by the permutation $\pi$. For the MT-resampling procedure and in many of our applications, the number of possible flaws created due to one invocation of $\FAddress$ is small. Thus, after the execution of $\FAddress$, we can iterate over them and then insert them into the priority $Q$ if they are introduced.

Whenever we require to select a flaw $F$, we pop from the priority queue $Q$. It is easy to see that $Q$ always contains the highest priority flaw. It is possible that the highest priority flaw in $Q$ has been collaterally fixed, in which case we move to next highest prority flaw in $Q$.

% More details are presented in \cref{sec:applications}.

\begin{algorithm}
    \Fn{\FSelect{$\sigma, \pi, Q$}} {
    
        \textbf{Input:}
        (i) $\pi \gets$ a permutation of $[q]$ indicating priorities of the flaws.

        (ii) $Q \gets$ a priority queue of the list of flaws.

        (iii) $\sigma$ is the current state of the algorithm.
        
        \While{$Q$ is not empty} {

            $F \gets $ pop the highest priority flaw from $Q$    

            \If{$F$ is a flaw at $\sigma$} {
                
                \Return $F$
                
            }
        } 
        \Return $\perp$
    }
    \Fn{\FAddress{$F, \sigma, t, j, Q$}} {
        \textbf{Input:} Chosen flaw fix $F$, current state $\sigma$, resampling number $t$, update number $j$ and priority queue $Q$.

        Execute the random process that fixes flaw $F$ at $\sigma$ and returns a new state $\tau$ with probability distribution $\rho^{\sigma}_F$.

        \For{every possible newly introduced flaw $F'$ due to the resampling of $F$ at $\sigma$ } {
            \If{$F'$ is a flaw at $\tau$} {
                Add $F'$ to $Q$
            }
        }

        Return $\tau$
    }
    \caption{Efficient implementation of the selection and address procedure for our applications of the LLL}
    \label{alg:select}
\end{algorithm}

\subsection{Overview of the analysis of resampling algorithms}
\label{sec:overview-analyzing-resampling-algorithms}

In this section, we provide an overview of the elegant proof by \cite{moser2010constructive} to analyze the running time of resampling algorithms. We also describe how our proofs differ from these works. Formal proofs of the analysis of the dynamic resampling algorithms in the two settings are described in \cref{sec:moser-tardos-convergence,sec:gen-framework-convergence}.

At a high level, the analysis considers an execution of the algorithm that takes a large number of resamples to converge. Note that the execution is fully specified by (i) the random string that the algorithm uses, (ii) the sequence of flaw resamplings that the algorithm makes, and (iii) the starting state.

\paragraph{Step 1 (Constructing a witness).} The execution of the algorithm is compressed into a succinct object called the ``witness''. The witness acts as an attestation or ``proof'' that the algorithm made a given number of resamplings. We next describe two desired properties of the witness, whose role will be clearer in subsequent steps. First, the witness must have an efficient representation that depends only on the resampling sequence, and independent of the random string $\cR$ used by the algorithm. 
% For e.g., the total number of witnesses must be ``reasonably small''.
Second, it must encode sufficient information such that given a random string, one can \textit{partially} verify whether the witness could correspond to an execution of the algorithm with $\cR$. The verification process must be a necessary condition, but need not be sufficient. We ideally require a witness that is as short as possible and whose verification passes with as small probability as possible against a uniform random string $\cR$.

\paragraph{Step 2 (Computing the weight of the witness).} Given a specific witness, we can compute the probability that the witness is consistent with a uniformly random string. More precisely, suppose that the random string $\cR$ is given. One can iterate through the witness and by the property mentioned above, identify for each of the invocations of flaw resamplings that failed, the exact random bits from $\cR$ that resulted in the need for this resampling. From this information, we can verify whether the witness can be produced for a given $\cR$. The probability that the witness passes this test for a uniformly selected $\cR$ is called the ``weight'' of the witness.

\paragraph{Step 3 (Calculating the total weight of all witnesses).} Observe that the sum of the weights of all possible witnesses is an upper bound on the expected number of resamples made by the algorithm. In particular, the probability that the witness makes at least $s$ resamplings is upper bounded by the total weight of all witnesses of size at least $s$.

If the \LLL conditions come with an $\eps$ slack, then the probability that a witness of size $s$ is consistent is reduced further by an $(1 - \eps)^s$ factor.

Let $\cF$ denote the set of all witnesses. For a witness $\phi \in \cF$, let $p_\phi$ denote its weight. Let $\cF_s$ denote the set of witnesses corresponding to executions that made at least $s$ resamplings. We have,
\begin{align*}
\Pr[\text{algorithm makes } s \text{ resamplings }] \leq \sum\limits_{\phi \in \cF_s} (1-\eps)^s p_\phi \\
\leq (1 - \eps)^s \cdot \sum\limits_{\phi \in \cF} p_\phi 
\end{align*}

Thus, the number of resamplings is logarithmic in the total weight of all witnesses. The total weight is bound depending on the application. \cite{moser2010constructive} provided an ingenious coupling argument that is able to provide an elegant closed-form bound for the expression. We next outline the specific details of the witness construction and bounds for three regimes of convergence criteria.

\paragraph{Moser-Tardos Framework.} In the MT-Framework, the witness was a labelled rooted tree that expressed dependencies between flaw resamplings in reverse order. 

We next outline the witness generation process. For every run of the algorithm that made $s$ resamplings, the resamplings are iterated through in reverse order and when a resampling of flaw $F$ is visited, a new vertex is added with label $F$ and it is either (i) attached as a child of the highest depth vertex already in the tree that had label $F'$ that depends on $F$ or (ii) added as a new root.
% The ideas are similar to the above adaptation of the general framework, except that the approach of \cite{moser2010constructive} requires the input distribution to be from a product distribution; however, the resampling can be \textit{arbitrary}. The key difference is that we build the witness forest in \textit{forward} direction. This allows us to argue that the total number of resamplings is only $\tilde{O}(q)$ where $q$ is the number of updates. The key difference is that between updates, it is guaranteed that the algorithm is in a flawless state. 

\paragraph{{\AIS} Framework.} The most general criteria incorporate point-to-set correlations by \cite{achlioptas2019beyond}. For this condition, we only require that there exists an execution-independent rule for selecting which flaws to resample at any given point in the resampling algorithm. Such a restriction is not prohibitive for most applications of the \LLL. As mentioned by \cite{achlioptas2019beyond}, one can even sample a flaw uniformly at random by sampling a priority rule for the flaws for every iteration of the resampling algorithm. We adapt their witness sequence for the dynamic setting.

\paragraph{Depth-first search-based resampling.} The initial analysis of local search-based resampling algorithms relied on DFS-based procedures \cite{moser2010constructive}. In such procedures, there is an additional restriction imposed on which flaw is selected to be resampled on fixed next. The procedure starts by selecting a flaw $F$ to fix. While $F$ is still flawed, it chooses any flaw $F'$ dependent on $F$ and then recursively invokes the fixing procedure on $F'$. Such procedures are easier to analyze because the witness tree can simply be the DFS execution tree, which is much more restricted than the executions arising out of arbitrary resampling procedures. Furthermore, one can impose a priority rule for selecting the flaws and ensure that \textit{higher} priority flaws are fixed \textit{only after} all lower priority flaws are fixed. This can help because such information may be used in the success probability of fixing the higher-priority flaws. Finally, an in-order traversal of the DFS tree recovers the exact resampling sequence, which makes arguing about the probability of being consistent with a random string $\cR$ easier. The arguments also used the idea of \textit{entropy compression}. One can also analyze the dynamic analogue of such algorithms with the entropy compression method.

% \subsection{Bounds on the number of resamplings.}

% Our bounds on the number of resamplings of the dynamic algorithm can be obtained from Remark 2.4 of \cite{achlioptas2019beyond} -- if applied in the right way. The dynamic algorithms that we analyze do not fit into the framework of permutation-based flaw selections that they use for most results in their paper. Recall that we imposed that the flaw selection procedure must not depend on the randomness $\cR$ used in flaw resamplings. In particular, this means that the flaw selection cannot depend on the current state of the algorithm. 
% We restate their theorem below.

% \begin{theorem}[Remark 2.4 of \cite{achlioptas2019beyond} (paraphrased).]
%     For a resampling algorithm describable under the framework of \cite{achlioptas2019beyond} for which there exists positive real parameters $\psi_1, \psi_2, \dots \psi_m$ whose charges satisfy \cref{cond:general}, the number of resamplings of flaws at at most $s$ with probability at least $1 - \delta$ where,

%     \begin{align*}
%         \eps s = O\left( \log {\max\limits_{\sigma \in \Omega} \frac{\theta(\sigma)}{\mu(\sigma)}} + \log \left(\sum\limits_{S \subseteq \text{span}(\theta)} \Psi(S) \right) + \log\limits_{S \subseteq [m]} \frac{1}{\Psi(S)} + \log \delta^{-1}\right)
%     \end{align*}

%     Here, $\theta, \mu$ denote the distributions of the initial state and the measure on $\Omega$ with which the charges are computed.
% \end{theorem}

% \paragraph{Oblivious update sequence.} Let us 

\section{Moser-Tardos in the variable setting -- proof of convergence}
\label{sec:moser-tardos-convergence}

We consider the same dynamic setting as before, except that the flaws after each update satisfy either \cref{cond:asymmetric} or \cref{cond:pde}. 
We follow largely the same procedure as \cite{moser2010constructive}, except that we need to modify the analysis to account for the fact that the dynamic algorithm only observes a single flaw in the object it constructs at the beginning of an update.
We consider two settings, depending on the \textit{initial} assignment of the variables. In the first setting, we assume that the initial assignment of the random variables is uniform. This is natural, since every variable must be assigned a value at some point in the algorithm, and hence, the initial value might as well be random. However, in some settings, we might want to \textit{compose} several dynamic algorithms together. In such cases, it might not be feasible to assume that the initial assignment is random, since the starting assignment might come from an independent instance of the LLL. For the second setting, we assume that the initial assignment is \textit{arbitrary}. 
The difference in the two settings turns out to be negligible in the symmetric setting and several sets of parameters in the asymmetric setting. We will outline both versions of the proof simultaneously and refer to them as ``Setting 1'' and ``Setting 2'' respectively. We describe setting 1 in the upcoming sections and outline the differences to ``Setting 2'' in \cref{sec:mt-arbitrary-init}.
% \stodo{Remember to update the reference.} \rtodo{done.}

Our approach starts by bounding the number of resamplings for a fixed update sequence of flaws $F_1, F_2, \dots F_q$. We observe that the probability bound on the number of resamplings for a fixed sequence has an exponential tail. Then, a union bound over all possible update sequences yields our desired bounds.
In \cref{sec:mt-witness-defn} we define our witness corresponding to a run of the algorithm on the fixed update sequence. This differs in a subtle way from the witness trees defined in \cite{moser2010constructive}.  Subsequently, in \cref{sec:mt-witness-weight} we derive an expression for the probability that a specific witness is generated when \cref{alg:moser-tardos} is run with a uniform random string. This expression differs in the two settings that we consider. 
Finally, in \cref{sec:mt-witness-union-bound}, we compute the union bound by coupling the output of the algorithm with a Galton-Watson process as in \cite{moser2010constructive}. 

\begin{algorithm}

    \Fn{\FInsert{$B$}}{
        \textbf{Ensure:} No flaw other than $B$ exists.
        
        \While{there exists at least one flaw $F$}{
            $F \gets $ an \textbf{arbitrary} flaw.
            
            $\FAddress{F}$.
        }
    }
    \Fn{\FRemove{$B$}} {
        Remove $B$ from the current set of flaws.
    }
    \Fn{\FAddress{$F$}}{
    \textbf{Ensure:} \cref{cond:asymmetric} or \cref{cond:pde} is satisfied.
    
        All variables in $\vbl(F)$ are resampled uniformly at random.
    }
    \caption{An algorithm to be executed after every update.}
    \label{alg:moser-tardos}
\end{algorithm}

\subsection{Witness Construction}
\label{sec:mt-witness-defn}

The original proof by \cite{moser2010constructive} constructed witness trees by traversing the resampling sequence in reverse order. In particular, there was a witness tree for every flaw $A$ whose root corresponds to the last resample of $A$. We then iterate over the witnesses in reverse order, and when we encounter a flaw $B$ that depends on any of the flaws added so far, we attach $B$ as a child of the highest-depth node. If there are multiple candidates, we arbitrarily choose one.

Such a witness is not a good candidate for the dynamic setting that we analyze. In particular, it is very hard to capture the fact that all flaws arise from the latest update. For this reason, we consider the \textit{forward} tree which we next define in a similar manner.

\def\lbl{\mathsf{label}}
\paragraph{Witness construction.} Recall that the witness is defined for a specific update sequence and a run of \cref{alg:moser-tardos} for this update sequence. Let $\bF_i$ denote the set of flaws after the $i^{th}$ update. Let $\{x_i(j) \mid j \in \bF_i\}$ denote the asymmetric \LLL parameters (see \cref{cond:asymmetric}) corresponding  to $\bF_i$.
The witness is an ordered set of $q$ labelled rooted trees with labels in $[q]$. We refer to witnesses by $\tau$ and the $i^{th}$ component of the witness by $\tau_i$.
The $\tau_i$ corresponds to the resamplings made by the algorithm after the $i^{th}$ update. We allow the trees to be empty. An empty tree indicates that no resamplings were needed when an update was made. This always occurs for deletion updates. An empty tree can also occur when the solution for the $(i-1)^{th}$ update (or the initial state) did not contain the added flaw.

\def\bS{\mathbf{S}}
Consider a run of the algorithm that makes a total of $s$ resamplings, and $s_i$ resamplings after the $i^{th}$ update given by $\bS_i = \{F_{i,j} \mid j\in[s_i]\}$.
If $\bS_i = \emptyset$, $\tau_i$ is empty. Otherwise, by the constraints of \cref{alg:moser-tardos}, the $i^{th}$ update must have been a flaw addition, and $F_{i, 1}$ must be the flaw that was added. 
Each vertex in $\tau_i$ corresponds to a flaw resampling $F_{i, j}$ and this vertex has label $F_{i, j}$.
The root of $\tau_i$ is $F_{i, 1}$. The parent of the vertex corresponding to the resampling $F_{i, j}$ for $j > 1$ is the flaw resampling $F_{i, k}$ for $k < j$ that introduced $F_{i, j}$.

\subsection{Weights of the witness}
\label{sec:mt-witness-weight}

\paragraph{Resampling Table.} For the MT-resampling procedures, we require explicitly stating the source of randomness. We assume that corresponding to each flaw $F$ there exists an infinite random string from which the random bits needed for the resampling of $F$ are taken. In particular, the $i^{th}$ resampling will take the $i^{th}$ consecutive sub-segment of bits from the random string.

\begin{definition}
\label{def:proper-witness}
    We say that a witness $\tau$ is \textit{proper} if no two children are labelled with flaws that depend on each other. In particular, the labels of the children of every vertex are distinct.
\end{definition}

\begin{claim}
\label{clm:mt-witness-prob}
    For a resampling table $\cR$, the probability that a witness $\tau$ is consistent with $\cR$ is given by,
    \begin{align*}
        \Pr[\tau \text{ appears in the witness forest}] \leq  \prod\limits_{v \in \tau} \Pr[F_v]
    \end{align*}
\end{claim}

\begin{proof}
    Given a proper witness forest $\tau$, we first show that one can identify the random bits in the table $\cR$ that were used to generate $\tau$. Since the roots are ordered by the labels $1, 2, \dots q$, we can determine which components correspond to which updates. We iterate through the components in increasing order in which they were added and identify the random bits used for the resamplings of the corresponding nodes in a breadth-first search order within each tree.

    Suppose we visit a node $v$ labelled with flaw $F_v$ in breadth-first search order. By definition, all the resamplings of all flaws $F_j$ that depend on $F_v$ must correspond to vertices in $\tau$ that have a strictly smaller depth than $v$ and must have been already visited by the breadth-first search traversal.

    For a variable $x_t \in \vbl(F_v)$ let the number of vertices already traversed that depend on $x_t$ be $k$. The resampling of $F_v$ must take the $(k+1)^{th}$ bit in the table $\cR$ corresponding to the variable $x_t$.

    Finally, just prior to each resampling $F_v$ of each node $v$ in $\tau$, one can determine a unique set of positions in $\cR$ that correspond to the current value of $F_v$. Observe that the values of some of the variables could have been set as a result of prior resamplings of other flaws that depend on $F_v$.

    As we run through the BFS traversal of each component of the witness forest, we make sure that bits in $\cR$ just prior to the resampling of $F_v$ result in an evaluation that violates $F_v$. This checks that $F_v$ was actually a flaw prior to the resampling. For a given table $\cR$, the probability that a particular vertex $v$ with label $F_v$ passes the check is exactly $\Pr[F_v]$ and the check procedure is independent for all flaws, since each entry in $\cR$ corresponds to at most one vertex.
\end{proof}

\subsection{Coupling with a forest generation process.}
\label{sec:mt-witness-union-bound}

Given \cref{clm:mt-witness-prob}, we would like to union bound the probability over all proper witness forests $\tau$. In the symmetric setting, one can simply count the number of distinct forests that can arise and each forest of size $s$ has the same probability bound. In the asymmetric setting, we require a coupling argument with a Galton-Watson type random witness generation process.

Consider a random process that generates the witness as follows. The roots of the forest are fixed to have labels $F_1, F_2, \dots F_q$. The connected components corresponding to each root are generated independently. For a vertex $v$ in the forest with label $F_v$ and whose set of dependent flaws are $\Gamma(F_v)$ recall that the \cref{cond:asymmetric} requires that there exists reals $x(F_i)$ for each $i = 1, 2, \dots q$ such that $\Pr[F_i] \leq (1 - \eps) x(F_i) \prod\limits_{F_j \in \Gamma(F_i) - F_i} (1 - x(F_j))$. Let $\gamma(F_i) = x(F_i) \cdot \prod\limits_{j \in \Gamma(F_i) - F_i} (1 - x(F_j))$.

For each $F_j \in \Gamma(F_v)$, with probability $x(F_j)$ we add a child to $v$ with label $F_j$.

The random process either ends when there are no new nodes to process, or proceeds indefinitely.

\begin{claim}
\label{clm:mt-witness-generation-prob}
    For a proper witness $\tau$, whose non-empty trees have roots $F_1, F_2, \dots F_q$ with \LLL parameters $x(F_i)$ when they were resampled, the probability that the random process outputs $\tau$ is,
    \begin{align*}
        \Pr[\tau \text{ is generated }] = \rb{\prod\limits_{i=1}^q \frac{1 - x(F_i)}{x(F_i)}} \cdot \prod\limits_{\substack{v \in \tau}} \gamma(F_v)
    \end{align*}
    where $\gamma(F_v) = x(F_v) \cdot \prod\limits_{j \in \Gamma(F_v) - F_v} (1 - x(F_j))$. In particular, recall that \cref{cond:asymmetric} implies that $\Pr[F_v] \leq (1 - \eps) \gamma(F_v)$.
\end{claim}

\begin{proof}
    We prove the statement by induction. We prove that for a component of $\tau$ with root $F_i$, the probability that the component $\tau_i$ with root $F_i$ is generated is given by $\frac{1 - x(F_i)}{x(F_i)} \prod\limits_{v \in \tau_i} \gamma(F_v)$. The claim then follows as the forest generation is a product distribution over each connected component.

    For the base case with a single root, the expression evaluates to $\prod\limits_{j \in \Gamma(F_i)} (1 - x(F_j))$ which is exactly the probability that the process dies out due to no child being added.

    Suppose the claim is true for all trees with size $n$. Consider a tree of size $n + 1$ that is obtained by attaching a leaf vertex $u$ to some $n$ vertex tree $\tau'$. We have,
    \begin{align*}
        \Pr[\tau \text{ is generated }] &= \Pr[\tau' \text{ is generated }] \cdot \frac{x(F_u)}{1 - x(F_u)} \cdot \prod\limits_{j \in \Gamma(F_u)} (1 - x(F_j)) \\
        &= \Pr[\tau' \text{ is generated }] \cdot x(F_u) \cdot \prod\limits_{j \in \Gamma(F_u) - F_u} (1 -x(F_j)) \\
        &= \Pr[\tau' \text{ is generated }] \cdot \gamma(F_u)
    \end{align*}
    Plugging the induction hypothesis for the probability of generating $\tau'$ in the previous step completes the induction.
\end{proof}

\begin{claim}
\label{clm:witness-generation-union-bound}
    Let $\cF$ denote the set of all proper witness forests. We have,
    \begin{align*}
        \sum\limits_{\tau \in \cF} \prod\limits_{v \in \tau} \gamma(F_v) \leq  \prod\limits_{i =1}^q \frac{1}{1 - x_i(F_i)}
    \end{align*}
\end{claim}
\begin{proof}
    Let $S \subseteq [q]$ be an arbitrary subset of $[q]$. Let $\cF(S)$ denote the set of all possible witnesses that have non-empty trees only in the entries of $S$. Consider the witness generation process that starts with $S$. Since the random process generates at most one proper witness forest, we have,
    \begin{align*}
    \sum\limits_{\tau \in \cF} \Pr[\tau \text{ is generated }] \leq 1 \\
    \Rightarrow \prod\limits_{i \in S} \frac{1 - x(F_i)}{x(F_i)} \cdot \sum\limits_{\tau \in \cF(S)} \prod\limits_{v\in\tau} \gamma(F_v) \leq 1
    \end{align*}
    Rearranging the last inequality and summing over all possible sets $S$ gives,
    \begin{align*}
        \sum\limits_{S \subseteq [q]} \sum\limits_{\tau \in \cF(S)} \prod\limits_{v \in \tau} \gamma(F_v) \leq \sum\limits_{S \subseteq [q]} \prod\limits_{i \in S} \frac{x_i(F_i)}{1 - x_i(F_i)} = \prod\limits_{i \in [q]} \frac{1}{1 - x_i(F_i)}
    \end{align*}
\end{proof}

\begin{theorem}
    For a fixed update sequence $F_1, F_2, \dots F_q$ with \LLL parameters satisfying \cref{cond:asymmetric}, the Moser-Tardos resampling algorithm,\cref{alg:moser-tardos} terminates within $s$ resamples with probability at least $1 - \delta$ where,
    \begin{align*}
        \eps s = O\left(\log \delta^{-1}  + \sum\limits_{i=1}^q \log (\frac{1}{1 - x_i(F_i)}) \right)
    \end{align*}
\end{theorem}

\begin{proof}
    We union bound \cref{clm:mt-witness-prob} over all witness forests using the bound in \cref{clm:mt-witness-generation-prob}.

    Let $\cF$ denote all proper witness forests. Let $\cF_s$ denote those witnesses that have size at least $s$.

    We have,
    \begin{align*}
        \Pr[\cref{alg:moser-tardos} \text{ makes at least } s \text{ resamples }] &\leq \sum\limits_{\tau \in \cF_s} \Pr[\tau \text{ is plausible}] \\
        &= \sum\limits_{\tau \in \cF_s } \prod\limits_{v \in \tau} \Pr[F_v] 
    \end{align*}

    Next, we bound the summand in the above expression.
    \begin{align*}
        \sum\limits_{\tau \in \cF_s } \prod\limits_{v \in \tau} \Pr[F_v] \leq (1-\eps)^{s}  \cdot \sum\limits_{\tau \in \cF_s } \prod\limits_{v \in \tau} \gamma(F_v) \\ 
        \leq (1-\eps)^{s} \cdot \sum\limits_{\tau \in \cF}\prod\limits_{v \in \tau} \gamma(F_v)
    \end{align*}

    Using \cref{clm:witness-generation-union-bound} we get,
    \begin{align*}
    \Pr[\cref{alg:moser-tardos} \text{ makes at least } s \text{ resamples }] &\leq 
        (1 - \eps)^{s} \cdot \prod\limits_{S \subseteq [q]} \cdot \prod\limits_{i \in S}\frac{x_i(F_i)}{1 - x_i(F_i)} \\
        = (1 - \eps)^{s} \prod\limits_{i=1}^q \left (1 + \frac{x_i(F_i)}{1 - x_i(F_i)} \right) = (1 - \eps)^{s} \prod\limits_{i=1}^q \frac{1}{1 - x_i(F_i)} 
    \end{align*}

    Finally we choose $s$ so that the expression above is at most $\delta$. We get,
    $$s \geq \log_{1/1-\eps} \left( \sum\limits_{i=1}^q \log (\frac{1}{1 - x_i(F_i)}) + \log \delta^{-1}\right)$$
\end{proof}
\def\cost{\mathsf{cost}}
\begin{theorem}
\label{thm:dyn-full}
    For a \LLL instance satisfying \cref{cond:asymmetric}, with flaws $F_1, F_2, \cdots F_k$, let its cost be given by,
    \begin{align*}
        \cost(F_1, F_2, \dots F_k) = \sum\limits_{i=1}^k \log \frac{1}{1-x_i(F_i)}
    \end{align*}

    If the number of update sequences is at most $2^m$, with each sequence satisfying \cref{cond:asymmetric} with cost at most $M$, then the probability that \cref{alg:moser-tardos} terminates before making $s$ resamples is at least $1 - \delta$ where,
    \begin{align*}
        \eps s = O\left( q + m + M + \log \delta^{-1}  \right)
    \end{align*}

    In particular, if the dependence degree is $d$ and $x_{\max{}}$ is the maximum value of the LLL parameter $x(F_i)$ over all update sequences and flaws $F_i$, then $M \leq q \log \frac{1}{1 - x_{\max{}}}$. In particular, in the typical symmetric setting where $x_{\max{}} \leq 1/ed$, $M = O(q)$.
\end{theorem}

\subsection{Arbitrary initial state}
\label{sec:mt-arbitrary-init}

We describe here the proof for the version when the initial state is arbitrary. The following claim is a variant of \cref{clm:mt-witness-prob}
\begin{claim}
\label{clm:mt-arb-witness-prob}
For (i) a given update sequence $F_1, F_2, \dots F_q$, (ii) initial state $\sigma \in \Omega$, (iii) witness $\tau$ and (iv) a random table $\cR$,
\begin{align*}
    \Pr[\text{witness } \tau \text{ is plausible with } \cR] \leq \left( \prod\limits_{F_i \in \tau} \frac{1}{\Pr[F_i]} \right) \cdot \prod\limits_{v \in \tau} \Pr[F_v]
\end{align*} 
\end{claim}
\begin{proof}
    Note that the only difference to \cref{clm:mt-witness-prob} is the term $\prod\limits_{i \in S} \frac{1}{\Pr[F_i]}$. This is because, the first resamplings of every flaw $F_i$ possibly cannot be accounted for, since some of the variables in $\vbl(F_i)$ may be initialized by an adversary.
\end{proof}

\begin{theorem}[MT-Oblivious-Arbitrary-Start]
\label{thm:mt-arb-oblivious}
    For a given update sequence, the probability that \cref{alg:moser-tardos} terminates within $s$ resamplings is at least $1 - \delta$ where,
    \begin{align*}
        \eps s = O \left(  q + \log |\Omega| + \log \prod\limits_{i = 1}^q \max\limits_{j \leq i}\frac{1}{\gamma_j(F_i)} + \sum\limits_{i=1}^q\log \frac{1}{1 - x_i(F_i)} + \log \delta^{-1}  \right)
    \end{align*}
\end{theorem}
\begin{proof}
    The witness generation procedure is the same. Using $\Pr[F_v] \leq (1 - \eps) \gamma(F_v)$ in \cref{clm:mt-arb-witness-prob},
    \begin{align*}
        \Pr[\text{witness } \tau \text{ is plausible with } \cR] \leq (1 - \eps)^s \prod\limits_{i \in [q]} \frac{1}{\gamma(F_i)} \prod\limits_{v \in \tau} \gamma(F_v)
    \end{align*}
    We next union bound the above over all possible witnesses $\tau$ and all possible initial assignments in $\Omega$. We get,
    \begin{align*}
        & \Pr[\cref{alg:moser-tardos} \text{ makes } s \text{ resamplings for arbitrary initialization}]  \\
        \leq & (1 - \eps)^{s - q} \cdot |\Omega| \cdot \prod\limits_{i \in [q]} \max\limits_{j \in [q]} \frac{1}{\gamma_j(F_i)} \cdot \prod\limits_{i \in [q]} \frac{1}{1 - x_i(F_i)}
    \end{align*}

    Choosing $s$ so that the above expression is at most $\delta$ we get,
    \begin{align*}
        \eps s = O\left(q +  \log |\Omega| + \log \max\limits_{\substack{(i,j) \in [q] \\ j \geq i}} \prod\limits_{i = 1}^q\frac{1}{\gamma(F_i)} + \sum\limits_{i=1}^q\log \frac{1}{1 - x_i(F_i)} + \log \delta^{-1} \right)
    \end{align*}
\end{proof}

\section{AIS framework -- proof of convergence}
\label{sec:gen-framework-convergence}

In this section, we show that \cref{alg:ais} terminates rapidly, assuming \cref{cond:general} is satisfied after each update. Note that the parameters may change between successive flaw updates.
 
% \begin{condition}[\cite{achlioptas2019beyond}]
% \label{cond:start}
%     There exists reals $\psi_1, \psi_2, \dots \psi_m$ such that for every flaw $i \in [m]$,

%     \begin{align*}
%         \frac{1}{\psi_i} \cdot \sum\limits_{S \subseteq [m]} \gamma^S_i \Psi(S) < 1 - \eps
%     \end{align*}

%     where $\Psi(S) = \prod\limits_{i \in S} \psi_i$.
% \end{condition}

\begin{restatable}[Oblivious-AIS]{theorem}{goldbach}
\label{thm:gen-dynLLL}
    For a fixed update sequence $F_1, F_2, \dots F_q$, \cref{alg:ais} makes at most $s$ resamples with probability at least $1 - \delta$ where,

    \begin{align*}
        \eps s = O\left(\log \max\limits_{\sigma \in \Omega}\frac{\theta(\sigma)}{\mu(\sigma)} + \sum\limits_{i=1}^q \log (1 + \psi_i(F_{i})) + \max\limits_{S \subseteq [q], i\in[q]} \frac{1}{\Psi_i(S)} + \log \delta^{-1}\right) 
    \end{align*}
    and $\Psi_i(S), \psi(F)$ for a subset $S \subseteq [q]$ denotes the \LLL parameters for set $S$ after the $j^{th}$ update, and in particular $\psi_j(F_j)$ is the \LLL parameter for the flaw $F_j$ during the $j^{th}$ update at which point the set of candidate flaws are $\{F_1, F_2, \dots F_j\}$.
\end{restatable}

\begin{remark}
If $\mu$ is uniform, the starting state is deterministic, and all $\psi$ parameters are bounded in $[\psi_{\min}, \psi_{\max}]$, then the number of resamples $s$ is $O\left( \frac{1}{\eps} \cdot (\log|\Omega| + q \log(\frac{1 + \psi_{\max}}{\psi_{\min}}) + \log \delta^{-1}) \right)$.

In most LLL settings, and in particular for all LLL instances satisfying \cref{cond:asymmetric} with a constant $\eps$, the $\log(\frac{1 + \psi_{max}}{\psi_{\min}})$ term is upper bounded by $O(\log q)$.
To see this, let's say \cref{cond:asymmetric} is satisfied with parameters $\eps, x_1, x_2, \dots, x_q$.
To convert it to \cref{cond:general}, one sets $\psi_i = x_i / (1 - x_i)$.
If all $x_i$ are in $[q^{-10}, 1 - q^{-10}]$, then $\log(\frac{1 + \psi_{max}}{\psi_{\min}}) = O(\log q)$.
Otherwise, we set $x_i \gets x_i + q^{-10}$ for $x_i < q^{-10}$ and set $x_i \gets x_i - q^{-10}$ for $x_i > 1 - q^{-10}$.
It can be verified that the resulting parameters satisfy \cref{cond:asymmetric} with slack $\eps - q^{-8} \geq \eps / 2$.
\end{remark}

\begin{restatable}[Adaptive-AIS]{theorem}{adaptiveais}
\label{thm:gen-dynLLL2}
    If every possible update can be expressed using identifiers in $[k]$ to identify the flaws and their LLL parameters, and $\psi$ is an upper bound on the value of the parameter for each flaw after each update, then the probability that the resampling algorithm (\cref{alg:ais}) makes more than $s$ resamples is given by,
    \begin{align*}
        \eps s = O\left( q \log (k \psi) + q \log \psi^{-1} + \log \delta^{-1} \right)
    \end{align*}
\end{restatable}

Observe that \cref{thm:gen-dynLLL} does not depend on flaw removal updates. This is because, during these updates, the state of the algorithm does not change, and the algorithm does not make any new resamples. Additionally, the algorithm is always in a flawless state in between queries.

The rest of the section will be the proof of \cref{thm:gen-dynLLL,thm:gen-dynLLL2}.
We follow the general framework for proving convergence of \LLL resampling algorithms. As is standard in such arguments, we assume that the random bits used by the $\mathtt{Address}(F)$ are obtained by a sufficiently large table $\mathcal{R}$ of random bits whose rows correspond to the values of random variables used by the resampling procedure. The $j^{th}$ column of the table will be used during the $j^{th}$ call of the Address function.

\subsection{Constructing the witness}

First, we define a witness $T$ for every execution of \cref{alg:ais}. The witness will be constructed in such a way that the resampling sequence can be generated completely.  In particular, for every execution of \cref{alg:ais} that makes at least $s$ resamples, given only (i) the witness corresponding to this execution and (ii) the random table $\mathcal{R}$, one can determine for each flaw $F$, the number of times $\mathtt{Address}(F)$ was called and which bits in $\mathcal{R}$ was used by these calls. We say that a table $\mathcal{R}$ is \textit{consistent} with a witness if it is plausible that upon executing \cref{alg:ais} using $\mathcal{R}$, the witness was generated. We say that a witness has weight $p$, if the probability that for a fixed table $\mathcal{R}$, $w$ is consistent with $\mathcal{R}$ with probability at most $p$.
Subsequently, the probability that the algorithm makes $s$ resamples is bounded by the sum of weights all witnesses. The challenge is to design the witnesses and enumerating / bounding the total weight over all possible witnesses.

\paragraph{Breakage Forests.} For a sequence of $s$ resamplings, $F_1, F_2, \dots $ where $F_i$ is the flaw that is fixed during the $i^{th}$ resampling, we define a forest $F$ of ordered trees whose vertices are labelled with the IDs of these flaws. The $i^{th}$ connected component of the forest denotes the resamplings made by \cref{alg:ais} after the $i^{th}$ update to reach a flawless state for this query. Each component tree of the forest is called a \textit{breakage} tree.

The roots of the breakage forest are the flaw updates. Consider the $j^{th}$ breakage tree with root $F_{j,1}$. Let the sequence of resamplings corresponding to this update be $F_{j,1}, F_{j,2}, \dots$. Suppose flaw $F_{j, k}$ for $k > 1$ was introduced by the resampling of flaw $F_{j, i}$ for some $i < k$. Note that there must be a unique such flaw under our definitions of flaw introductions.  We attach $F_{j, k}$ as a child of $F_{j, i}$.

Upon constructing the tree, observe that for a vertex $F_{j, i}$ in the breakage tree, the set of its children $C_i = \{F_{j, k} \mid (F_{j, k}, F_{j, i}) \in T_j\}$ correspond to flaws introduced by the resampling of flaw $F_{j, i}$ that were \textit{not} collaterally fixed. We denote the witness forest by $\phi$. For every node $v \in \phi$, let $\psi(v)$ denote the \LLL parameters corresponding to the node $v$.
The resampling sequence can be obtained by a certain graph traversal of $\phi$. We describe the graph traversal process, and we shall also associate with each resampling $v \in \phi$ a set $S_v$ of flaws that represent the charges for this resampling.

The graph traversal visits each tree in the forest in order. To visit a component, the root is first visited (say $v_1$). We define $S_{v_1} = C_{v_1}$, i.e., the set of flaws corresponding to the children of $v_1$. To compute the next vertex to visit, we take $S_{v_1}$ and compute the next flaw that is fixed using the $\FSelect{}$ procedure. Typically, this is the highest priority flaw in $S_{v_1}$ where the priorities are computed via a permutation $\pi$ of the flaws. The node corresponding to this flaw is $v_2$, the next vertex to be visited. We compute $S_{v_2} \gets (S_{v_1} - F_2) \cup C_{v_2}$. We repeat the graph traversal until all the vertices are visited. Note that every vertex $v \in \phi$ receives a (possibly empty) set $S_v$.

\begin{claim}
    Given a witness forest $\tau$, the resampling sequence can be reconstructed.
\end{claim}
\begin{proof}
    Assume that we can determine the resampling sequence up to $j$ updates for some $j \geq 0$. The next resampling is the root of the next component of the forest. Suppose we know a prefix of the resampling sequence $F_{j,1}, F_{j,2} \cdots F_{j, i}$ for some $i$. The unvisited boundary of these vertices in the tree contain the set of known flaws at this stage of the algorithm. By construction, the next flaw resampled is not collaterally fixed and must be in this boundary. Moreover, the deterministic rule must select this flaw since we consider permutation based rankings. Hence, the flaw selection rule with the currently known set of flaws uniquely determines the next flaw that is resampled.
\end{proof}

\subsection{Weight of the witness tree}

For a given breakage tree $\phi$, we compute the probability that $\phi$ could be generated by \cref{alg:ais}.

\begin{lemma}
\label{lem:witness-prob}
    For every plausible witness forest $\phi$, 
    \begin{align*}
        \Pr[\phi \text{ is generated by }\cref{alg:ais} ] \leq \left(\max\limits_{\sigma \in \Omega} \frac{\theta(\sigma)}{\mu(\sigma)} \right)\prod\limits_{v \in \phi} \gamma^{S_v}_v
    \end{align*}
\end{lemma}

\def\cS{\mathcal{S}}

\begin{proof}
    Let $\sigma_t$ be the random variable denoting the state in $\Omega$ that the algorithm is in after the $t^{th}$ resample. We let $\sigma_0 \sim \theta$, i.e., the initial state is drawn from the distribution $\theta$.

    From the breakage forest $\phi$, we can reconstruct the resampling sequence by performing the graph traversal described above. Let $F_1, F_2, \dots F_t$ denote the resampling sequence and $v_1, v_2, \dots v_t$ denote the nodes of the forest corresponding to these resamplings.

    We prove by induction that,
    \begin{align*}
        \Pr[\sigma_{t+1} = \tau] &= \left(\max\limits_{\sigma \in \Omega} \frac{\theta(\sigma)}{\mu(\sigma)}\right) \cdot \prod\limits_{i=1}^t \gamma^{S_{v_i}}_{v_i} \mu(\tau)
    \end{align*}
    The result then follows by adding over all $\tau \in \Omega$.
    
    As the base case of the induction, consider $t = 0$. The statement is trivial as $\Pr[\sigma_0 = \tau] = \theta(\tau) \leq \left(\max\limits_{\sigma \in \Omega} \frac{\theta(\sigma)}{\mu(\sigma)} \right) \mu(\tau) $ by definition.

    For every $t > 0$, by the induction hypothesis, we have
    \begin{align*}
        \Pr[\sigma_{t-1} = \sigma'] &\leq \left(\max\limits_{\sigma \in \Omega} \frac{\theta(\sigma)}{\mu(\sigma)} \right) \cdot \prod\limits_{i=1}^{t-1} \gamma^{S_{v_i}}_{v_i} \cdot  \mu(\sigma') \\
        \Rightarrow \Pr[\sigma_t = \tau] &\leq \left(\max\limits_{\sigma \in \Omega} \frac{\theta(\sigma)}{\mu(\sigma)} \right) \cdot \sum\limits_{\sigma' \in \Inc^{S_{v_i}}_{F_i}(\tau)} \Pr[\sigma_{t-1} = \sigma'] \cdot \rho^{S_{v_i}}_{v_i}(\sigma', \tau) \\
        &\leq \left(\max\limits_{\sigma \in \Omega} \frac{\theta(\sigma)}{\mu(\sigma)} \right) \cdot \prod\limits_{i=1}^{t} \gamma^{S_{v_i}}_{v_i} \cdot  \mu(\tau)
    \end{align*}
\end{proof}

\subsection{Coupling with a forest generation process}
We would like to apply the union bound with \cref{lem:witness-prob} over all possible witness sequences. Let $\cF_t$ denote the set of forests of size $t$ that correspond to valid resamplings of length $t$. For a node $v$ in some witness forest, let $F_v$ denote the corresponding flaw. The union bound gives,
\begin{align*}
    \Pr[\cref{alg:ais} \text{ makes at least } t \text{ resamplings }] \leq \left(\max\limits_{\sigma \in \Omega} \frac{\theta(\sigma)}{\mu(\sigma)}\right)\sum\limits_{\phi \in \cF_t} \prod\limits_{v \in \phi } \gamma^{S_{v}}_{F_v}
\end{align*}

To bound the summand on the right, we consider the following random process that generates witness forests (which is independent of the algorithm). We consider a fixed update sequence $F_1, F_2, \dots F_q$ to the algorithm. The process is described by random sets of flaws, $R_0, R_1, \dots R_t$ defined inductively. Each $R_k$ is associated with either (i) a resampling of the algorithm or (ii) an update in the algorithm.
We will also iteratively compute sets $S_0, S_1, S_2, \dots S_t$ corresponding to the incoming flaw sets. $S_0 = \emptyset$.

$R_0 = \emptyset$. For $k \geq 0$, if $R_{k} = \emptyset$ and there are still unprocessed updates, then $R_{k + 1} = \{F_i\}$ or $\emptyset$ depending on whether the resampling of $F_i$, the next unprocessed update, is flawed. In these cases, we set $S_{k+1} \gets R_{k+1}$. Flaw deletions are explicitly marked. If $R_{k} \neq \emptyset$, then let $k$ correspond to the $i^{th}$ update with LLL parameters $\psi, \gamma$ (the update number is omitted for convenience). Let $F_k$ denote the highest priority flaw in $S_k$. For more general flaw selection rules, $F_k$ is computed by calling $\mathtt{Select}({S_k, k})$. We set $R_{k+1}$ according to the following probability distribution, 
\begin{align*}
    \Pr[R_{k+1} = S] = \frac{\gamma^{S}_{F_k} \Psi(S)}{\sum\limits_{S' \subseteq [m]} \gamma^{S'}_{F_k} \Psi(S')} 
\end{align*}

The set $S_{k+1}$ is computed as $S_{k+1} \gets (S_k - F_k) \cup S$.

Recall, the convergence criterion,
\begin{align}
\label{eqn:conv-criterion-ais}
    \zeta_i = \frac{1}{\psi_i}\sum\limits_{S \subseteq [m]} \gamma^{S}_{i} \Psi(S) < 1 - \eps \ \ \ \ \ \forall \ \ i \in [m]
\end{align}
where $m$ is the number of flaws in the current query.

\begin{claim}
\label{clm:prob-phi}
    For a given sequence of updates, the random process outputs a plausible witness forest $\phi$ whose (i) flaw resampling sequence is given by $F_1, F_2, \dots F_t$, (ii) incoming flaw sets is given by $S_1, S_2, \dots S_t$ and (iii) ordered roots are given by $v_1, v_2, \dots v_{q}$ with probability,
    \begin{align*}
        p_{\phi} = \prod\limits_{v \in \phi}  \frac{\gamma^{S_{v}}_{F_v} \Psi(S_v)}{\zeta_{F_v} \psi_{F_v}} = \frac{\Psi(S_{t+1})}{\prod\limits_{i=1}^q\psi(F_{v_i})} \cdot \prod\limits_{v \in \phi} \frac{\gamma^{S_{v}}_{F_v}}{\zeta_{F_v}}
    \end{align*}
    where $S_{t+1}$ is the set of flaws present after $t$ resamplings of the algorithm. 
\end{claim}
\begin{proof}
    The first equality can be obtained from \cref{eqn:conv-criterion-ais} and multiplying the probabilities for each step. Note that, by definition of our random process, these events are independent.

    For the second equality, we need to simplify $\prod\limits_{v \in \phi} \frac{\Psi(S_{v})}{\psi(F_v)}$. Here, the numerator $\prod\limits_{i=1}^t \Psi(S_{i, j})$ includes the product of $\psi$ values of all flaws that were not fixed collaterally, i.e., were fixed by one of the resamplings done in the algorithm \textbf{except} those created due to insertion events and those that remain flawed at the end. For the first set of flaws, during its resampling, the contribution of that flaw to the numerator is cancelled out by the corresponding $\psi$ factor in the denominator. The only unaccounted for flaws are those in the denominator corresponding to the flaws created due to the insertion updates, and (ii) those that remain flawed after $t$ resamplings. Thus, this entire product simplifies to $\frac{\Psi(S_{t+1})}{\prod\limits_{j=1}^q \psi(F_{v_j})}$.
\end{proof}

% \begin{theorem}
% \label{thm:dyn-obliv-seq}
%     For a fixed update sequence $F_1, F_2, \dots F_q$, \cref{alg:ais} makes at most $s$ resamples with probability at least $1 - \delta$ where,

%     \begin{align*}
%         \eps s = O\left(\log \max\limits_{\sigma \in \Omega}\frac{\theta(\sigma)}{\mu(\sigma)} + \sum\limits_{i=1}^q \log (1 + \psi_i(F_{i})) + \max\limits_{S \subseteq [q], i\in[q]} \frac{1}{\Psi_i(S)} + \log \delta^{-1}\right) 
%     \end{align*}
%     and $\Psi_i(S), \psi(F)$ for a subset $S \subseteq [q]$ denotes the \LLL parameters for set $S$ after the $j^{th}$ update, and in particular $\psi_j(F_j)$ is the \LLL parameter for the flaw $F_j$ during the $j^{th}$ update at which point the set of candidate flaws are $\{F_1, F_2, \dots F_j\}$.
% \end{theorem}

\goldbach*

\begin{proof}
    \def\bR{\mathbf{R}} 
    Let $\bR$ denote the ordered set of non-trivial roots of the breakage forest. We say that a root of the breakage forest is trivial if its corresponding connected component is a single vertex. Such components correspond to either deletions or updates that did not cause any resamples.
    Let $\cF_{t, \bR}$ denote the set of all witness forests of size $t$ that are plausible for the given update sequence, and whose ordered set of non-trivial roots is $\bR$. We have $\sum\limits_{\phi \in \cF_t} p_\phi \leq 1$ since the random process outputs at most one witness.
    From the above inequality and using $\zeta_{F_v} \leq (1 - \eps)$ we have, 
    $$\sum\limits_{\phi \in \cF_{t, \bR}} \prod\limits_{v \in \phi} \gamma^{S_v}_{F_v}\leq (1 - \eps)^s \cdot \max\limits_{S \subseteq [q]}\frac{ \prod\limits_{v \in \bR} \psi(F_{v})}{\Psi(S)}$$.
    % Note that in the above equation $v_i$ denotes the roots of the witness trees, i.e., the update sequence.

    The union bound of the probability that at least one of the witness forests occurs is,

    \begin{align*}
        \left(\max\frac{\theta(\sigma)}{\mu(\sigma)}\right) \sum\limits_{\bR \subseteq [q]}\sum\limits_{\phi \in \cF_{t,\bR}} \prod\limits_{v \in \phi} \gamma^{S_v}_{F_v} &\leq (1 - \eps)^s \left(\max\frac{\theta(\sigma)}{\mu(\sigma)}\right) \cdot \max\limits_{S \subseteq [q]}\frac{ \sum\limits_{\bR \subseteq [q]}\prod\limits_{v \in \bR } \psi(F_{v})}{\Psi(S)}  \\
        &\leq (1 - \eps)^s \left(\max\frac{\theta(\sigma)}{\mu(\sigma)}\right) \cdot \max\limits_{S \subseteq [q]}\frac{\prod\limits_{v \in \bR } (1 + \psi(F_{v}))}{\Psi(S)}
        \end{align*}
    
\end{proof}

% \begin{theorem} 
% \label{thm:gen-dynLLL2}
%     If the total number of distinct flaws that are added is at most $m$ and the flaws are numbered $F_1, F_2, \dots F_m$ and $\psi(F_i)$ is an upper bound on the LLL parameter of the flaw $F_i$ over all possible flaw additions, \cref{alg:ais} makes at most $s$ resamples with probability at least $1 - \delta$ where,
%     \begin{align*}
%         \eps s = O\left( \log \left(\max\limits_{\sigma \in \Omega} \frac{\theta(\sigma)}{\mu(\sigma)} \right) +  (\sum\limits_{j=1}^q \log_2{(1 + \sum\limits_{i=1}^m  \psi(F_i)}) +  \max\limits_{S \subseteq [m]} \log_2 \frac{1}{\Psi(S)} + \log_2 \delta^{-1} \right)
%     \end{align*}
%     resamples. In particular, if it holds that $\psi(F_i) \leq \psi$ for every flaw $F_i$ and every instance of \LLL then, 
%     $$\eps s = O\left(q \log (m \psi) + m \log \psi + \log_2(\delta^{-1}) \right)$$
% \end{theorem}
\adaptiveais*

\begin{proof}
    We union bound the probabilities that $(F_1, F_2, \dots F_q)$ occur over all the $k^q$ possible choices of flaw updates.
    Summing over the probability of witnesses in the proof of \cref{thm:gen-dynLLL} over all such sequences we require to compute, we obtain
    $\sum\limits_{(F_1, F_2, \cdots F_q)} \prod\limits_{i=1}^q \psi_{(F_1,F_2,\dots F_i)}(F_i)$. Here we denote by $\psi_{(F_1, F_2, \dots F_i)} (F_i)$ the value of $\psi(F_i)$ after the sequence of updates $(F_1, F_2, \dots F_i)$. If we have a global upper bound $\psi$ for every such $F_i$, then the expression can be simplified into the following,
    $\prod\limits_{i=1}^q(1 + kq)$. In other cases, one has to manually compute the above summation. Plugging in the above term, we get the desired bound.
\end{proof}

% \noindent (1) Static general framework: We choose $\psi_i = e\gamma_i = e/2^{k}$. Number of resampling is dominated by the $\max_S \log{\frac{1}{\Psi(S)}}$ term. $\Psi(S) = \prod_{i\in S} \psi_i = (e/2^{k})^{|S|} \geq (e/2^{k})^q$. So $\log{\frac{1}{\Psi(S)}}$ can be as large as {\color{red}$kq\neq \tO(q)$}.

% \noindent (2) Static MT framework: We choose $x_i = \frac{e\gamma_i}{1+e\gamma_i}$. Number of resampling is $\sum_{i=1}^q \frac{x_i}{1-x_i} = \sum_i e \gamma_i = {\color{red}q \cdot 2^{-k}}$.

% \noindent (3) Forward tree for symmetric LLL: $q \log q / \eps$ resampling

% \noindent (3) For triangle-free coloring, $\psi_i$ can be as small as $e^{-L}$, and that term can be $qL = O(\Delta^\eps \cdot q)$.
% }

% \begin{theorem}
% \label{thm:gen-dynLLL2}
%     If the total number of update sequences is at most $2^m$, and all \LLL parameters satisfying \cref{cond:general} with $\eps$ slack are bounded by $\psi$, then the probability that \cref{alg:ais} makes at least $s$ resamples is at most $\delta$ where,
%     \begin{align*}
%         s = O\rb{\frac{1}{\eps} \left( q \log (kq) \log \psi^{-1} + \log_2 \delta^{-1} \right)}
%     \end{align*}
% \end{theorem}

\section{Applications}
\label{sec:applications}
In this section, we introduce several applications of our algorithms.
The description of our algorithms consists of the following key components.
\begin{itemize}
    \item \textbf{Flaws and resampling procedures:}
    One should specify the state space $\Omg$, the sequence of flaw updates (insertions or deletions) $F_1, F_2, \dots, F_q$, and the permutation $\pi$.
    Each flaw $F_i$ should be associated with a resampling procedure.
    % {\color{red}The resampling procedure can depend on the current set of flaws.}
    % In the variable setting, the permutation $\pi$ can be omitted.
    
    \item \textbf{Proof of convergence:} A proof that after each flaw update, \cref{cond:general} is satisfied.
    Let $\cF_j$ be the set of flaws after the first $j$ updates.
    One needs to show that there exists a constant $\eps > 0$ and positive parameters $\{\psi_j(F_i)\}_{i,j \in [q]}$, such that for each $j = 1, 2, \dots, q$, $\cF_j$ satisfies \cref{cond:general} with parameters $(\eps, \psi_j(F_1), \psi_j(F_2), \dots, \psi_j(F_q))$.
    
    \item \textbf{Implementation details:} One should specify the following implementation details:
    \begin{itemize}
        \item Procedure $\algResample(f, \sigma)$:
        This procedure implements the resampling procedure.
        The input is the current state $\sigma$ and the flaw $f = \pi(\sigma)$.
        The procedure outputs two objects:
        (1) a state $\sigma'$ obtained by addressing $f$ at $\sigma$, sampled from the distribution $\rho_f$, and (2) a list $L'$ of flaws introduced by the resampling.
        
        % \item Procedure $\algFind(\sigma, \sigma', f)$: The input are states $\sigma, \sigma'$ and flaw $f$. The procedure is only called after we address the flaw $f$ at $\sigma$ and move to $\sigma'$. The output is the list of all flaws introduced by the flaw-addressing step.
        \item Procedure $\algCheck(f, \sigma)$: Check whether a state $\sigma$ has a flaw $f$.
        
        \item Procedure $\algCompare(f_1, f_2)$: Given two flaws $f_1, f_2$, compare their order in $\pi$.
        
        \item The representation of a flaw and a state. Typically, a state is represented using an $O(\log|\Omega|)$-bit ID. A flaw is represented by an $O(\log q)$-bit ID and access to $\algResample(f, \cdot)$, $\algCheck(f, \cdot)$, and $\algCompare(f, \cdot)$.

        \item \textbf{Remark:}
        A technical detail about flaw IDs and resampling procedures is as follows.
        In our definition, we assume that each flaw is associated with the oracle access to its resampling procedure.
        This oracle is often implemented assuming access to an auxiliary data structure.
        For the example of triangle-free coloring (\cref{sec:triangle-free}), the oracle assumes access to the input dynamic graph.
        
        However, in \cref{sec:dynamic-setup}, the resampling procedure should be determined only by the current state and the flaw IDs of all previous updates, and not an auxiliary data structure.
        This inconsistency can be fixed by adding dummy flaws:
        Whenever we update the data structure in $O(T)$ time, we add a dummy flaw whose ID encodes the update in $O(T)$ bits.
        The flaw contains an empty set of states and will never be introduced or compared.
        It is only added to represent the additional information used in the resampling.
        For example, if the data structure is a dynamic graph $G$, then a dummy flaw is added whenever $G$ receives an edge update, and the ID of the flaw represents the edge.
        In total, only $q$ dummy flaws are added, which does not affect the time complexity.
    \end{itemize}

    \item \textbf{Analysis for clairvoyant adversary:}
    All our algorithms in this section assume a clairvoyant adversary.
    Our analysis will first show a guarantee on any fixed update sequence using \cref{thm:gen-dynLLL}.
    The guarantee will be extended against clairvoyant adversaries by union bounding over all possible sequences.
\end{itemize}

Implementation details are outlined in \cref{alg:app,alg:app-insert}.
We summarize the time complexity in \cref{thm:application}.

\begin{algorithm}  
    \Fn{\FDynamic} {
        \KwIn{state space $\Omg$, permutation $\pi$, sequence of $q$ flaw updates chosen by an oblivious adversary}
        \KwOut{a flawless state $\sigma$ after each update}
        $\sigma \gets \text{any state in $\Omg$}$ \DontPrintSemicolon \tcp*{$\sigma$ is flawless since there are no flaws} \label{line:intialize}
        \For{each flaw update}{
            \If{the update is an insertion of a flaw $f$}{
                $\FInsert(\sigma, f)$ \label{line:insert}
            }
        }
    }
    \caption{Outline of the dynamic algorithm.} 
    \label{alg:app}  
\end{algorithm}

\begin{algorithm}  
    \Fn{\FInsert} {
        \KwIn{current state $\sigma$, a new flaw $f$}
        \KwEnsure{no flaw other than $f$ exists}
        % Sample $B$ such that except possibly $B$, none of the other flaws exist.
        $Q \gets $ a priority queue of flaws, initially containing only $f$ \\ 
        \While{$Q$ is non-empty}{
            $f' \gets $ the flaw in $Q$ with the highest priority in $\pi$ \\
            \If{$\sigma$ has the flaw $f'$} { \label{line:check} 
                Update $\sigma$ by resampling $f'$ \label{line:resample} \\
                $F \gets$ the set of flaws introduced in the last step \label{line:find-flaw} \\
                $Q \gets Q - \{f'\} \cup F$ 
            }
        }
    }
    \caption{Implementation details for \cref{alg:ais}.} 
    \label{alg:app-insert}  
\end{algorithm}

\begin{theorem} \label{thm:application}
    Let $\Omg, \pi, F_1, F_2, \dots, F_q$ be an instance of \cref{alg:app} such that \cref{cond:general} is satisfied after each update with parameters $\eps$ and $\{\psi_j(F_i)\}_{i,j \in [q]}$.
    Assume further that each flaw can be represented with $\tilde{O}(1)$ bits, each state can be represented in $O(\log \Omega)$ bits, and $\algCompare$, $\algResample$, $\algCheck$ can be implemented in time $\TCompare, \TResample, \TCheck$ time, respectively.
    Then, the following holds:
    \begin{itemize}
        \item[] (A1) Let $d = \max_{i \in [q]} |\Gamma_i|$, $\delta \in (0, 1)$ be any real number, and $R = (\log |\Omega| + q\log(\frac{1 + \psi_{\max}} {\psi_{\min}}) + \log \delta^{-1}) / \eps$. Against an oblivious adversary, \cref{alg:app} handles all updates using $O(R)$ calls to $\algResample$, $\tO(R \cdot d)$ calls to \algCompare and \algCheck, and $\tO(R\cdot d)$ additional time with probability at least $1 - \delta$.
        \item[] (A2) Let $\zeta$ be an upper bound on the number of distinct update sequences that an adversary can choose from, and $R' = (\log |\Omega| + q\log(\frac{1 + \psi_{\max}} {\psi_{\min}}) + \log \zeta) / \eps$.
        Against a clairvoyant adversary, \cref{alg:app} handles all updates using $\tO(R' \cdot \TResample + R' \cdot d \cdot (\TCompare + \TCheck))$ time with probability at least $1 - 2^{-q}$.
    \end{itemize}
    % Then, with probability at least $1 - \delta$, \cref{alg:app} handles all flaw updates with $\tO(T)$ calls to $\algResample$, $\tO(TD)$ calls of $\algCompare$ and $\algCheck$, and $\tO(T)$ additional time, where $D = \max_{i \in [q]} |\Gamma_i|$ and $T =  \log_2 |\Omega| + q\log_2(\frac{1 + \psi_{\max}} {\psi_{\min}}) + \log_2 \delta^{-1}$.
\end{theorem}
\begin{proof}[Proof of (A1)]
    Fix a $\delta \in (0, 1)$.
    We first upper bound the time for executing \cref{alg:app-insert} for all updates.
    Note that \cref{line:resample,line:find-flaw} of \cref{alg:app-insert} can be implemented by invoking \algResample, and \cref{line:check} can be done by calling \algCheck.
    We use \algCompare to maintain the order of flaws in $Q$.
    By \cref{thm:gen-dynLLL}, \cref{line:resample} is executed $O(R)$ time in total.
    Since each resampling introduces at most $d$ flaws, at most $O(Rd)$ flaws are inserted into $Q$, and hence the while-loop has at most $O(Rd)$ iterations.
    Hence, in total \algCheck is called $O(Rd)$ times.
    By a standard implementation of priority queues, the while loop can be executed by $O(Rd \log Rd) = \tO(Rd)$ calls to $\algCompare$ and $\tO(Rd)$ additional time.
    Over all updates, we spend $\tO(Rd + R \cdot \TResample + Rd \cdot \TCheck + Rd \cdot \TCompare)$ time on \cref{alg:app-insert}.

    The time complexity of \cref{alg:app} is determined by \cref{line:intialize,line:insert}, which take $O(\log_2 |\Omega|)$ and $\tO(R \cdot \TResample + Rd \cdot \TCheck + Rd \cdot \TCompare)$ time, respectively.
    This proves (A1).

    \medskip
    \noindent \emph{Proof of (A2).}
    We can view \cref{alg:app} as running on an infinite sequence of random real numbers $S = s_1, s_2, \dots$, where $s_i \in [0, 1]$.
    At each step $t$, the algorithm first selects a flaw $F(t)$ from the current state $\sigma$ to address.
    The algorithm then computes the new state as a random sample from the distribution $\rho_{F(t)}(\sigma, \cdot)$, generated using $s_t$.
    
    Using this model, one can rephrase (A1) as follows.
    Suppose that we draw an infinite sequence of random real numbers $S$ and run \cref{alg:app} on it.
    For any fixed update sequence and $\delta \in (0, 1)$, the running time of \cref{alg:app} is $\mathcal{T'} = \tO(R \cdot \TResample + R \cdot d \cdot (\TCompare + \TCheck))$ with probability at least $1 - \delta$.

    By applying a union bound over all $\zeta$ possible update sequences, we obtain that the following event happens with probability $1 - \delta \cdot  \zeta$: The running time of \cref{alg:app}, when run on $S$, is upper bounded by $\mathcal{T'}$ for all update sequences.
    Note that if the above event happens, then no adversary can make \cref{alg:app} run more than $\mathcal{T'}$ time, since no update sequence can achieve it.
    Setting $\delta = \frac{1}{\zeta \cdot 2^q}$ then gives the result in (A2).
\end{proof}

% In most LLL settings, $\Omega$ is a product space on $n$ variables, $d$ is an upper bound on the dependency degree, and there is an upper bound $k = \max_{i \in [q]} |\vbl(F_i)|$ on the number of variables for a flaw, and all of \algResample, \algCheck, \algCompare can be implemented in $\poly(D, k)$ time.
% Hence, \cref{thm:application}(A1) gives an $\tO(\poly(k, D))$ amortized update time with an exponential tail.
To use (A2), a trivial way is to set $\zeta = (2^{|\Omega| + 1})^q$, because there are $q$ updates, and each update can remove or add any flaw $f \in 2^\Omega$.
However, the resulting time complexity will be linear in $|\Omega|$, which is often exponential in $n$.
In most cases, we can show a better complexity by using an application-specific upper bound on $\zeta$.

In the remainder of this section, we provide application-specific details for the dynamic CNF, triangle-free coloring, LMR-scheduling, and edge coloring problems.

% \btodo{ Note: here we hide $\log q + \log(\frac{1 + \psi_{\max}}{\psi_{\min}})+ \lg\lg(|\Omega|)$ in $\tO(1)$. }

\section{Dynamic CNF problem} \label{sec:CNF}

\paragraph{Notations and definitions.}
Let $v_1, v_2, \dots, v_n$ denote $n$ Boolean variables.
A \emph{literal} is either a variable $v_i$ or its negation $\neg v_i$.
A \emph{clause} is a disjunction of literals.
W.L.O.G., assume that all literals in a clause correspond to different variables.
A \emph{CNF} is a conjunction of one or more clauses.
% For example, $(v_1 \vee \neg v_2) \wedge (v_1 \vee v_3 \vee v_4) \wedge(\neg v_1 \vee \neg v_2)$ is a CNF of three clauses.
%        A \emph{$k$-CNF} is a CNF in which each clause has exactly $k$ literals.
% An assi gnment of the variables \emph{satisfies} a CNF if it makes the formula evaluate to true.
Let $\phi$ be a CNF.
For each clause $C$ in $\phi$, let $N(C)$ denote the set of clauses in $\phi$ that share variables with $C$.
The \emph{dependency degree} of a clause $C$ is defined as $|N(C)|$.
% The \emph{dependency degree} of $\phi$ is defined as $\max_{C \in \phi} |N(C)|$.
% An \emph{assignment} is a function that assigns each variable with either 0 (false) or 1 (true).
% Given a CNF $\phi$, the \emph{CNF problem} asks to find an assignment that satisfies $\phi$, or determine that such an assignment does not exist.
Two important parameters are often considered: The maximum degree $\Delta = \max_{C \in \phi} |N(C)|$ and the maximum number $k$ of variables in a clause.

The CNF problem is perhaps the simplest one falling into the Moser-Tardos framework.
In particular, if every clause has exactly $k$ variables, and the maximum degree is less than $2^k/e$, then the CNF is satisfiable, and \cite{moser2010constructive}'s algorithm finds a satisfying assignment in polynomial time.
We generalize such CNFs by allowing variable length of clauses, as follows.

\begin{definition}[CNF with bounded dependence] \label{def:CNF-bounded}
For a clause $C$ of $x$ literals, define $p(C)$ as $2^{-x}$.
That is, $p(C)$ is the probability that a random assignment violates $C$.
A CNF is of \emph{bounded dependence} if there exists a positive $\eps$ such that $\sum_{D \in N(C)} p(D) \leq 1/e - \eps$ holds for every clause $C$.
\end{definition}

\begin{definition}[Dynamic CNF problem with bounded dependence (DCNF)]
Let $\phi$ be an initially empty CNF on $n$ variables $v_1, v_2, \dots, v_n$.
The input is a sequence of $q$ insertions and deletions of clauses on $\phi$.
Throughout, it is guaranteed that $\phi$ is of bounded dependence for a fixed constant $\eps$.
The goal is to maintain a satisfying assignment after each update on $\phi$.
Throughout, the maximum degree never exceeds $\Delta$, and each clause contains at most $k$ variables.
% After a fixed update, the \emph{dependency degree} of a clause $C$ is the current number of clauses in $\phi$ that share variables with $C$.
\end{definition}

\paragraph{Algorithm specification.}
We now use our algorithmic framework to solve DCNF.
% We note that the parameters $k, \Delta, \eps$ are only for the analysis.
% The algorithm does not assume the knowledge of them.
% \paragraph{Flaw and resamplings.}
Let $\Omg = \{0, 1\}^n$ be the set of truth assignments to $v_1, v_2, \dots, v_n$.
Denote by $C_t$ the clause inserted or deleted in the $t$-th update, where $t = 1, 2, \dots, q$.
For each inserted clause $C_t$, define a flaw $F_t$ to be the set of assignments that violate the $C_t$.
Thus, we can map the update sequence into a sequence of flaw updates: When a clause $C_t$ is inserted, we insert $F_t$, and the mapping for deletions is symmetric.
At any moment, any flawless state with respect to the current set of flaws satisfies all clauses.
The permutation $\pi$ can be arbitrary.

Each flaw $F_j$ is associated with the MT-resampling, i.e., to resample all variables in $C_j$ uniformly at random.
The algorithm is hence an instance of the variable setting with MT-resampling (\cref{def:variable}).

\paragraph{Convergence proof.} 
We now show that the charges satisfy the local union bound condition (\cref{cond:local-union}) after each update.
Consider the CNF instance after a fixed update, where each clause $C_i$ is associated with a flaw $F_i$.
By \cref{lem:prod}, for each flaw $F_i$, the charge $\gamma_i = \Pr[F_i] = p(C_i)$.
% Note that $\gamma_i$ is invariant throughout the updates.
The set $\Gamma_i$ is exactly $N(C_i)$, because the resampling of $F_i$ can introduce $F_j$ only if $C_i$ and $C_j$ share variables.
Therefore, for each flaw $F_i$, it holds that
\[
    \sum_{j \in \Gamma_i} \gamma_j = \sum_{j \in \Gamma_i} p(C_j) = \sum_{C_j \in N(C_i)} p(C_j) \leq 1/e - \eps,
\]
where the last inequality is due to \cref{def:CNF-bounded}.
This shows that the local union bound condition is satisfied after each update.
By our discussion in \cref{sec:condition}, the local union bound condition can be converted to \cref{cond:general} by setting $\psi_i = e \cdot \gamma_i = e \cdot p(C_i)$.

\paragraph{Implementation details.}
In the following, we provide details and invoke \cref{thm:application} to prove the time complexity.
At any moment, we store the current CNF $\phi$ in the memory, which takes $O(k)$ time per update.
A flaw $F_j$ is represented by its index, that is, $j$.
% The current state $\sigma$ is represented by an $n$-bit binary string, where the $i$-th bit is the assignment for $v_i$.

To implement $\algResample$ efficiently, we maintain the following auxiliary data structures:
For each clause $C_i \in \phi$, we maintain all clauses in $N(C_i)$, represented as a balanced binary search tree (BBST).
For each variable $v_j$, we store a BBST $L_j$ of clauses that contain $v_j$.
% Note that the size of each $N(C_i)$ and $L_j$ is at most $\Delta$.
% Thus, inserting a clause into one of the lists takes $O(\Delta)$ time.
% When a clause is added to one of the lists, we store a handle to its location in the list so that the clause can later be removed in $O(1)$ time.
% We also store a dynamic table $A$ of $q$ bits used for temporary storage.
% We maintain the invariant that all bits in $A$ are zeroes after each update. 
% When an update comes, we first append a zero to $A$ so that its size matches the number of queries.
If the update inserts a clause $C_i$, we add $C_i$ into $L_j$ of each variable $v_j \in C_i$.
This step requires $\tO(k)$ time.
Then, we compute $N(C_i)$ as the union of all clauses in $\bigcup_{v_j \in C_i} L_j$.
% To compute this union, we examine each clause $C_x \in \bigcup_{v_j \in C_i} L_j$, and if $A[x] = 0$, we set $A[x]$ to be $1$ and add $C_x$ to $N(C_i)$.
% After this step, we reset $A$ as all zeroes by iterating through all clauses $C_x \in \bigcup_{v_j \in C_i} L_j$ and set $A[x]$ as 0. 
This can be done in $\tO(\sum_{v_j \in C_i} |L_j|) = \tO(k\Delta) $ time.
After $N(C_i)$ is computed, we add $C_i$ to each $N(C_j)$ such that $C_j \in N(C_i)$.
When a clause $C_i$ is removed, we remove it from all $N(C_x)$ with $C_x \in N(C_i)$ and all $L_j$ with $v_j \in C_i$.
This step removes $C_i$ from at most $\Delta + k$ lists and takes $\tO(\Delta + k)$ time.
Therefore, all auxiliary data structures require $\tO(k\Delta)$ time per update.

The implementation of the key procedures is described as follows.
\begin{itemize}
    \item Given two flaws $F_i, F_j$, $\algCompare(F_i, F_j)$ can be done in $O(1)$ time by returning the flaw with smaller index.
    \item $\algCheck(\sigma, F_i)$ can be done in $O(k)$ time by checking whether $C_i$ is violated.
    \item Given the current state $\sigma$ and a flaw $F_j$, $\algResample$ can be implemented in $\tO(k\Delta)$ time as follows:
    First, we compute the new state $\sigma'$ by resampling all variables in $C_i$.
    Next, we compute the list of flaws introduced.
    To do this, we examine each clause in $N(C_i)$ and check whether it is now violated.
    Examining a clause takes $O(k)$ time, and at most $\Delta$ clauses are examined.
    Therefore, $\algResample$ takes $\tO(k\Delta)$ time.
\end{itemize}

In DCNF, the number of possible update sequences is $\zeta = (2n^k)^q$.
By applying (A2) of \cref{thm:application} with $\zeta = (2n^k)^q, |\Omega| = 2^n, d = \max_{i \in [q]}|\Gamma_i| =  \Delta, \psi_{\min} = \min_{C_i \in \phi} e \cdot p(C_i) = \frac{e}{2^{k}}, \psi_{\max} = e, \TResample = \tO(k\Delta), \TCheck = O(k), \TCompare = O(1)$, the total running time is $\tO((n + qk) \cdot k \Delta)$.
Hence, we obtain the following.

\begin{theorem}
DCNF can be solved in $\tO(k^2 \Delta)$ amortized update time with probability $1 - 2^{-q}$ against a clairvoyant adversary.
% , where $k$ is the maximum number of variables in a clause and $\Delta$ is the maximum number of clauses that a clause can share variables with.
The algorithm does not assume the knowledge of $k, q,$ and $\Delta$.
\end{theorem}

\section{Dynamic \texorpdfstring{$O(\Delta / \ln \Delta)$-}{}coloring of triangle-free graphs} \label{sec:triangle-free}

Let $G = (V, E)$ be a graph with maximum degree $\Delta$.
The list chromatic number of $G$ is the smallest positive integer $\ell$ satisfying the following:
For any assignment of color lists of size $\ell$ to each vertex, it is possible to obtain a proper coloring by giving each vertex a color from its list.
A graph is \emph{triangle-free} if it does not contain a clique of size 3 as a subgraph.
\cite{molloy2019list} showed that every triangle-free graph has list chromatic number at most $(1+o(1))\Delta/\ln \Delta$, matching a $0.5\Delta/\ln\Delta$ lower bound (see, e.g., \cite{bollobas1978chromatic,kim1995brooks}) up to a constant factor.
More precisely, it was shown that for any $\eps > 0$ there exists an integer $\Delta_\eps > 0$, such that any triangle-free graph of maximum degree $\Delta \geq \Delta_\eps$ has chromatic number at most $(1+\eps) \cdot \Delta / \ln \Delta$.

\cite{molloy2019list} also gives an algorithm that finds such a coloring efficiently.
In this section, we extend their result by obtaining a fully dynamic algorithm that maintains a list-coloring with $6\Delta / \ln\Delta$ colors under insertions and deletions of edges.

\begin{remark}
Several constants in our analysis are not the tightest possible.
In particular, it is possible to instead maintain a $((2+o(1)) \Delta / \ln \Delta)$-coloring by carefully choosing the constants in our analysis.
Furthermore, by using the more complicated flaws in \cite{achlioptas2019beyond}, one can maintain a $(1+o(1))\Delta / \ln\Delta$ list-coloring even when the graph contains a bounded number of triangles.
The primary objective of this paper is to provide a framework for adapting existing LLL-based algorithms to the dynamic setting; therefore, we do not aim to obtain the best possible constants.
\end{remark}
% We first present a simpler version which colors the graph with $(2+o(1))\Delta/\ln\Delta$ colors.

\paragraph{Problem formulation.}
% The input consists of parameters $\eps, \Delta, f > 0$ and a graph $G = (V, E)$ of $n$ vertices.
% We assume that $\Delta \geq \Delta_\eps$ for some $\Delta_\eps$ depending only on $\eps$, and $f \in [\Delta^{\frac{2+2\eps}{1+2\eps}}(\ln \Delta)^2, \Delta^2]$.
% Initially, $G$ has no edges, and each vertex $v$ is associated with a color list $C_v$ of size $\ell = (2+\eps)\Delta / \ln \sqrt{f}$.
% The algorithm then receives a sequence of $q$ edge insertions or deletions in $G$.
% The edge updates are such that the maximum degree of $G$ never exceeds $\Delta$, and the neighborhood of each vertex spans at most $\Delta^2 / f$.
% The objective is to maintain a proper list coloring of $G$.
% Note that the triangle-free coloring problem is a special case of the above problem where $f$ is fixed as $1$.

% \btodo{TO DISCUSS: main obstacle for unknown $\Delta$: (1) the flaws depends on $L = \Delta^\eps$. Changing $\Delta$ will change all flaws at the same time. (2) When the number of colors decreases, we need to fix all vertices having the removed color. This can potentially change colors of $O(n)$ vertices, and it is not clear how to combine the analysis with this change!}

% The problem is formally defined as follows.
The input consists of a parameter $\Delta > 0$ and a graph $G = (V, E)$ of $n$ vertices.
Initially, $G$ has no edges, and each vertex $v$ is associated with a color list $C_v$ of size $\ell = 6\Delta / \ln \Delta$.
The algorithm then receives a sequence of $q$ edge insertions or deletions in $G$.
Throughout, $G$ is a triangle-free graph, and the maximum degree of $G$ never exceeds $\Delta$.
The objective is to maintain a proper list-coloring of $G$, that is, to assign each vertex $v$ a color from $C_v$ such that no two adjacent vertices share a color.
Without loss of generality, we assume that $\Delta \geq 100$.
For $\Delta < 100$, we have $\ell \geq \Delta + 1$, and thus an $\ell$-list coloring can be easily maintained in $O(\Delta)$ time.

% \begin{remark}
% The formulation above assumes that an upper bound $\Delta$ of maximum degree is known.
% This assumption does not induce any loss of generality in the $(2+o(1))(\Delta / \ln \Delta)$-list coloring problem, because knowing the size of the color lists already implies an upper bound on $\Delta$.
% The assumption can be removed in the standard coloring setting, that is, when the color lists of all vertices are the same and not specified in the input.
% \end{remark}

\subsection{Overview of the algorithm}
Our algorithm is an adaptation of \cite{molloy2019list,achlioptas2019beyond}'s algorithm.
We will treat $\blank$ as a color and view uncolored vertices as having this color.
A \emph{partial coloring} is an assignment of colors such that each vertex $v$ gets a color from $C_v \cup \{\blank\}$ and two adjacent vertices do not share a non-blank color.
For a partial coloring $\sigma$, let $\sigma(v)$ denote the color of $v$ in $\sigma$.

The algorithm consists of two phases.
In Phase 1, we maintain a partial coloring $\sigma_1$ in which a substantial number of vertices are colored.
Then, we complete $\sigma_1$ into a proper coloring in Phase 2.
We first present a static version, showing how to color a triangle-free graph properly.

\paragraph{Phase 1.} Let $\Omega$ be the set of partial coloring, i.e., $\Omega = \prod_{v\in V} (C_v \cup \{\blank\})$.
For a partial coloring $\sigma$, we let $L_v(\sigma)$ be the set of colors we can assign to $v$ without creating a monochromatic edge.
Note that $\blank \in L_v(\sigma)$.
Let $L = \Delta^{0.7}$.
Define two flaws for each vertex $v$:
\begin{itemize}
    \item $B_v = \{\sigma \in \Omega \mid |L_v(\sigma)| < L\}$ .
    \item $Z_v = \{\sigma \in \Omega \mid \text{$v$ has at least $L - 1$ neighbors colored $\blank$ in $\sigma$}\}$.
\end{itemize}
For a flaw $f \in \{B_v, Z_v\}$, define $v(f) = v$.
Phase 1 finds a flawless partial coloring $\sigma_1$ using the algorithmic framework for \cref{cond:general}.
For each vertex $v$, the resampling procedures of $B_v$ and $Z_v$ are the same -- to re-assign each vertex $u \in N_v$ a uniformly chosen color from $L_v(\sigma)$.
The flaw choice strategy $\pi$ can be any permutation that prioritizes $B$-flaws over $Z$-flaws.
% Flaws $B_u$ and $Z_v$ are ordered by the index of $u$ and $v$, respectively.
We show in \cref{sec:app-tri-convergence} that the above setup satisfies \cref{cond:general}.
Hence, a flawless state can be found efficiently.

\paragraph{Phase 2.}
% Phase 2 completes $\sigma_1$ by solving the following subproblem:
Let $V'$ be the set of blank vertices in $\sigma_1$.
Denote by $N_v$ the neighbor set of $v$ in $G$.
% Each vertex $v \in V'$ is associated with a color list $:'_v(\sigma_1) = L_v(\sigma_1) \backslash \{ \blank \}$.
% Let $G'$ be the subgraph of $G$ induced by $V'$.
If $G$ is a static graph, $\sigma_1$ can be completed by a simple greedy approach:
Iterate through each vertex $v \in V'$, and color $v$ with any color in $L_v(\sigma_1) \setminus \{\blank\}$ that has not been chosen by any other vertex in $N_v \cap V'$.
Since every vertex has at least $L - 1$ colors to choose from and only has $L - 2$ blank neighbors, the approach must successfully find a proper coloring of $G$
% First, compute a list coloring $\sigma_2$ for $G'$ w.r.t. the color lists $C'(\cdot)$.
% Next, color each blank vertex $v$ with $\sigma_2(v)$.

\paragraph{Dynamic setting.}
In our case, $G$ is a dynamic graph, and we have to maintain $\sigma_1$ and $\sigma_2$ under edge updates.
Since Phase 1 satisfies \cref{cond:general}, it can be turned into a dynamic algorithm for maintaining $\sigma_1$.
When an edge update is received, we first translate it into $O(1)$ flaw updates to Phase 1.
The algorithm will then resample the colors of several vertices until $\sigma_1$ becomes flawless again.
As we will show, when amortized over all updates, Phase 1 requires $\tO(\Delta^3)$ time and recolors $\tO(\Delta)$ vertices.

The change in Phase 1 will modify $V'$ and $L_v(\sigma_1)$ for $\tO(\Delta^2)$ vertices $v$.
These vertices are treated as the input to Phase 2.
In Phase 2, we recolor these vertices using the greedy approach.
We will show that Phase 2 requires an amortized $\tO(\Delta^3)$ time.
Hence, the algorithm has an amortized $\tO(\Delta^3)$ upper bound on update time.

% Let $\Omega'$ be the set of proper colorings for $G'$, that is, $\prod_{v\in V'} C'_v$.
% For each edge $(u, v) \in E'$ and each color $c \in C'(u) \cap C'_v$, we define the following flaws:
% \begin{itemize}
%     \item $A_{uv, c} = \{\sigma \in \Omega' \mid \sigma(u) = \sigma(v) = c\}$.
% \end{itemize}
% Again using an LLL-based argument, it can be shown that a flawless state exists, showing that a proper coloring for $G'$ exists.
% Consequently, a $(2+\eps)\Delta/\ln\Delta$ list-coloring for $G$ exists.

\subsection{Details for Phase 1}
We use the framework in \cref{alg:app} to maintain a flawless partial coloring $\sigma_1$ in Phase 1.
We assume that the framework supports an additional operation: replacing an existing flaw $f$ by a new flaw $f'$.
This operation can be implemented by first removing $f$ and then adding $f'$.

\paragraph{Data structure setup:}
Assume $V = \{1, 2, \dots, n\}$ and $G$ is given as adjacency lists, i.e., $N_v$ of each vertex $v$.
Note that $v \notin N_v$, and since $G$ is triangle-free, $N_v$ is an independent set.
A flaw $f$ is represented by the ID of $v(f)$ and a binary variable indicating whether it is a $B$-flaw.
The color list $C_v$ of a vertex $v$ is given as an array of size $\ell$.
Phase 1 maintains the following:
\begin{itemize}
    \item for each vertex $v$, the neighbor set $N_v$,
    % and a set $N'_v(\sigma_1)$ containing all blank neighbors of $v$,
    \item the color list $L_v(\sigma_1)$ for each $v$,
    % \item for each vertex $v$, a BBST $\cnt_v$ that maps each color $c \in C_v$ to the number of neighbors of $v$ having color $c$,
    % \item the set of blank vertices $V'$,
    \item the current set of flaws $\cF$ for graph $G$, and
    \item the flawless partial coloring $\sigma_1$.
    % \item For each vertex $v$, we also store the The subgraph $G'$, represented by its vertex set $V'$ and its adjacency list. The adjacency lists are stored in BBSTs.
    % \item The color lists $C'(\cdot)$. Recall that $C'(\cdot) = L_\cdot(\sigma_1) \backslash \{\blank\}$. In our analysis, we often omit the maintenance of $C'(\cdot)$, because they can be maintained in the same way as $L_\cdot(v)$.
\end{itemize}
We store $N_v$ and $L_v(\sigma_1)$ as BBSTs, so that insertion and deletion take $\tO(1)$ time.
The coloring $\sigma_1$ is stored as an array of size $n$.

\paragraph{Framework specification.}
We proceed to describe the implementation of the key procedures in \cref{alg:app}.
\begin{itemize}
    \item Given two flaws $f_1, f_2$, $\algCompare(f_1, f_2)$ is done by comparing their IDs.
    Recall that we only require all $B$-flaws to be prioritized over $Z$-flaws.
    This can be done in $O(1)$ time.

    \item Given a flaw $f$, $\algCheck(f, \sigma_1)$ is done by checking, for each $u \in N_v$, $|L_u(\sigma_1)|$ and the color of $u$.
    This requires $O(\Delta)$ time.

    \item Given a flaw $f$ of a vertex $v$, $\algResample(f, \sigma_1)$ first resamples the color of each $u \in N_v$ and then finds the list of flaws introduced.
    To do this, we iterate through each $u \in N(v)$ and recolor it using any randomly chosen color from $L_u(\sigma_1)$.
    Consider a vertex $u \in N_v$ whose color is changed from $c$ to $c'$ in the resampling.
    We first update whether $c$ and $c'$ are in $L_w(\sigma_1)$ for each $w \in N(u)$.
    This is done by checking whether $c$ and $c'$ are used by any neighbor of $w$.
    % Then, if $c$ or $c'$ is $\blank$ we also update $V'(\sigma_1)$ and $N'_v(\sigma_1)$.
    Hence, recoloring a vertex $u$ takes $\tO(\Delta^2)$ time, and recoloring all $u \in N_v$ takes $\tO(\Delta^3)$ time.

    Let $D_v$ be the set of vertices with distance at most 2 from $v$.
    Note that the resampling only changes the colors of vertices in $N_v$, and it only changes the color list $L_w(\sigma_1)$ for $w \in D_v$.
    Hence, at most $O(\Delta)$ vertices change their colors, and at most $O(\Delta^2)$ lists $L_w(\sigma_1)$ are changed.
    This also shows that the resampling can only introduce a flaw $B_u$ or $Z_u$ if $u$ is in $D_v$.

    Based on the observation, we compute the list of flaws introduced by invoking $\algCheck(B_w, \sigma_1)$ and $\algCheck(Z_w, \sigma_1)$ for each $w \in D_v$.
    This requires $O(\Delta)$ time on each of the $|D_v| = O(\Delta^2)$ vertices.
    The running time is therefore $O(\Delta^3)$.
\end{itemize}

\paragraph{Initialization.}
Initially, $G$ does not contain any edges and $\cF$ has no flaws.
We initialize $\sigma_1$ as any coloring without blank vertices, and set $\sigma_2 = \sigma_1$.
The initialization of $\{L_v(\sigma_1)\}_{v\in V}$ is straightforward.
This takes $\tO(n\ell)$ time.

Next, we insert the flaws $B_v$ and $Z_v$ for each $v \in V$, so that $\cF$ corresponds to the flaw set of an empty graph.
Note that the flaws $B_v$ and $Z_v$ are determined solely by $N_v$ -- $B_v$ contains any coloring $\sigma$ such that any vertex in $N_v$ has $|L_v(\sigma)| < L$, and $Z_v$ contains any coloring $\sigma$ such that at least $L$ vertices in $N_v$ are blank.
We do not compute these sets explicitly -- the flaws are specified as their IDs and the procedures $\algResample, \algCompare$, and $\algCheck$.

Since there are no edges, $\sigma_1$ is flawless and $\sigma_2$ is proper.
Overall, the initialization takes $\tO(n\ell)$ time and $2n$ flaw updates.
% In addition, $L_v(\sigma) = C_v$ for each vertex $v$, and $F = V' = \emptyset$.
% The initialization can be done in $O(\sum_{v \in V} |C_v|) = O(n \ell)$ time, which is linear in the input size.

\paragraph{Handling updates.}
To process an edge update $(u, v)$, the algorithm performs three steps:
\begin{enumerate}
    \item Update the auxiliary data structures ($N_v, L_v(\sigma_1)$, etc.).
    \item Update the set of flaws $\cF$ to reflect the change on $G$.
    \item Collect all vertices $w$ that changes its color or its list $L_w(\sigma_1)$ during Step 2, and send them as the input of Phase 2.
\end{enumerate}

The first step is to update $N_w$, and $L_w(\sigma_1)$ for $w \in \{u, v\}$ to reflect the change on $(u, v)$.
This can be easily done in $\tO(1)$ time.

We then proceed to the second step.
Recall that each flaw $B_w$ or $Z_w$ is a set of partial colorings that depends only on the neighbor set of $w$.
Therefore, only the flaws $B_v, Z_v, B_u$, and $Z_u$ are changed.
(Here, a flaw is considered changed if it contains a different set of partial colorings.)
Hence, the change can be seen as replacing four flaws (or eight flaw updates) in our algorithmic framework.
We prove the following lemma in \cref{sec:app-tri-convergence}, showing that the LLL condition is satisfied after each update.

\begin{lemma} \label{lem:tri-convergence}
After each flaw update, the current set of flaws $\cF$ satisfies \cref{cond:general} with parameters $\eps = \frac{1}{4}$ and $\psi_f = 2\Delta^{-3}$ for all flaws $f \in \cF$.
\end{lemma}

The flaw updates will execute \algResample several times and reach a flawless state again.
In the finaly step, the algorithm collects the set of all vertices $w$ such that $\sigma_1(w)$ or $L_w(\sigma_1)$ changes by any of the resampling.
This set, denoted by $W$, is sent to Phase 2 as the input.

% By the result in \cref{sec:framework}, the algorithm only needs amortized $\tO(1)$ flaw-addressing steps to make $\sigma_1$ flawless again.

% The modification is presented as follows.
% For any partial coloring $\sigma$ (which may not be $\sigma_1$), the lists $L_u(\sigma)$ and $L_v(\sigma)$ are changed by at most one color.
% Therefore, $B_u$ and $B_v$ are changed.
% Similarly, the flaws $Z_u$ and $Z_v$ are changed due to changes in their neighbor sets.
% It is not hard to see that all other flaws stay the same, as the neighbor sets of all other vertices are unchanged.
% Hence, the change can be seen as replacing four flaws in the framework of \cref{sec:framework}.

% \begin{lemma}
% \label{lem:app1:phase1-framework}
%     Let $\langle F_1, F_2, \dots F_t \rangle$ denote a sequence of flaw updates made corresponding to edge deletions. We have two different kinds of flaws, $B_v$ and $Z_v$ for each vertex $v$.
%     Assuming the guarantee that the edge updates yield a graph of maximum degree $\Delta$
% \end{lemma}

\paragraph{Running time and recourse.}
By the above analysis, Phase 1 requires $O(n \ell)$ time to initialize, and the update time is dominated by the time spent in updating the flaws.
In total, the algorithm generates $O(n + q)$ flaw updates to the framework, and the mapping from edge updates to flaw updates is deterministic.
Therefore, the number of possible update sequences is at most $\zeta = (2n^2)^q$.
By applying (A2) of \cref{thm:application} with $\zeta = (2n^2)^q, |\Omega| = \ell^n, d = \max_{v \in V}|D_v| \leq \Delta^2 + 1, \psi_{\min} = \psi_{\max} = 2\Delta^{-3}, \eps = 1/4, \TResample = \tO(\Delta^3), \TCheck = O(\Delta), \TCompare = O(1)$, and the number of updates is $O(n + q)$, we obtain the following.

\begin{lemma} \label{lem:tri-phase1}
Phase 1 maintains a flawless partial coloring $\sigma_1$ using $\tO((n + q) \cdot \Delta^3)$ time. In its execution, \algResample is invoked $\tO(n + q)$ times.
\end{lemma}

% We first spend $\tO(1)$ time to update $L_u(\sigma_1)$ and $L_v(\sigma_1)$.
% Next, we check whether each of $B_u, B_v, Z_u, Z_v$ should be added to $F$. This is done by examining $u, v$ and their neighbors, which requires $O(\Delta)$ time.
% We then address the flaws in $F$ until it becomes empty.

% Addressing a flaw of $v$ is done by resampling the color of each neighbor of $v$.
% When we change the color of a vertex $u$, the lists $L_w(\sigma_1)$ of all its neighbors $w$ are also changed by adding and removing at most one color.
% Therefore, a resampling operation changes the color of at most $\Delta$ vertices and modifies the list $L_w(\sigma_1)$ of at most $O(\Delta^2)$ vertices $w$.
\noindent Since the algorithm only calls \algResample $\tO(n + q)$ times, the overall recourse is bounded as follows:
\begin{itemize}
    \item In the execution of Phase 1, $\tilde{O}((n + q) \cdot \Delta)$ vertices change their colors, and
    \item $\tilde{O}((n + q) \cdot \Delta^2)$ vertices $w$ change their $L_w(\sigma_1)$.
\end{itemize}

% Changing a color requires $O(1)$ time, and updating an $L_w(\sigma_1)$ takes $\tilde{O}(1)$ time.
% Thus, the amortized running time is $\tilde{O}(\Delta^2)$.

\subsection{Details for Phase 2}
We maintain a proper coloring $\sigma_2$ obtained by completing $\sigma_1$.
Initially, $G$ is empty and $\sigma_1$ is any proper coloring, and therefore we simply set $\sigma_2 = \sigma_1$.

When an edge update comes, Phase 1 will process the change and provide a set $W$ of vertices that either change their color or color list to Phase 2.
% In addition, Phase 1 maintains the set of blank neighbors $N'_v(\sigma_1)$ for each vertex $v$.
Phase 2 processes the change by recoloring three types of vertices in $\sigma_2$:
\begin{enumerate}
    \item the vertices that change their colors in $\sigma_1$, and the new color is not $\blank$,
    \item the vertices that change their colors in $\sigma_1$, and the new color is $\blank$, and
    \item the blank vertices $v$ that do not change their colors, but $L_v(\sigma_1)$ is changed.
\end{enumerate}
These vertices form a subset of $W$.
For the first type of vertices, we simply set $\sigma_2(v) = \sigma_1(v)$.
We then iterate through the other type of vertices, and for each vertex $v$, select the first non-blank color from $L_v(\sigma_1)$ that is not used by any $u \in N(v)$.
Note that such a color must exist, because $\sigma_1$ is flawless (that is, each vertex has at least $L - 1$ non-blank colors in $L_v(\sigma_1)$ and at most $L - 2$ blank neighbors).
Recoloring a vertex takes $\tO(\Delta)$ time.
By the recourse bound for Phase 1, in the whole execution Phase 2 only recolors $\tO((n + q) \cdot \Delta^2)$, hence spending $\tO((n + q) \cdot \Delta ^3)$ time.
Combining with \cref{lem:tri-phase1}, we obtain the following.

\begin{theorem}
Given a fully dynamic triangle-free graph, a list-coloring of $\frac{6\Delta}{\ln \Delta}$ colors can be maintained under edge updates in $\tO(\Delta^3)$ amortized update time with probability $1 - 2^{-q}$ against a clairvoyant adversary.
\end{theorem}

% \subsubsection{Extension to $(1+o(1)) \Delta / \ln \Delta$ coloring}
% In Phase 1, we need to use a more complicated flaw $Z'_v$, which is defined based on $T_{v, c}$.
% With this modification, the list coloring problem in Phase 2 satisfies the variable-setting LLL.
% However, we should note that:
% \begin{itemize}
%     \item When adding or removing an edge in Phase 1, at most $\Delta$ flaws can be changed.
%     This increases the running time by $\Delta$.
%     \item In the dynamic setting, we may change $L_v(\sigma_1)$ in Phase 1.
% This corresponds to changing the domains of some variables in Phase 2.
% When a variable changes domain, we can handle it by adding a new variable.
% It should be possible to modify the MT framework to allow additions of variables.

% \end{itemize}

% Since the amortized recourse of Phase 1 is $\tilde{O}(\Delta)$, the number of vertices in $G'$ is 

% which can be solved using the framework of Moser and Tardos.

% \subsubsection{Possible generalizations} 
% \begin{itemize}
%     \item list coloring
%     \item high-girth graphs or limited number of edges in neighborhood graphs
%     \item correct or best-known constant?
% \end{itemize}

\section{Dynamic routing schedules}
\label{sec:LMR-scheduling}

The Leighton–Maggs–Rao (LMR) \cite{Leighton1994} scheduling result addresses the problem of routing packets or messages through a network.
The network is modeled as a directed graph $G=(V,E)$, where vertices represent processors and edges represent communication links.
The input is a collection of paths, each corresponding to a packet that must be routed from a source vertex to a destination vertex along the specified path in the network.
Two critical parameters of this routing scenario are:
\begin{itemize}
    \item  Congestion ($C$): The maximum number of paths/messages that traverse any single edge in the network.
    \item  Dilation ($D$): The maximum length of any of these paths, measured as the number of edges.
\end{itemize}

The scheduling problem is to determine an efficient schedule that delivers all messages along their prescribed routes while minimizing the total completion time, assuming that each edge can transmit at most one message per time step.
Clearly, both $C$ and $D$ are lower bounds on the completion time. Leighton, Maggs, and Rao proved that any set of messages/paths with congestion $C$ and dilation $D$ can be scheduled and completed in $O(C+D)$ time steps.

They first give a randomized way to achieve an $O(C + D \log n)$ schedule via random delays and a union bound over all paths.
Then use the LLL to make the union bound ``local'' and therefore remove the dependence on $n$, giving an $O(C + D \log CD)$ schedule.
By repeated LLL applications this can be reduced further to $O(C + D \log \log CD)$ and more generally to $O((C + D \log^{(i)} CD) \cdot 2^{O(i)})$.
The arguments on how to get to the full $O(C+D)$ result are rather involved for LMR but a later drastic simplification by Rothvoss~\cite{rothvoss2013simpler} gives a beautiful proof that fits into the MT-framework with $pde < 1$ convergence and constant slack.

% (Each bad event depends on at most $OPT^2$ many other bad events where $OPT = O(C + D)$. There are $n*OPT^2$ many such events, and their probability is $OPT^{-10}$ applied only $O(\log \log n)$ times.)
% One can also think of it as a lopsided (ordered by type) single LLL instance. 

%  We should therefore be able to get the $O(C + D)$ result dynamically with some poly OPT computational overhead, i.e.: 

% Suppose an edit adds or removes a new route (of length at most D and causing congestion at most C with the currently active routes). The goal is to maintain a schedule for each packet which is valid and has low completion time, while hopefully only investing a small amount of work to update the schedule as new routes come in. 

\paragraph{Dynamic setting.}
We extend \cite{rothvoss2013simpler}'s result to the dynamic settings where it is allowed to add or remove paths.
The problem is formally defined as follows.
The input consists of a directed graph $G = (V, E)$, a parameter $D$, and a sequence of updates, where each update inserts or deletes a path.
Let $n = |V|$ and $m = |E|$.
The network may contain parallel edges and self-loops, and each path may visit each vertex multiple times.
Throughout, the congestion and dilation never exceed $D$.
The \emph{makespan} of a schedule is the last time point where a packet arrives at its destination.
The goal is to maintain a schedule with $O(D)$ makespan.
% For a sequence of $q$ updates, the amortized update time of our algorithm will be $O(\poly(D, \log q))$.
We call this problem \emph{the dynamic scheduling problem}.

After adding dummy edges, one can assume that every path has exactly $D$ edges.
% Suppose that the current problem instance contains $k$ paths.
% Recall that each path $P_i$ corresponds to a packet $i$.
% % We represent a solution by a $k \times D$ matrix $\beta$, where $\beta_{i,j}$ is the time point that packet $i$ should cross the $j$-th edge on $P_i$.
% % Such a solution is called a \emph{schedule}.
In our analysis, we allow a schedule to have multiple packets passing through a single edge simultaneously.
The \emph{load} of a schedule is the maximum number of packets crossing an edge at any time point.
A schedule $S$ with load $c$ and makespan $s$ can be converted to another schedule $S'$ with load 1 and makespan $c \cdot s$, by simulating each step in $S$ with $c$ steps in $S'$.
Therefore, it suffices to maintain a schedule with $O(D)$ makespan and $O(1)$ load.

Without loss of generality, we assume that the algorithm knows $q$.
The case of unknown $q$ can be addressed by a standard doubling argument, in which we guess an upper bound $\hat{q} \geq q$ and double the guess when $q$ exceeds $\hat{q}$.
% If $q$ is unknown, the algorithm can guess an upper bound $\hat{q}$ on $q$ and operate as if $q = \hat{q}$;
% when the number of updates exceeds $\hat{q}$, we set $\hat{q}$ as $2\hat{q}$ and rerun the algorithm through all updates from scratch.
% By a standard argument, it can be shown that the time complexity is increased by only a constant factor.

\subsection{Review of \texorpdfstring{\cite{rothvoss2013simpler}}{Rothvoss}'s construction} \label{sec:scheduling-review}
This section gives a high-level overview of \cite{rothvoss2013simpler}'s proof for the static setting, in which the input is a static set of $q$ paths on a graph $G$.
The construction is via $\Theta(\log \log n)$ applications of \LLL.

\paragraph{Path dissection.}
Consider a packet $i$.
We partition its path $P_i$ into a laminar family of \emph{blocks} such that a block in \emph{level $\ell$} contains $D_\ell = D^{0.5^\ell}$ consecutive edges in $P_i$.
The dissection is stopped when the blocks in the last level have length between $\tau$ and $\tau^2$, where $\tau$ is a sufficiently large constant.
% See {\color{red}Figure ???} for an example.
% We assume that $D > \tau^2$ so there are at least two levels.
It will be convenient to assume that each $D_\ell$ is an integer and $D_{\ell}$ divides $D_{\ell-1}$; all calculations in the analysis will have enough slack so that one can replace each $D_\ell$ by the nearest power of two if the assumption does not hold.

Denote by $L = \Theta(\log \log n)$ the index of the last level.
The resulting laminar family has $L+1$ levels, where level $0$ contains a single block ($P_i$ itself) and each block in level $\ell$ has children of length $D_{\ell+1} = \sqrt{D_\ell}$ in level $\ell + 1$.
Each block has two \emph{boundary nodes}, the \emph{start node} and the \emph{end node}.

\paragraph{Construction of schedules.}
For $\ell \in [0, L]$, define
\begin{equation} \label{eqn:def-waiting-times}
    W_\ell = \begin{cases}
        D_\ell, & \ell = 0 \\
        D_\ell^{1/4}, & \text{otherwise.}
    \end{cases}    
\end{equation}
The schedule of a packet $i$ is as follows.
% Let $\range(a, b)$ denote $\{a, a+1, \dots, b\}$ and $\range(r)$ denote $\{1, 2, \dots, r\}$.
For each level $\ell$ block $b$, we assign uniformly and independently an integer $\alp_\ell(i, b)$ from $[1, W_\ell]$.
\footnote{Here, for ease of notation, we use $[a, b]$ to denote $\{a, a+1, \dots, b\}$ instead of the set of real numbers between $a$ and $b$. All levels and times mentioned in this section are integers.}
At each level $\ell$ block $b$, the packet $i$ waits $\alp_\ell(i, b)$ units of time at the start node and $W_\ell - \alp_\ell(i, b)$ time at the end node.
% (See {\color{red}Figure ???}.)

The policy has several properties:
First, observe that the time $t$ at which a packet $i$ crosses an edge $e$ is a random variable that depends only on the random waiting times of the blocks containing $e$ -- that is, only one block from each level.
Second, the total waiting time for a packet is exactly $T = \sum_{\ell=0}^{L} W_\ell \cdot \frac{D}{D_\ell} = O(D)$.
Hence, each packet traverses its path in exactly $T + D$ time for any choice of waiting times.
This shows that any waiting times give a schedule of $O(D)$ makespan, but the load is left undetermined.

In the following, we show that there exists a proper choice of waiting times $\alp_0, \alp_1, \dots, \alp_L$ such that the resulting schedule has $O(1)$ load.
The approach is to formulate each level as an LLL instance, and fix the waiting times from level $0$ to $L - 1$.
The last level, i.e., level $L$, is a special case.
\cite{rothvoss2013simpler} showed that after all waiting times in $\alp_0, \dots, \alp_{L-1}$ are fixed, we can simply assign all level-$L$ waiting times arbitrarily.
The result will be a schedule of $O(1)$ load.
Hence, in the following, we only consider levels $\ell < L$.
% Let $\hat{\beta_\ell}$ denote the vector $(\beta_0, \beta_1, \dots, \beta_\ell)$. Define $\hat{\alp_\ell}$ similarly.

\paragraph{Variables and flaws.}
Consider a fixed level $\ell$.
The set of variables is $\{\alp_\ell(i, b) \mid i \in [q], \text{ $b$ is a } \allowbreak \text{level-$\ell$ block for packet $i$}\}$.
To ease the notation, denote this set simply by $\alp_\ell$.
The state space, denoted by $\Omega_{\ell}$, is the set of assignments for the variables in $\alp_\ell$.
Note that $|\alp_\ell| = q \cdot \frac{D}{D_\ell}$, and each variable takes value from $[1, W_\ell]$.

The definition of flaws is recursive.
Let $\cF_\ell$ denote the set of level-$\ell$ flaws.
In the construction, we will first define $\cF_0$ and find a flawless assignment $\sigma_0$ for $\alp_0$.
The set of flaws $\cF_1$ for level $1$ are defined based on $\sigma_0$.
We then find a flawless assignment $\sigma_1$ for $\alp_1$ and use $\sigma_0, \sigma_1$ to define $\cF_2$.
This process is repeated from level $0$ to $L$.
In general, we will define $\cF_{\ell}$ based on $\sigma_0, \sigma_1, \dots, \sigma_{\ell-1}$, and $\sigma_\ell$ is computed as a flawless assignment w.r.t. $\cF_{\ell}$.
After all levels are completed, the assignments form a solution, i.e., a schedule of $O(1)$ load.

The flaws are formally defined as follows.
An \textit{$\ell$-assignment} $\sigma = (\sigma_0, \sigma_1, \dots, \sigma_{\ell-1})$ is an assignment for all variables in $(\alp_0, \alp_1, \dots, \alp_{\ell-1})$, where the assignment for $\alp_x$ is denoted by $\sigma_x$.
Given an $\ell$-assignment $\sigma$, define $X_\ell(e, t, i)$ as a binary random variable indicating whether packet $i$ crosses edge $e$ at time $t$, conditioned on $\alp_x = \sigma_x$ for $x = 0, 1, \dots, \ell-1$.
Note that the definition of $X_\ell(e, t)$ depends on $\sigma$, which shall be the output of all previous levels.
In addition, the outcome of $X_\ell(e, t, i)$ depends on the random choices for $\alp_{\ell}, \alp_{\ell+1}, \dots, \alp_{L}$.
% At level $\ell \in [0, L-1]$, we assume that we have already fixed all waiting times in previous levels, i.e., $\alp_0 = \beta_0, \alp_1 = \beta_1, \dots, \alp_{\ell-1} = \beta_{\ell-1}$ for some $\beta_0, \beta_1, \dots, \beta_{\ell - 1}$.
% Denote by $X_\ell(e, t, i)$ a binary variable indicating whether packet $i$ crosses edge $e$ at time $t$, conditioned on $\alp_x = \beta_x$ for $x = 0, 1, \dots, \ell-1$.
% Note that the outcome of $X_\ell(e, t, i)$ is defined based on the random distribution for $\alp_{\ell}, \alp_{\ell+1}, \dots, \alp_{L}$.
Let $X_\ell(e, t) = \sum_{i=0}^k X_\ell(e, t, i)$ be the number of packet crossing edge $e$ at time $t$ under the same condition.
Let $\lambda_\ell \in [1, 2]$ be a real number that will be specified later.
A flaw $F_\ell(e, t)$ is defined as the set of assignments $\sigma_\ell$ for $\alp_\ell$ such that $\E[X_\ell(e, t) \mid \alp_\ell = \sigma_\ell] \geq \lambda_\ell$.
Intuitively, a flaw represents the following event: ``Suppose that we have fixed all waiting times in $\alp_0, \alp_1, \dots, \alp_{\ell-1}$ as $\sigma$, and now we attempt to choose an assignment for $\alp_\ell$. The current choice of $\alp_\ell$ is such that the expected load of edge $e$ at time $t$ is larger $\lambda_\ell$.''
For each level, the flaw choice strategy $\pi$ can be any permutation on $\cF_\ell$.

The exact value of $\lambda_\ell$ is not important for our purpose.
For completeness, we provide its definition: $\lambda_0 = 1$, and $\lambda_\ell = \lambda_{\ell-1} + \frac{1}{D_{\ell}^{-1/32}}$ for $\ell > 0$.
Note that $\lambda_\ell \in [1, 2]$.

% The analysis will focus on upper-bounding $\E[X_\ell(e, t)]$, the expected load of an edge $e$ at time $t$.
% Find the waiting times for $\alp_\ell$ such that none of the bad events happen, where the bad events are defined below.

\paragraph{Resampling procedures.}
Consider a level $\ell$ and assume that we already have flawless assignments $\sigma_0, \sigma_1, \dots, \sigma_{\ell-1}$ for all previous levels.
% Let $\cF_\ell$ be the set of flaws for level $\ell$.

% \paragraph{Proof structure.}
% The goal is to determine $\sigma_\ell$ for $\ell = 0,1,\dots,L$ such that the expected load $\E[X_\ell(e,t)]$ is low for every edge $e$ and time step $t$.
% We first derive an upper bound on the expected load \emph{before} any waiting time is fixed.
% To this end, we need the following lemma.

\begin{lemma}[{\cite[Lemma 3]{rothvoss2013simpler}}] \label{lem:scheduling-static1}
    Let $\ell \in \{0, \dots, L\}$ and condition on any fixed assignments $\sigma_0, \dots, \sigma_{\ell-1}$ for $\alp_0, \dots, \alp_{\ell - 1}$.
    Then, for any packet $i$, edge $e \in E$ and time $t \in [1, T]$, there are only $O(D_\ell)$ packets that still has a non-zero probability of passing $e$ at $t$, where the probability is taken over the random choice of $\alp_\ell, \alp_{\ell+1}, \dots, \alp_{L}$.
\end{lemma}

\noindent
The above lemma shows that each $F_\ell(e, t)$ depends only on $O(D_\ell)$ variables in $\alp_\ell$ that can be found as follows:
Let $K_{\ell}(e, t)$ be the set of packets that $i$ that still have a non-zero probability of passing $e$ at $t$.
Then, $\vbl(F_\ell(e, t))$ is the set of waiting times $\alp(i, b)$ such that $i \in K_{\ell}(e,t)$ and $b$ is the unique level-$\ell$ block for $P_i$ that contains $e$.
The resampling procedure for $F_\ell(e, t)$ is the MT-resampling -- i.e., to resample all variables in $\vbl(F_\ell(e, t))$ uniformly at random.

\paragraph{Convergence proof for a level.}
Note that the LLL setup falls into the variable setting (\cref{def:variable}) and all resampling procedures are MT-resamplings.
Hence, by \cref{lem:prod}, for each flaw $F' \in \cF_\ell$, $\gamma_{F'} = \Pr[F']$.
\cite{rothvoss2013simpler} showed that $\Pr[F_\ell(e, t)] \leq e^{-{D_\ell^{1/16}/12}} \leq D_{\ell}^{-4}$ for a sufficiently large $\tau$.
In addition, each flaw depends only on $D_{\ell}^{3}$ other flaws.
Hence, the setup satisfies the symmetric LLL condition (\cref{cond:pde}) with $p = D_{\ell}^{-4}$, $d = D_\ell^3$, and $\eps = \Theta(1)$.
We can convert the condition to \cref{cond:general} by setting $\psi_{F'} = ep$ for each flaw $F'$.
This implies the following.

\begin{lemma} \label{lem:scheduling-convergence}
    The flaw set in level $\ell$ satisfy \cref{cond:general} with parameters $\eps = \Theta(1)$ and $\psi_{F'} = \Theta(D_\ell^{-3})$ for all $F' \in \cF_\ell$.
\end{lemma}

\subsection{Overview of the algorithm} \label{sec:scheduling-overview}
% \btodo{TODO:
% Fix the value of $\tau$ to make the algorithm more exact.
% (Each level decays at least 4 times, rounding of the power of 2, the super polynomial is larger than $D_\ell^3$.)
% % Also, is (P2) really needed?
% Also, add data structure setup (conversion of $O(1)$ load to unit load).
% Notation: $[1, T]$ and $[T]$ are integers instead of reals.}

This section presents an overview of our dynamic algorithm, which maintains a schedule under insertions and deletions of paths.
For simplicity, we focus on how to maintain the waiting times, and we only present updates for path insertions.
Path deletions can be handled symmetrically.

Since each level is an LLL instance, it can be turned into a dynamic algorithm, where each level $\ell$ maintains a flawless assignment $\sigma_\ell$.
The main technical difficulty is to figure out the update sequence:
Recall that the flaw set of each level depends on a flawless assignment for all previous levels.
Therefore, the algorithm proceeds as follows:
First, it turns a path insertion into several flaw updates in level 0.
These updates will trigger a series of resamplings and change $\poly(D)$ entries in $\sigma_0$.
Each entry appears in the definition of $\poly(D)$ flaws in level 1, and hence the update sequence of level 1 is to update these changed flaws.
The above process is repeated for each of the higher levels.
More details are given below.

% On a high level, the algorithm proceeds as follows.
\paragraph{Updating the first level.}
% Let $\alp = \bigcup_{\ell \in [0, L]} \{\alp_\ell\}$ be the set of variables.
When a path $P_i$ is inserted, we first dissect it into $L$ levels of blocks.
Each block should be assigned a variable.
Since the algorithm knows $q$ beforehand, it can initialize $q \cdot \poly(D)$ variables in the beginning so there are enough variables for $q$ paths.
We assign any of the unused variables to each new block.

The path $P_i$ passes $D$ edges and changes the expected load of them.
We update the flaws for these edges in level $0$ so that the load caused by $P_i$ is considered.
% (Again, an update of a flaw is represented by assigning a new flaw ID and new oracle access to the resampling procedure.)
Each edge only has $T = O(D)$ flaws associated in level $0$, so in total only $U_0(i) = \poly(D)$ updates are needed.
Let $R_0(i)$ be the number of resamplings caused by these updates.
Each resample changes $\poly(D)$ variables in $\alp_0$.
In total, the updates change $C_0(i) = R_0(i) \cdot \poly(D)$ variables.

\paragraph{Updating the higher levels.}
Consider a level $\ell \in [1, L - 1]$.
We first show that each lower-level variable affects at most $\poly(D)$ level-$\ell$ flaws.
Fix a variable $\alp_{\ell'}(i, b)$ where $\ell' < \ell$.
This variable only affects a flaw $F_\ell(e, t)$ if the block $b$ contains $e$.
Since the block contains at most $D$ edges, and there are only $T = O(D)$ choices of $t$, only $\poly(D)$ flaws are affected.

Suppose that we have finished the update for all levels $0, 1, \dots, \ell - 1$ and in total $C_{\ell-1}(i)$ variables from previous levels are changed.
Each of the changed variables affects $\poly(D)$ flaws in $\ell$.
The update sequence for level $\ell$ is generated by: (1) updating the affected flaws and (2) updating the flaws for all edges in $P_i$, as described in the update for level 0.
The number of updates is therefore $U_{\ell}(i) = C_{\ell-1}(i) \cdot \poly(D) + \poly(D)$.
The number of resamplings is denoted by $R_{\ell}(i)$, which depends on $U_{\ell}(i)$ and the internal randomness used in level $\ell$.
Since each resample changes $\poly(D)$ variables, the number of changed variables is $C_\ell(i) = R_{\ell}(i) \cdot \poly(D)$.

The update of the last level is a special case:
It does not require any flaw update because the waiting times can be arbitrary.

\paragraph{Note on adversary model.}
Our analysis works in two steps:
First, we obtain a guarantee for oblivious adversaries.
The second step extends the guarantee to clairvoyant adversaries.
To this end, we need a property of an oblivious adversary:
By definition, an oblivious adversary chooses the update sequence beforehand, without looking at the output of the algorithm.
We can, in fact, allow an oblivious adversary to choose the update sequence from any probability distribution independent of the output of the algorithm.
By a standard argument, the guarantee in \cref{thm:application} (A1) against a fixed sequence can be extended to against such a probability distribution.
Using this property, we can show that when path updates are chosen obliviously, then so are flaw updates for all levels.

Suppose that the path updates are chosen by an oblivious adversary.
Since the path updates are mapped to level-0 flaw updates deterministically, flaw updates for level 0 are also oblivious.
Consider a level $\ell > 0$.
The flaw updates of level $\ell$ depend both on the path updates and the internal randomness for levels $0, 1, \dots, \ell - 1$, but not level $\ell$.
Hence, the update sequence for level $\ell$ distribution is independent of level-$\ell$ output.
Hence, if the path updates are chosen obliviously, then each level also receives an oblivious sequence of flaw updates.

\paragraph{Total number of resamplings.}
Consider the oblivious setting and let $\delta = \frac{1}{2^{q+D}}  \cdot  (\frac{1}{m^D})^q \in (0, 1)$ be a parameter controlling the success probability.
Let $\hat{U}_\ell$ denote $\sum_{i\in[q]}{U}_\ell(i)$.
Define $\hat{R}_{\ell}$ and $\hat{C}_\ell$ similarly.
Fix a level $\ell$.
By \cref{lem:scheduling-convergence}, we can satisfy \cref{cond:general} with $\psi_F = \Theta(D_\ell^{-3})$.
By \cref{thm:application} (A1), with probability $1 - \delta$, the total number of resampling steps in level $\ell$ is $\hat{R}_\ell = O(\log |\Omega_\ell| + \hat{U}_\ell + \log \delta^{-1}) = \tO(\hat{U}_\ell + qD)$, where the second equality holds because $|\Omega_\ell| = O(D^{qD})$ and $\log \delta^{-1} = \tO(q + qD)$.
Let $B = \log(m + q)$ be the log factors hidden in $\tO(\cdot)$.
In summary, we upper bound the number of updates, resamplings, and variable changes as follows:
\begin{itemize}
    \item $\hat{U}_0 \leq q \cdot \poly(D)$, and $\hat{U}_\ell  \leq \hat{C}_{\ell-1} \cdot \poly(D)$ for $\ell > 0$.
    \item $\hat{R}_\ell \leq (\hat{U}_\ell + qD) \cdot \poly(B) \leq \hat{U}_\ell \cdot \poly(D, B)$ for $\ell \geq 0$.
    \item $\hat{C}_\ell \leq \hat{R}_\ell \cdot \poly(D)$ for $\ell \geq 0$.
\end{itemize}
All above inequalities hold with probability at least $1 - \delta$.
By a union bound, all of them hold with probability $1 - 3L \cdot \delta$.
By induction, $\hat{R}_\ell \leq q \cdot (BD)^{\Theta(\ell)}$, which holds with probability $1 - 3L \cdot \delta$ against an oblivious adversary.

The total number of path update sequences that an clairvoyant adversary can choose from is $\zeta = O((2m^D)^q)$, because each path can be specified by a series of $D$ edges.
Note that the probability guarantee of our algorithm against any fixed sequence is $1 - 3L \delta \geq 1 - \frac{1}{2^q \cdot \zeta}$.
By repeating the union bound argument in the proof of \cref{thm:application} (A2), we can obtain the following guarantee against a clairvoyant adversary.

\begin{lemma} \label{lem:scheduling-resampling}
With probability at least $1 - 2^{-q}$, the total number of resampling steps over all levels is $q \cdot (D \cdot \log(m + q))^{\Theta(\ell)}$ against a clairvoyant adversary.
\end{lemma}

\subsection{Implementation details}
It remains to specify the details and the implementation of $\algResample, \algCompare$, and $\algCheck$.
% Let $k$ denote the current number of paths in the problem instance.
We first specify our data structure setup.

\begin{itemize}
    \item \textbf{Path dissection:}
    We store a path as a sequence of edges.
    The blocks are stored as an array of size at most $\poly(D)$, where each block is represented by its first edge, last edge, and waiting time.
    Overall, the number of blocks is at most $\poly(D)$
    \item \textbf{Packet lists:} Each edge $e$ stores a BBST $T_e$ that contains all paths passing $e$.
    \item \textbf{Flawless assignments:}
    Each level $\ell$ maintains a set $\cF_\ell$ of flaws and a flawless assignment $\sigma_\ell$ for $\cF_\ell$.
    The update of $\cF_\ell$ is specified in \cref{sec:scheduling-overview}.
\end{itemize}

\noindent
Maintaining path dissections and packet lists can be done in $\tO(\poly(D))$ time per update.
The flawless assignments are maintained by our LLL framework as follows.
Consider a level $\ell$ and assume that we have access to flawless assignments $\sigma_0, \sigma_1, \dots, \sigma_{\ell-1}$ for previous levels.
We now state the implementations of \algResample, \algCheck, \algCompare for level $\ell$.

\begin{itemize}
    \item We need a subroutine $\textsc{Calculate}(\ell', i, e, t)$: Given a packet $i$, edge $e$, and time $t$, find the probability that $i$ passes $e$ at time $t$ on the condition that $\alp_x = \sigma_x$ for $x \in [0, \ell]$.
    The procedure also assumes access to waiting time assignments $\sigma_0, \sigma_1, \dots, \sigma_{\ell'}$ for the first $\ell'$ levels.
    To do this, we find the $O(\log \log D)$ blocks containing $e$ in the path dissection of $i$.
    Then, we enumerate all possible assignments of waiting times for all these blocks of level higher than $\ell'$.
    For each such assignment, we check in $O(D)$ time whether it makes $i$ pass $e$ at $t$.
    Let $W'$ be the number of enumerated assignments, and $x$ be the number of assignments that make $i$ pass $e$ at $t$.
    The desired probability is $x / W'$.    
    Note that $W' = \prod_{\ell'' > \ell'} W_{\ell'} = O(D^2)$.
    Hence, this procedure runs in $\poly(D)$ time.
    
    \item Given two flaws $F_1, F_2$, \algCompare compares them in $O(1)$ time by comparing their edge ID.
    (Recall that the permutation $\pi$ can be arbitrary.)
    
    \item Given a flaw $F' = F_\ell(e, t)$ and a state $\sigma_\ell \in F_\ell(e, t)$, \algCheck first invoks $\textsc{Calculate}(\ell, i, e, t)$ on all $O(D)$ packets that pass $e$.
    This step computes their respective probabilities of passing $e$ at $t$.
    The expected load $\E[X_\ell(e, t)]$ is the sum of those probabilities.
    The procedure then checks whether $\E[X_\ell(e, t)] > \lambda_\ell$.
    In total, $\poly(D)$ time is spent.
    
    \item Given a flaw $F' = F_\ell(e, t)$ and a state $\sigma_\ell \in F_\ell(e, t)$, \algResample simulates the resampling as follows.
    First, extract all $O(D)$ packets passing $e$ from its packet list.
    Next, find the $O(D_\ell)$ packets that still have a non-zero chance of passing $e$ at $t$.
    This can be done by calling $\textsc{Calculate}(\ell - 1, i, e, t)$ on each packet.
    For each of these packets $i$, we resample $\sigma_\ell(i, b_i)$, where $b_i$ is the level-$\ell$ block containing $e$ in the path dissection of $i$.
    In total, the resampling takes $\poly(D)$ time.

    The list of introduced flaws is found as follows.
    The resampling changes $O(D_\ell)$ blocks, each of which contains at most $D_{ell}$ edges.
    For each of these edges $e'$ and for all time $t' \in [0, T]$, we invoke $\algCheck(F_\ell(e', t'), \sigma_\ell)$ to check if the flaw is introduced.
    In total, this takes $\poly(D)$ time.
\end{itemize}

By our analysis, all key procedures can be implemented in $\poly(D)$ time.
Combining with the number of resampling steps $q \cdot (D \cdot \log(m + q)^{\Theta(\ell)})$, we know that the total running time for all levels is $q \cdot (D \cdot \log(m + q))^{\Theta(\log \log D)}$.
Hence, we have the following.

\begin{theorem}
There is an algorithm for the dynamic scheduling problem running in $(D \cdot \log(m + q))^{\Theta(\log \log D)}$ amortized time per update with probability $1- 2^{-q}$ against a clairvoyant adversary.
\end{theorem}

\section{\texorpdfstring{$(1+\eps) \Delta$}{ } Edge coloring}
\label{sec:edge-coloring}
Let $G = (V, E)$ be a graph with maximum degree $\Delta$, and $\eps > 0$ be a parameter.
The task of $(1+\eps) \Delta$-edge coloring is to assign colors to the edges -- each color chosen from $\{1, 2, \ldots, (1+\eps)\Delta\}$ -- so that no two incident edges receive the same color.
This section discusses the application of our LLL result to the fully dynamic version of this problem.
Our dynamic implementation for $(1+\eps) \Delta$-edge coloring builds on the distributed algorithm by Chang, He, Li, Pettie, and Uitto~\cite{chang2018complexity}.
First, we recall their approach, primarily focusing on their arXiv version~\cite[Section 3]{chang2017complexityarxiv}. 
Then, we discuss how to solve a dynamic version of that problem.

\subsection{Review of \texorpdfstring{\cite{chang2017complexityarxiv}}{CHL+}'s approach}
At a very high level, \cite{chang2017complexityarxiv}'s approach is composed of two steps: a R\"odl Nibble type step, and a simple finishing one (\cref{sec:edge-coloring-finishing-step}).\footnote{The finishing step we present differs than the one presented in \cite{chang2017complexityarxiv}.}
The nibble step consists of several $t_\eps = O(\log 1/\eps)$ phases.
The goal of these phases is to use $(1+\xi) \Delta$ many colors, for a $\xi < \eps$, to properly color a subset of the edges such that at most $\frac{1}{5} \cdot (\eps - \xi) \Delta$ edges incident to an edge remain uncolored.
Then, in the finishing step, the uncolored edges are colored by the $(\eps - \xi) \Delta$ unused colors in an (almost) trivial way.
The main result of \cite{chang2017complexityarxiv,chang2018complexity} is the following.
\begin{theorem}[Rephrased Theorem~3 of \cite{chang2017complexityarxiv}]
\label{theorem:chang-coloring}
    Let $\eps = \omega\rb{\frac{\log^{2.5} \Delta}{\sqrt{\Delta}}}$ be a function of $\Delta$. 
    If $\Delta > \Delta_\eps$ is sufficiently large, then each of the $t_\eps$ coloring phases can be performed by invoking an LLL instance with parameters: $m$ random variables; $d = \poly(\Delta)$ maximum degree in the flaw-dependence graph; and, probability at most $p = \exp\rb{-\eps^2 \Delta / \log^{4+o(1)}\Delta}$ that a given flaw is caused by a random assignment.
    %\stodo{Do we even need $\Delta > \Delta_\eps$ condition, as it is already imposed by the condition on $\eps$?}
\end{theorem}
This result suffices to obtain $\Delta^{O(t_\eps)}$ amortized update time. 
We now discuss details, and begin by recalling some details from the approach developed in \cite{chang2017complexityarxiv}.
% \begin{corollary}
% \label{corollary:edge-coloring-slower}
%     Let $\eps = \omega\rb{\frac{\log^{2.5} \Delta}{\sqrt{\Delta}}}$ be a function of $\Delta$. 
%     If $\Delta > \Delta_\eps$ is sufficiently large, then the $t_\eps$ coloring phases can be dynamically maintained in the $\Delta^{O(\log 1/\eps)}$ amortized update time.
% \end{corollary}
% Our goal is to obtain $\poly(\Delta)$ update time.
% To explain how to go beyond \cref{corollary:edge-coloring-slower}, we first recall more details from the approach developed in \cite{chang2017complexityarxiv}.

Each phase performs partial coloring of the so-far uncolored edges.
Let $G_i$ denote the subgraph of uncolored edges of $G$ at the beginning of phase $i$.
Let $\cC = \{1, \ldots, (1+\xi) \Delta \}$ be the colors used throughout the $t_\eps$ phases.
This coloring is maintained to balance several quantities:
\begin{itemize}
    \item $\Phi_i(e) \subseteq \cC$ is the palette of colors available to $e$ at the beginning of iteration $i$;
    \item $\deg_i(v)$ denotes the number of edges incident to $v$ in $G_i$;
    \item $\deg_{i, c}(v)$ denotes the number of edges incident to $v$ in $G_i$ that have color $c$ in their palettes.
\end{itemize}
For the base case, we let $G_1 = G$ and $\Phi_1(e) = \cC$ for each edge $e$.
Then, in graph $G_i$ is maintained the following invariant:
\begin{align}
    \deg_i(v) & \le d_i, \label{eq:deg-i-ub} \\
    \deg_{i, c}(v) & \le t_i, \label{eq:deg-i-c-ub} \\
    |\Phi_i(e)| & \ge p_i. \label{eq:phi-i-lb}
\end{align}
The precise values of $d_i$, $t_i$, and $p_i$ are not important for our purposes; a reader is referred to \cite[Section 3.1]{chang2017complexityarxiv} for the exact values.
%where $d_1 \eqdef \Delta$, $t_1 \eqdef \Delta$, $p_1 \eqdef (1+\xi)\Delta$
For the coloring, the authors use algorithm \algOneShot, described in \cref{alg:one-shot}.
\begin{algorithm}[htbp]
    \Fn{\algOneShot}{
        \KwIn{Graph $G_i$}
        Each edge $e$ select a color $\Color(e) \in \Phi_i(e)$ uniformly at random.

        An edge $e$ successfully colors itself in $\Color(e)$ if no edge incident to $e$ also selects $\Color(e)$.
    }
    
    \caption{Description of algorithm \algOneShot.\label{alg:one-shot}}
\end{algorithm}

Based on \cref{eq:deg-i-ub}, \cref{eq:deg-i-c-ub}, and \cref{eq:phi-i-lb}, we naturally define the following flaws:
\begin{itemize}
    \item $F_{j}^{\deg}(u)$ occurs when \eqref{eq:deg-i-ub} is \textbf{not} satisfied for $i = j$
    %, but \cref{eq:deg-i-ub,eq:deg-i-c-ub,eq:phi-i-lb} are satisfied for $i < j$ for each vertex and edge in the graph.
    \item $F_{j, c}^{\deg}(u)$ occurs when \eqref{eq:deg-i-c-ub} is \textbf{not} satisfied for $i = j$
    %, but \cref{eq:deg-i-ub,eq:deg-i-c-ub,eq:phi-i-lb} are satisfied for $i < j$ for each vertex and edge in the graph.
    \item $F_{j}^{\Phi}(e)$ occurs when \eqref{eq:phi-i-lb} is \textbf{not} satisfied for $i = j$
    %, but \cref{eq:deg-i-ub,eq:deg-i-c-ub,eq:phi-i-lb} are satisfied for $i < j$ for each vertex and edge in the graph.
\end{itemize}
These flaws capture the $t_\eps$ nibble phases we discussed above.
For the implementation details, we need to discuss the radius of the neighborhood affecting the outcome of \cref{eq:deg-i-ub,eq:deg-i-c-ub,eq:phi-i-lb} in \textbf{one random color assignment}.
Recall that $G_i$ represents the graph induced by the uncolored edges of $G$ at the beginning of phase $i$.
\begin{itemize}
    \item \cref{eq:deg-i-ub}, i.e., the number of uncolored edges incident to a vertex $v$, after a random coloring depends only on the $2$-hop neighborhood of $v$.
    \item \cref{eq:deg-i-c-ub}, i.e., the number of edges incident to a vertex $v$ that have color $c$ in their palettes, depends only on the $3$-hop neighborhood of $v$.
    To see why it holds, consider an edge $e$ incident to $v$.
    Then, whether $e$ keeps $c$ in its palette or not depends on whether any edge $f$ incident to $e$ has been successfully colored by $c$ or not.
    However, a successful coloring of $f$ depends on the edges incident to it. All such edges have both endpoints within the $3$-hop neighborhood of $v$.
    \item \cref{eq:phi-i-lb}, i.e., the palette of colors available to $e = \{u, v\}$, after a random coloring depends on the colors of the edges in the $2$-hop neighborhood of $u$ and $v$.
\end{itemize}

\subsection{Finishing step}
\label{sec:edge-coloring-finishing-step}
Let $\eps'$ be the input parameter for which our goal is to maintain a $(1+\eps') \Delta$-edge coloring. 
Invoke the $t_\eps$ phases described above for $\eps = \eps'/12$.
Let $\Gfinal$ be the graph induced by the uncolored edges after those $t_\eps$ phases.
The process guarantees that the maximum degree of $\Gfinal$ is less than $\Deltafinal \eqdef \eps' \Delta / 12$.
Moreover, at least $(1+\eps') \Delta - (1+\eps)\Delta = 11 \cdot \eps' \Delta / 12 = 11 \cdot \Deltafinal$ colors remain unused after the $t_\eps$ phases, which we use to color $\Gfinal$.
Let $\cCfinal = \{(1+\eps)\Delta + 1, (1+\eps)\Delta + 2, \ldots, (1+\eps')\Delta \}$ be a set of those remaining colors.

To color $\Gfinal$, we invoke LLL with the following setup:
\begin{itemize}
    \item Each edge samples a color from $\cCfinal$ uniformly at random.
    \item We define a bad even $B_{e, f}$ if two incident edges $e$ and $f$ sampled the same color.
\end{itemize}
We now prove why it is indeed an LLL instance.
First, it holds that,
\[
    \prob{\text{$B_{e, f}$ occurs}} = \frac{1}{|\cCfinal|} = \frac{1}{11 \cdot \Deltafinal}.
\]
Second, a bad event $B_{e, f}$ shares variables with the bad events containing $e$ or $f$. 
Hence, there are less than $4 \Deltafinal$ -- including $B_{e, f}$ itself -- such bad events.
Since
\[
    \prob{\text{$B_{e, f}$ occurs}} \cdot 4 \Deltafinal \cdot e < 1,
\]
this setup can indeed be solved with a single LLL instance.

\subsection{Implementation details}
We instantiate $t_\eps + 1$ many LLL instances: the first $t_\eps$ of them represent the nibble phases, while the last one implements the finishing step. 
We refer to those instances by $L_1, \ldots, L_{t_\eps + 1}$.
A graph update is directly passed to $L_1$, while $L_k$ defines updates passed to $L_{k+1}$.
Each instance $L_k$ maintains a vector of edge colors $c_k$, and also the palette of colors available to each edge/vertex.

We now describe how are updates processed:
\begin{itemize}
    \item In our process, a single update to $L_1$, triggers many updates to $L_2$. Generally, a single update to $L_i$ triggers many updates to $L_{i+1}$. This chain of reactions is all caused by a single (adversarial) update to $L_1$.
    The updates are processed in the order of $i$: only after all updates to $L_i$ are processed, then $L_{i+1}$ starts processing its updates.
    \item Instance $L_i$ maintains graph $G_i$; as a reminder, $G_i$ is the graph consisting of unsuccessfully colored edges at the beginning of phase $i$.
    In particular, when an edge $e$ becomes successfully colored in $L_i$, it triggers removal of $e$ from $L_j$, for $j > i$.
    Likewise, if $e$ was successfully colored in $L_i$ before, but have become unsuccessfully colored due to flaw resampling in $L_i$, then $L_i$ triggers addition of $e$ to $L_{i + 1}$. 
    %If $e$ remains unsuccessfully colored in $L_{i+1}$, $L_{i+1}$ will trigger addition of $e$ to $L_{i+2}$, and so on.

    \item When an edge $e$ changes its color at $L_i$, $e$ tests whether any edge incident to it has become successfully or unsuccessfully colored.

    \item When an edge $e$ in $L_i$ becomes successfully colored, then $\Phi_{i+1}(h)$ for all edges $h$ incident to $e$ gets updated.
    If any of those edges $h$ was (unsuccessfully or successfully) colored in a color that got removed from $\Phi_{i+1}(h)$, then we trigger two updates to $L_{i+1}$: remove $h$; add $h$.

    \item Similarly to the previous step, when an edge $e$ in $L_i$ becomes unsuccessfully colored, it potentially updates $\Phi_{i + 1}(h)$ for each edge $h$ incident to $e$.
    In particular, $\Phi_{i + 1}(h)$ gets augmented by $c_i(e)$ if no other edge incident to $h$ has been successfully colored by $c_i(e)$ at $L_j$ for $j \le i$.
    Since this does not reduce $\Phi_{i + 1}(h)$, no further update is needed.
    
    \item When edge $e$ is added into $L_i$, it selects a random color from its palette.
    
    \item Let $e$ be a removed edge.
    Let $L_i$ be the instance at which $e$ is unsuccessfully colored.
    Removing $e$ from $L_i$ might result in at most two edges -- say $h$ and $f$ -- incident to $e$ to become successfully colored. Edges $h$ and $f$ are those which had the same color as $e$ in $L_i$, and $e$ was the only edge incident to $h$ and/or $f$ with color $c_i(e)$.
    Hence, our update also tests whether any such $h$ and $f$ exist. If they do, then corresponding update is passed to $L_{i+1}$ as already described.
    
    Let $L_k$ be the instance at which $e$ is successfully colored.
    Removing $e$ triggers potential augmentation of $\Phi_{k+1}(h)$ by $c_k(e)$ of the edges $h$ incident to $e$.

    \item If an edge is removed from $L_i$, after $L_i$ process all its updates, that edge removal is also triggered at $L_{i+1}$.
    Similarly, an addition of an edge $e$ to $L_i$ is also propagated to $L_{i+1}$, until $e$ becomes successfully colored.

    \item After each edge addition, edge removal, an edge-color change, or the palette change, we test whether the update creates a new flaw.
\end{itemize}

Each of the updates and tests described above can be performed in $\poly(\Delta)$ time. 
Moreover, since each flaw depends on at most $3$-hop neighborhood, an edge $e$ is part of $O(\Delta^3)$ many flaws in a given $L_i$.

Next, we want to analyze the running time for processing updates and fixing flaws by all the LLL instances.
We will leverage \cref{thm:gen-dynLLL2} in the analysis.
Consider an instance $L_i$, and assume it receives $q_i$ updates and resamples $s_i$ many flaws.
Based on our discussion, it spends $r_i = (q_i + s_i) \poly(\Delta)$ time processing those updates and flaws.
In the process, it also triggers updates to $L_{i+1}$. 
However, every triggered update is a consequence of the work spent for $L_i$. 
Therefore, $L_{i+1}$ receives $O(r_i)$ updates from $L_i$.

Now, using \cref{thm:gen-dynLLL2}, we would like to upper-bound $s_i$ in terms of $q_i$. 
To that end, we invoke \cref{thm:gen-dynLLL2} with $k = q_i \log q_i$, $\epsilon$ is a constant, $\delta^{-1} = \rb{n^2}^{q_i}$. 
The parameter $\psi$ for our application is $O(1/\Delta)$. 
% It is not hard to verify that letting $\psi = \Theta(1)$ suffices.
That implies that
\[
    s_i = O\rb{q_i \log (\Delta q_i) + q_i \log n},
\]
and hence
\[
    r_i = O\rb{q_i \log (\Delta q_i) + q_i \log n} \cdot \poly(\Delta).
\]
Given that $r_1 = q$, we derive that
\begin{align*}
    r_{t_\eps} &= \tO\rb{q \cdot \rb{\log (\Delta q) + \log n}^{1 + \log 1/\eps} \cdot \Delta^{O(\log 1/\eps)}} \\
    &= \tO\rb{q \cdot \rb{\log q + \log n}^{1 + \log 1/\eps} \cdot \Delta^{O(\log 1/\eps)}}
\end{align*}
This proves \cref{thm:edge-coloring}.

\section{Conclusion and future work} \label{sec:future-work}
In this paper, we introduce the study of \LLL\ in the dynamic setting and obtain resampling algorithms whose costs match the static setting up to polylogarithmic factors against even an adaptive adversary.
Several directions remain open.
A particularly interesting question concerns the efficiency of composing multiple LLL instances.
Such compositions are essential in at least two of our applications -- the LMR-scheduling scheme and the $(1+\eps)\Delta$-edge-coloring algorithm.
In these settings, our query complexities currently exhibit exponential dependence on the number of composed LLLs.
Is this exponential blow-up inherent, or can it be avoided?
More generally, is it possible to design a black-box composition of multiple LLL instances whose resampling complexity scales only polynomially with the number of compositions, at least for some natural classes of problems?
Alternatively, can one prove lower bounds showing that such efficient composition is impossible?
\bibliographystyle{alpha}
\bibliography{ref}

\appendix
\section{Proof of \texorpdfstring{\cref{lem:prod}}{Lemma }} \label{sec:missing-proof}
Without loss of generality, assume that $\vbl(F_j) = \{v_1, v_2, \dots, v_k\}$ for some $k \leq n$.
Let $v_i(\tau)$ be the value of $v_i$ in the state $\tau$.
Consider a fixed state $\tau$.
Then $\InSet(\tau)$ is the set of states $\sigma$ such that:
\begin{enumerate}
    \item $\sigma$ must have the flaw $F_j$. That is, $(v_1(\sigma), v_2(\sigma), \dots, v_k(\sigma))$ must violate $F_j$.
    \item The transition $\sigma \rightarrow \tau$ is possible by resampling $F_j$. That is, $v_i(\sigma) = v_i(\tau)$ for all $i > k$.
\end{enumerate}
Therefore, $|\InSet(\tau)| = \vio(F_j)$.
For each $\sigma \in \InSet(\tau)$, the transition probability is $\rho_j(\sigma, \tau) = \frac{1}{\prod_{i=1}^k c_i}$, because the resampling must set $v_i$ as $v_i(\tau)$ for $i \leq k$.
Therefore, we have

\begin{align*}
    \gamma_j &= \max_{\tau \in \Omega} \left\{ \sum_{\sigma \in \InSet(\tau)} \frac{\mu(\sigma)}{\mu(\tau)} \cdot \rho_i(\sigma, \tau) \right\} \\
    &= \max_{\tau \in \Omega} \left\{ |\InSet(\tau)| \cdot \frac{1}{\prod_{i=1}^k c_i} \right\} \\
    &= \max_{\tau \in \Omega} \left\{  \frac{\vio(F_j)}{\prod_{i=1}^k c_i} \right\} = \frac{\vio(F_j)}{\prod_{i=1}^k c_i} = \Pr[F_j].
\end{align*}
This completes the proof.

\section{Proof of \texorpdfstring{\cref{lem:tri-convergence}}{Lemma }} \label{sec:app-tri-convergence}
We first show that the flaws $F = \bigcup_{v \in V} \{B_v, Z_v\}$ satisfies \cref{cond:general} for any \emph{static} graph $G = (V,E)$.
To lighten notation, in the following we write $\gamma^S(f)$ instead of $\gamma^S_f$.
For each vertex $v$, let $\algRecolor(v, \sigma)$ denote the resampling procedure for $B_v$ and $Z_v$.
Recall that $D_v$ is the set of vertices of distance at most 2 from $v$.
Let $S_v = \{B_u\}_{u \in D_v} \cup \{Z_u\}_{u\in D_v}$.
Recall that the permutation $\pi$ prioritizes all $B$-flaws over $Z$-flaws.
The number of colors is $\ell = \frac{6\Delta}{\ln \Delta}$, and $L = \Delta^{0.7}$.
% \cite{achlioptas2019beyond} showed that the charges for $B_v$ are bounded.

% \begin{lemma}[\cite{achlioptas2019beyond}] \label{lem:Bv-charge}
% For every vertex $v \in V$, the following bounds on charges hold:
%     \emph{(a)} $\gamma^S(B_v) = 0$ if $S \nsubseteq D_v$.
%     \emph{(b)} $\max_{S \subseteq F, v \in V} \gamma^S(B_v) \leq 2e^{-\frac{L}{6}}$.
% \end{lemma}
% \btodo{Cannot cite because our $L$ is different.}

% We now upper bound the charges for $Z_v$.

\begin{lemma} \label{lem:tri-charge}
For every vertex $v \in V$, the following hold:
    \emph{(a)} $\Gamma_{B_v}$ and $\Gamma_{Zv}$ are subsets of $S_v$.
    \emph{(b)} $\max_{S \subseteq F, v \in V} \gamma^S(Z_v) \leq \Delta^{-3}$, and
    \emph{(c) $\max_{S \subseteq F, v \in V} \gamma^S(B_v) \leq \Delta^{-3}$, where the charges are computed with respect to the uniform distribution.}
\end{lemma}
\begin{proof}
Consider a fixed vertex $v$.

\vspace{6pt}
\noindent \emph{Proof of part (a).}
Recall that the resampling procedure for $B_v$ or $Z_v$ can only change the colors of vertices in $N_v$.
Hence, it never introduces any flaw $B_u$ or $Z_u$ with $u \notin S_v$.
This proves part (a).
% This shows that $\rho^S_{B_v}(\sigma, \tau) = \rho^S_{Z_v}(\sigma, \tau) = 0$ if $S \nsubseteq S_v$.
% Therefore, the charges $\gamma^S(B_v) = \gamma^S(Z_v) = 0$ and part (a) holds.

\vspace{6pt}
\noindent \emph{Proof of part (b).}
It suffices to upper bound $\gamma(Z_v) = \gamma^{\emptyset}(Z_v) = \max_{\tau \in \Omega} \sum_{\sigma \in \InSet_{Z_v}(\tau)} \rho_{Z_v}(\sigma, \tau)$.
Fix a state $\tau$.
Observe that each state $\sigma \in \InSet_{Z_v}(\tau)$ satisfies the following:
\begin{itemize}
    \item (P1) $\sigma \in Z_v$ and $\sigma$ does not contain any $B$-flaw because $Z_v$ should be flaw with the lowest index in $\pi$. This shows that $\InSet_{Z_v}(\tau)$ is a subset of $Z_v \setminus \bigcup_{u \in V} B_u$. 
    \item (P2) By (P1), $|L_u(\sigma)| \geq L$ for all $u \in V$.
    \item (P3) The transition $\sigma \rightarrow \tau$ is possible by resampling $Z_v$ at $\sigma$.
    Since $\algRecolor(\sigma, v)$ only changes the colors of $N_v$, we know that $\sigma(u) = \sigma(v)$ for all $u \notin N_v$.
    Furthermore, since $N_v$ is an independent set, $L_u(\sigma) = L_v(\sigma)$ for all $u \in N_v$.
\end{itemize}
By (P1) and (P3), if any neighbor $u \in N_v$ has $|L_u(\tau)| < L$, then $\InSet_{Z_v}(\tau) = \emptyset$ and $\sum_{\sigma \in \InSet_{Z_v}(\tau)} \rho_{Z_v}(\sigma, \tau) = 0$.
Hence, we assume that $|L_u(\tau)| \geq L$ for all $u \in N_v$.
(P3) further implies that $\Pr[\algRecolor(\sigma, v) = \tau]$, the probability that recoloring $v$ in $\sigma$ results in $\tau$, is equal to $\Pr[\algRecolor(\tau, v) = \sigma]$.
Therefore, we have
\begin{align*}
\sum_{\sigma \in \InSet_{Z_v}(\tau)} \rho_{Z_v}(\sigma, \tau)
&= \sum_{\sigma \in \InSet_{Z_v}(\tau)} \Pr[\algRecolor(\sigma, v) \text{ outputs } \tau] \\
&= \sum_{\sigma \in \InSet_{Z_v}(\tau)} \Pr[\algRecolor(\tau, v) \text{ outputs } \sigma] \\
&\leq \Pr[\algRecolor(\tau, v) \text{ outputs a state $\sigma \in Z_v$}]. & (\text{by (P1)})
\end{align*}
In \cite[Lemma 15]{molloy2019list}, it has been shown that the latter probability is at most $\frac{\binom{\Delta}{L-1}}{L^L}$, assuming that $|L_u(\tau)| \geq L$ for all $u \in N_v$.
Notice that
\[
    \frac{\binom{\Delta}{L - 1}}{L^L} \leq \left( \frac{e\Delta}{L-1} \right)^L / L^L = \left( \frac{e\Delta}{L(L-1)} \right)^L \leq \left( \frac{2e}{\Delta^{0.4}} \right)^{\Delta^{0.7}} \leq \Delta^{-3},
\]
where the first inequality uses $\binom{a}{b} \leq (\frac{ea}{b})^b$ \cite{cormen2009introduction}, and the last two inequalities hold for $\Delta \geq 100$.
This completes the proof of part (b).

% We complete part (b) by showing that the latter probability is at most $\Delta^{-3}$.
% % To upper bound the latter probability, consider the following setup:
% % We are in the state $\tau$, where every neighbor $u \in N_v$ has $|L_u(\tau)| \geq L$.
% % $\algRecolor(\tau, v)$ is invoked to recolor all neighbors of $v$, and we want to upper bound the probability that $v$ has at least $L$ blank neighbors.

% Given any $L$ neighbors of $v$, the probability that all of them are blank is $1/L^L$, because $|L_u(\tau)| \geq L$ for all $u \in N_v$.
% By union bounding over all $\binom{\Delta}{L}$ possible subsets of $L$ neighbors, the probability that $v$ has $L$ blank neighbors is at most 
% \[
%     \frac{\binom{\Delta}{L}}{L^L} \leq \left( \frac{e\Delta}{L} \right)^L / L^L = \left( \frac{e\Delta}{L^2} \right)^L = \left( \frac{e}{\Delta^{0.4}} \right)^{\Delta^{0.7}} \leq \Delta^{-3},
% \]
% where the first inequality uses $\binom{a}{b} \leq (\frac{ea}{b})^b$ \cite{cormen2009introduction}, and the second inequality holds for $\Delta \geq 100$.
% This completes the proof of part (b).

\vspace{6pt}
\noindent \emph{Proof of part (c).}
Similar to part (b), our goal is to upper bound $\max_{\tau \in \Omega}  \sum_{\sigma \in \InSet_{B_v}(\tau)} \rho_{B_v}(\sigma, \tau)$.
Fix a state $\tau$.
By the same argument, $\sum_{\sigma \in \InSet_{B_v}(\tau)} \rho_{B_v}(\sigma, \tau) \leq \Pr[\algRecolor(\tau, v) \text{ outputs a state in $B_v$}]$.
To bound the latter probability, consider the following setup:
We start from an arbitrary state $\tau$ and perform $\algRecolor(\tau, v)$ on $v$.
Let $\sigma$ be the resulting state and let $X$ denote $|L_v(\sigma)|$.
We want to show that $\Pr[X < L] \leq \Delta^{-3}$.

Recall that $q$ is the number of colors.
For any $q$, \cite[Lemma 7]{molloy2019list} gave the following concentration bound for $X$:
\begin{itemize}
    \item (Q1) $\E[X] \geq qe^{-\Delta/q}$, and
    \item (Q2) $\Pr[X < \frac{1}{2} \E[X]] < \exp(-\frac{1}{8} \E[X])$.
\end{itemize}
\noindent By our choice of $q$ and $L$, $\E[X] \geq \frac{5\Delta}{\ln \Delta} \cdot e^{-\ln \Delta / 6} = \frac{5\Delta^{5/6}}{\ln \Delta} \geq 2L$ for $\Delta \geq 10$.
By (Q2), $\Pr[X < L] < \exp(-\frac{L}{4}) \leq \Delta^{-3}$ for $\Delta \geq 100$.
This proves part (c).
\end{proof}

We now examine \cref{cond:general} using the bounds in \cref{lem:tri-charge}.
Let $\gamma = \Delta^{-3}$ be an upper bound on all charges.
Let $\psi_F = 2\gamma$ for a flaw $F$, and thus we have

\begin{align*}
    \frac{1}{\psi_F} \sum_{S \subseteq \Gamma_F} \gamma_F^S \prod_{j \in S} \psi_j
    \leq &~\frac{1}{2\gamma} \sum_{S \subseteq \Gamma_F} \gamma \prod_{j \in S} \psi_j \\ 
    = &~\frac{1}{2} \prod_{j \in \Gamma_F} (1 + \psi_j) \\
    % \leq &\frac{1}{2} \prod_{j \in \Gamma_F} \exp(2\gamma) \\
    \leq &~\frac{1}{2} \exp(\sum_{j \in \Gamma_F} 2\gamma) & (1+x \leq e^x \text{ and } \psi_j = 2\gamma)\\
    \leq &~\frac{1}{2} \exp((2\Delta^2+2) \cdot 2\Delta^{-3}) & (|\Gamma_F| \leq 2(\Delta^2+1) \text{ by (P1)})
    % \leq &~\frac{1}{2} \exp(4\Delta^{-1}+4\Delta^{-2})
\end{align*}

\noindent 
The last term is at most $3/4$ for $\Delta \geq 100$.
Thus, \cref{cond:general} is satisfied with slack parameter $\eps = 1/4$ and $\psi_F = 2\Delta^{-3}$ for each flaw $F$.

\paragraph{Dynamic setting.}
We now show that \cref{cond:general} is satisfied at any point of our algorithm, even when the current set of flaws does not correspond to any static graph.

In our algorithm, each edge update $(u, v)$ corresponds to eight flaw updates:
We first remove the old version of $B_u, Z_u, B_v,$ and $Z_v$, and then add the new version of them.
Let $\cF_1$ and $\cF_2$ be, respectively, the sets of flaws before and after the eight flaw updates.
Since each of $\cF_1$ and $\cF_2$ corresponds to a static graph, they satisfy \cref{cond:general}.
The sets of flaws between the eight updates do not correspond to any static graph.
However, these sets are either subsets of $\cF_1$ or of $\cF_2$.
Hence, by \cref{obs:subset-convergence}, \cref{cond:general} is satisfied after each flaw update.

\end{document}